\pdfoutput=1
\RequirePackage{ifpdf}
\ifpdf 
\documentclass[pdftex]{sigma}
\else
\documentclass{sigma}
\fi

\numberwithin{equation}{section}

\usepackage{amstext,amscd,array}

\usepackage{mathrsfs}
\usepackage{pdfsync}
\usepackage{bbm}
\usepackage{bm}
\usepackage[arrow,matrix,curve]{xy}
\usepackage{bbding}
\usepackage{wasysym}

\DeclareMathOperator{\Sphere}{\mathbb{S}}

\def\nn{\nonumber}

\newcommand{\C}{\complex}
\newcommand{\Z}{\zed}

\newcommand{\F}{\mathbb{F}}

\newcommand{\mbf}[1]{{\boldsymbol {#1} }}
\def\ii{{{\rm i}}}
\def\dd{{\rm d}}

\def\Id{{\rm id}}

\def\A{{\sf A}}
\def\V{{\sf V}}

\def\T{{\sf T}}

\def\C{{\sf C}}

\newcommand{\unit}{\mathbbm{1}} 			

\def\mcF{{\mathcal F}}

\newcommand{\CCK}{\mathscr{K}}

\newcommand{\CK}{\mathcal{K}}
\newcommand{\CI}{\mathcal{I}}
\newcommand{\Dcal}{\mathcal{D}}
\newcommand{\Ccal}{\mathcal{C}}

\newcommand{\sfG}{\mathsf{G}}
\newcommand{\sfH}{\mathsf{H}}
\newcommand{\sfN}{\mathsf{N}}
\newcommand{\sfM}{\mathsf{M}}
\newcommand{\sfK}{\mathsf{K}}
\newcommand{\sfU}{\mathsf{U}}
\newcommand{\sfGL}{\mathsf{GL}}
\newcommand{\sfSL}{\mathsf{SL}}
\newcommand{\sfSU}{\mathsf{SU}}

\newcommand{\ttM}{\mathtt{M}}

\newcommand{\rmH}{{\rm H}}
\newcommand{\Pim}{{\mit\Pi}}
\newcommand{\End}{{\rm End}}
\newcommand{\Aut}{\operatorname{Aut}}
\newcommand{\Leb}{{\rm L}}

\newcommand{\Heis}{{\sf Heis}}
\newcommand{\ISO}{{\sf ISO}}
\newcommand{\SU}{{\sf SU}}

\newcommand{\rmB}{{\rm B}}
\newcommand{\lt}{{\rm lt}}
\newcommand{\rt}{{\rm rt}}

\newcommand{\mcoprod}{\mbox{$\coprod\limits_{x\in\R/\Z}$}}

\newcommand{\R}{\real}

\newcommand{\complex}{{\mathbb C}} 
\newcommand{\zed}{{\mathbb Z}} 
\newcommand{\real}{{\mathbb R}} 

\newcommand{\torus}{{\mathbb T}} 
\def\alg{{\mathcal A}}
\def\balg{{\mathcal B}}
\def\hil{{\mathcal H}}

\def\Mcal{{\mathcal M}}

\newif\ifold \oldtrue

\def\nn{\nonumber}

\newcommand{\Tr}[1]{{\rm Tr}\,#1}

\def\e{{\,\rm e}\,}

\newcommand{\proj}{{\rm pr}}

\makeatletter
\newdimen\normalarrayskip 
\newdimen\minarrayskip 
\normalarrayskip\baselineskip
\minarrayskip\jot
\newif\ifold \oldtrue 
\def\arraymode{\ifold\relax\else\displaystyle\fi} 
\def\@arrayskip{\ifold\baselineskip\z@\lineskip\z@
 \else
 \baselineskip\minarrayskip\lineskip2\minarrayskip\fi}
\def\@arrayclassz{\ifcase \@lastchclass \@acolampacol \or
\@ampacol \or \or \or \@addamp \or
 \@acolampacol \or \@firstampfalse \@acol \fi
\edef\@preamble{\@preamble
 \ifcase \@chnum
 \hfil$\relax\arraymode\@sharp$\hfil
 \or $\relax\arraymode\@sharp$\hfil
 \or \hfil$\relax\arraymode\@sharp$\fi}}
\def\@array[#1]#2{\setbox\@arstrutbox=\hbox{\vrule
 height\arraystretch \ht\strutbox
 depth\arraystretch \dp\strutbox
 width\z@}\@mkpream{#2}\edef\@preamble{\halign \noexpand\@halignto
\bgroup \tabskip\z@ \@arstrut \@preamble \tabskip\z@ \cr}%
\let\@startpbox\@@startpbox \let\@endpbox\@@endpbox
 \if #1t\vtop \else \if#1b\vbox \else \vcenter \fi\fi
 \bgroup \let\par\relax
 \let\@sharp##\let\protect\relax
 \@arrayskip\@preamble}
\makeatother

\newtheorem{Theorem}{Theorem}[section]

\newtheorem{Lemma}[Theorem]{Lemma}
\newtheorem{Proposition}[Theorem]{Proposition}
 { \theoremstyle{definition}
\newtheorem{Definition}[Theorem]{Definition}

\newtheorem{Example}[Theorem]{Example}
\newtheorem{Remark}[Theorem]{Remark} }

\begin{document}

\newcommand{\arXivNumber}{2006.10048}

\renewcommand{\thefootnote}{}

\renewcommand{\PaperNumber}{012}

\FirstPageHeading

\ShortArticleName{Topological T-Duality for Twisted Tori}

\ArticleName{Topological T-Duality for Twisted Tori\footnote{This paper is a~contribution to the Special Issue on Noncommutative Manifolds and their Symmetries in honour of~Giovanni Landi. The full collection is available at \href{https://www.emis.de/journals/SIGMA/Landi.html}{https://www.emis.de/journals/SIGMA/Landi.html}}}

\Author{Paolo ASCHIERI~$^{\dag^1\dag^2\dag^3}$ and Richard J.~SZABO~$^{\dag^1\dag^2\dag^4\dag^5\dag^6}$}

\AuthorNameForHeading{P.~Aschieri and R.J.~Szabo}

\Address{$^{\dag^1}$~Dipartimento di Scienze e Innovazione Tecnologica, Universit\`a del Piemonte Orientale, \\
\hphantom{$^{\dag^1}$}~Viale T.~Michel 11, 15121 Alessandria, Italy}
\EmailDD{\href{mailto:paolo.aschieri@uniupo.it}{paolo.aschieri@uniupo.it}}
\Address{$^{\dag^2}$~Arnold--Regge Centre, Via P.~Giuria 1, 10125 Torino, Italy}
\Address{$^{\dag^3}$~Istituto Nazionale di Fisica Nucleare, Torino, Via P.~Giuria 1, 10125 Torino, Italy}

\Address{$^{\dag^4}$~Department of Mathematics, Heriot-Watt University,\\
\hphantom{$^{\dag^4}$}~Colin Maclaurin Building, Riccarton, Edinburgh EH14 4AS, UK}
\EmailDD{\href{mailto:R.J.Szabo@hw.ac.uk}{R.J.Szabo@hw.ac.uk}}

\Address{$^{\dag^5}$~Maxwell Institute for Mathematical Sciences, Edinburgh, UK}
\Address{$^{\dag^6}$~Higgs Centre for Theoretical Physics, Edinburgh, UK}

\ArticleDates{Received June 30, 2020, in final form January 22, 2021; Published online February 05, 2021}

\Abstract{We apply the $C^*$-algebraic formalism of topological T-duality due to Mathai and Rosenberg to a broad class of topological spaces that include the torus bundles appearing in string theory compactifications with duality twists, such as nilmanifolds, as well as many other examples. We develop a simple procedure in this setting for constructing the T-duals starting from a commutative $C^*$-algebra with an action of ${\mathbb R}^n$. We treat the general class of almost abelian solvmanifolds in arbitrary dimension in detail, where we provide necessary and sufficient criteria for the existence of classical T-duals in terms of purely group theoretic data, and compute them explicitly as continuous-trace algebras with non-trivial Dixmier--Douady classes. We prove that any such solvmanifold has a topological T-dual given by a $C^*$-algebra bundle of noncommutative tori, which we also compute explicitly. The monodromy of the original torus bundle becomes a Morita equivalence among the fiber algebras, so that these $C^*$-algebras rigorously describe the T-folds from non-geometric string theory.}


\Classification{46L55; 81T30; 16D90}

{\rightline{\it Dedicated to Giovanni Landi~~~~~~~~\,}}
{\rightline{\it on the occasion of his 60th birthday}}

{\small
\tableofcontents
}
\renewcommand{\thefootnote}{\arabic{footnote}}
\setcounter{footnote}{0}

\section{Introduction}\label{sec:intro}

\subsection{Background}

T-duality is a symmetry of string theory which relates distinct
spaces that describe the same physics. It has presented a challenge to
mathematics in finding a rigorous framework in which these
`equivalences' of spaces is manifest. It was realized early on that
noncommutative geometry provides such a
framework~\cite{Frohlich1993,Lizzi1997}, at least in the
simplest cases of tori endowed with a trivial
gerbe, where subsequently it was shown that T-duality is realised as
Morita equivalence of noncommutative
tori~\cite{Brace:1998ku,Landi1998,Pioline1999,Schwarz1998,Seiberg1999}.

T-duality of spaces which are compactified on tori, or more generally
torus bundles, can be explained topologically
in terms of correspondence spaces which implement a smooth analog of the
Fourier--Mukai transform~\cite{Hori1999}. In the correspondence space
picture, T-duality transformations are realised as homeomorphisms in the
mapping class group of the fibres of doubled torus bundles. This gives rise to an
isomorphism of K-theory groups, which are the groups of D-brane
charges on the pertinent space; as this only concerns how topological
data of the space change under T-duality, it is commonly refered to as
`topological T-duality', to distinguish it from the more physical
notion of T-duality which also dictates how geometric data on
the space should transform. It was shown by~\cite{Mathai2004} that
this can be reformulated in terms of the $C^*$-algebra of functions on
the space by considering its crossed product by an action of the
abelian Lie group~$\R^n$, leading to a general T-duality formalism
that can be regarded as a noncommutative version of the topological
aspects of the Fourier--Mukai transform; this version of T-duality is
often called the `$C^*$-algebraic formulation' of topological
T-duality.

The story becomes more interesting for spaces that are endowed with a
non-trivial gerbe, which in string theory typically comprise torus
bundles with `$H$-flux'. The gerbe can be encoded in the data of a
continuous-trace $C^*$-algebra with a non-trivial Dixmier--Douady
class, which is a noncommutative algebra to which the formalism of
topological T-duality was applied originally
in~\cite{Bouwknegt2003,Mathai2004,Mathai2005}. In addition to relating
spaces with different topologies, T-duality in string theory for such
instances predicts the existence of `non-geometric' spaces, called T-folds~\cite{Hull2004},
which cannot be viewed as conventional Hausdorff topological
spaces. In these instances the correspondence space picture
`geometrizes' the action of T-duality. It was shown
by~\cite{Bouwknegt2008, Mathai2004,Mathai2005} that the T-folds
of~\cite{Hull2004} have a rigorous incarnation in noncommutative
geometry as $C^*$-algebra bundles of noncommutative tori; necessary
and sufficient conditions for the existence of `classical' T-dual
Hausdorff spaces were developed in terms of topological data, and explicit constructions of `non-classical'
T-duals as noncommutative torus bundles were given. These points of view were harmonised
in~\cite{Brodzki2006}, and in~\cite{Brodzki2007} a $C^*$-algebraic
version of the correspondence space construction was given. The explicit connections of these
noncommutative torus bundles to the T-folds of~\cite{Hull2004} in the
setting of noncommutative gauge theories on D-branes in T-folds
was elucidated
in~\cite{Ellwood2006,Grange2006,Hull2019, Lowe2003}. Topological
T-duality and T-folds have also been studied rigorously from other
approaches based on homotopy
theory~\cite{Bunke2006, Bunke2005} and on higher
geometry~\cite{Nikolaus2018}.

\looseness=-1 In string theory, the simplest examples of torus bundles are sometimes
called `twisted to\-ri'~\cite{Dabholkar2002}; although this name is a
misnomer, we continue to use it as it is convenient for our
purposes. These are fibrations of $n$-dimensional tori $\torus^n$ over
a circle $\torus$ which do not carry the extra data of a gerbe; they have
monodromy in the mapping class group $\sfSL(n,\Z)$ of the torus
fibers. The simplest examples of these, the Heisenberg nilmanifolds,
are T-dual to torus bundles with $H$-flux and also to T-folds, and they
arise in the $C^*$-algebraic constructions
of~\cite{Mathai2004,Mathai2005}. However, there are other examples
which do not have any classical dual with $H$-flux, and these are
missed by the usual $C^*$-algebraic framework which starts from
continuous-trace algebras. The simplest examples of these with $n=2$ were studied
in~\cite{Hull2019} in the language of noncommutative gauge theories,
where it was shown that the monodromy of the original torus bundles
becomes a non-trivial Morita equivalence of the fiber noncommutative
tori of the dual $C^*$-algebra bundle. As far as we are aware, these
are new examples of noncommutative torus bundles which have not been
rigorously studied in the mathematics literature, and the primary
purpose of this paper is to fill this gap: starting from the
$C^*$-algebra of functions on a twisted torus in any dimension, we
give a rigorous construction of the topological T-duals in the
$C^*$-algebraic framework and precisely describe the non-classical
$C^*$-algebra bundles with their Morita equivalence monodromies. This
includes some of the examples from~\cite{Mathai2004,Mathai2005} based
on topological T-duality applied to continuous-trace algebras, and the
examples of~\cite{Hull2019} based on T-duality in noncommutative gauge
theory, while at the same time it produces many new examples. In
particular, we give a unified description of the noncommutative torus bundles
which are T-dual to twisted tori in any dimension.

The noncommutative gauge theory on a D-brane comes with other moduli,
in addition to the noncommutativity parameters, which also generally
transform in a non-trivial way under the monodromies so as to leave the physics
unchanged~\cite{Hull2019}. In the absence of other moduli, as in topological
T-duality, the non-trivial Morita equivalences of the fibres of the
$C^*$-algebra bundles require an interpretation akin to the
topological monodromies, which act as homeomorphisms in the mapping
class group $\sfSL(n,\Z)$ of the fibres of the original twisted
torus. This is naturally achieved by considering our $C^*$-algebras as objects in a
category where both the usual $*$-isomorphisms as well as Morita equivalences are
realised as isomorphisms. This category is well-known to experts and
all of our considerations of topological T-duality in this paper will
take place therein. This perspective will also be advantageous for
eventual rigorous considerations of noncommutative gauge theories on
these $C^*$-algebra bundles in terms of projective modules, as well as
for the treatments of D-branes in terms of their K-theory, though we
do not pursue these further aspects in the present paper.

\subsection{Summary and outline}

In this paper our starting point is a very general definition of a
twisted torus $\torus_{\Lambda_\sfG}$ as the quotient of a locally
compact group $\sfG$ by a lattice $\Lambda_\sfG$ in $\sfG$; this definition
encompasses the $\torus^n$-bundles over $\torus$ discussed above,
along with many other known examples from string theory. We regard~$\torus_{\Lambda_\sfG}$ as a `torus bundle without $H$-flux', which is
captured simply by the $C^*$-algebra of functions~$C(\torus_{\Lambda_\sfG})$. This is ultimately the novelty of our
approach, which leads to a simpler perspective on topological
T-duality as compared to the approach of~\cite{Mathai2004,Mathai2005}
based on the more complicated continuous-trace algebras. Our approach
uses similar techniques as those of~\cite{Mathai2004,Mathai2005} for
evaluating Morita equivalences of cross products by actions of $\R^n$,
though with a much simpler $C^*$-algebraic structure. In particular,
in this paper we do not develop any new $C^*$-algebraic machinery as such, but instead
we gather a fortuitously existing collection of results that enable us
to explicitly identify both classical and non-classical T-duals of
twisted tori with relatively straightforward algebraic techniques. On
the other hand, the tradeoff for the simplicity of our framework is
the absence of some key constructions
from~\cite{Bouwknegt2008, Mathai2004,Mathai2005}.

We have endeavoured throughout to provide a fairly self-contained, and at times
pedagogical, presentation. For this reason we have collected all the
key concepts and tools involving cross products of $C^*$-algebras and
Morita equivalence in Section~\ref{sec:NAdef}. Experts versed in
$C^*$-algebra theory may safely skip this section.

In Section~\ref{sec:TopTtorus} we give our definition of twisted tori
$\torus_{\Lambda_\sfG}$ and discuss the $C^*$-algebraic formulation of topological
T-duality that we employ in this paper. We describe how T-dual
$C^*$-algebras are naturally isomorphic when regarded as objects of
the additive category $\CCK\hspace{-1mm}\CCK$ that underlies
Kasparov's bivariant K-theory, and we adapt the construction of
noncommuative correspondences from~\cite{Brodzki2007} as diagrams in
this category. We spell out some simple techniques that we use to
compute classical T-duals with $H$-flux, i.e., the T-dual algebra is a
certain continuous-trace $C^*$-algebra with non-trivial
Dixmier--Douady class, and more general techniques based on Green's
symmetric imprimitivity
theorem which enable the computation of noncommutative T-duals. We
illustrate our scheme on two well-known examples which have classical
T-duals: we reproduce the standard rules for T-duality of tori as well
as the topology changing mechanism for T-duality of orbifolds of
compact Lie groups $\sfG$.

In Section~\ref{sec:TdualityMostow} we come to the main class of
examples and results of this paper. We review the definition and
topology of the special class of twisted tori provided by almost
abelian solvmanifolds, which are $\torus^n$-bundles over a circle
$\torus$. Accordingly, we regard the algebra of functions
$C(\torus_{\Lambda_\sfG})$ as an object in the category
$\mathcal{R}\CCK\hspace{-1mm}\CCK_\torus$ of $C^*$-algebra bundles
over $\torus$, where in particular fibrewise Morita equivalences are
isomorphisms. T-duality in this category requires fibrewise actions of
$\R^n$, and in particular $\R^n$-actions which act non-trivially on
the base $\torus$ would take the algebra out of the category
$\mathcal{R}\CCK\hspace{-1mm}\CCK_\torus$. This means that the `essentially
doubled spaces' of~\cite{Hull2019}, which arise from T-duality along
the base circle and require a completely doubled formalism, are not
considered in this paper; they require working in a different
category, which we do not discuss here. We give necessary and
sufficient criteria for the existence of classical T-duals with
$H$-flux in this case which are based on simple algebraic data of the
underlying group $\sfG$, and we explicitly compute the corresponding
continuous-trace $C^*$-algebras dual to any almost abelian
solvmanifold $\torus_{\Lambda_\sfG}$ satisfying these conditions. We
 further show that any such solvmanifold has a non-classical T-dual
 that is a $C^*$-algebra bundle of noncommutative $n$-tori over
 $\torus$, which we also compute explicitly; this rigorously
 confirms, in particular, arguments from string theory suggesting that non-geometric
 solutions result from T-duality on some six-dimensional almost
 abelian solvmanifolds~\cite{Andriot:2015sia}.

Finally, Section~\ref{sec:3dexamples} is devoted to explicit examples
of the general formalism of Section~\ref{sec:TdualityMostow}. We apply
our results to all three classes of three-dimensional
solvmanifolds. We recover in this way a new perspective on the
well-known T-duals of the Heisenberg nilmanifolds: the three-torus~$\torus^3$ with $H$-flux, and the basic noncommutative principal
$\torus^2$-bundle over $\torus$ given by the group $C^*$-algebra of the integer
Heisenberg group. Our general formalism also rigorously reproduces
the noncommutative torus bundles from~\cite{Hull2019} which are T-dual
to Euclidean solvmanifolds for the $\Z_4$ elliptic conjugacy class of
$\sfSL(2,\Z)$. We extend these results to give new examples of
noncommutative torus bundles dual to Euclidean solvmanifolds for the
$\Z_2$ and $\Z_6$ elliptic conjugacy classes, as well as to the
Poincar\'e solvmanifolds. In particular, our formalism extends the
$C^*$-algebraic formulation of topological T-duality to the case of
non-principal torus bundles, which have also been previously
considered in~\cite{Hannabuss2011}.

\section{Crossed products and duality\label{sec:NAdef}}

In this section we summarise some of the mathematical tools that we
will use in this paper. A~good reference for the material covered in
the following is the book~\cite{Williams2007}. Throughout this paper,
all topological spaces are assumed to be second countable (hence
separable), locally compact and Hausdorff.

\subsection{Dynamical systems and their crossed products}

Let $\sfG$ be a locally compact group and let $X$ be a $\sfG$-space, i.e., a
topological space which is acted
upon
homeomorphically by $\sfG$; we denote the $\sfG$-action $\sfG\times
X\to X$ by $(\gamma,x)\mapsto\gamma\cdot x$. The pair $(X,\sfG)$ is
called a \emph{transformation group}.
 A related concept is that of a
\emph{dynamical system}, which is a triple $(\alg,\sfG,\alpha)$
consisting of an algebra $\alg$ and a locally compact group $\sfG$ acting
on~$\alg$ via a group homomorphism $\alpha\colon \sfG\to\Aut (\alg)$,
denoted $\gamma\mapsto(\alpha_\gamma\colon \alg\to\alg)$ for $\gamma\in
\sfG$. In topological
T-duality one usually requires $\alg$ to be a $C^*$-algebra, in which
case $(\alg,\sfG,\alpha)$ is called a \emph{$C^*$-dynamical
 system}. Two dynamical systems $(\alg,\sfG,\alpha)$ and $(\balg,\sfG,\beta)$ are
\emph{equivalent} if there is an algebra isomorphism
$\varphi\colon \alg\to\balg$ which intertwines the $\sfG$-actions: $\varphi\circ\alpha_\gamma =
\beta_\gamma\circ\varphi$ for all $\gamma\in \sfG$.

If $\alg$ is a commutative $C^*$-algebra, then we call $(\alg,\sfG,\alpha)$ a
\emph{commutative} dynamical system. In that case, by Gelfand duality
$\alg=C_0(X)$ is the algebra of continuous functions vanishing at
infinity on a
topological space $X$
equipped with the $\sfG$-action \smash{$\alpha_\gamma^\dag\big|_X$} for $\gamma\in \sfG$ under
the identification of points $x\in X$ with irreducible representations
of $C_0(X)$, which are one-dimensional and given by the point evaluation
maps ${\rm ev}_x\colon \alg\to\alg$ with ${\rm ev}_x(\mathtt{f})=\mathtt{f}(x)$; then
$(X,\sfG)$ is a~transformation group. Conversely, given a transformation
group $(X,\sfG)$, there is an associated commutative dynamical system
$(C_0(X), \sfG,\alpha)$, where $\alpha_\gamma(\mathtt{f})(x)=\mathtt{f}\big(\gamma^{-1}\cdot x\big)$ for $\gamma\in
\sfG$, $\mathtt{f}\in C_0(X)$ and $x\in X$. In other words, there is a one-to-one
correspondence between transformation groups and commutative $C^*$-dynamical
systems.

If the $C^*$-algebra $\alg$ is not commutative, then we call
$(\alg,\sfG,\alpha)$ a \emph{noncommutative} $C^*$-dynamical system.

As usual, it is more useful to work with a representation rather than the
abstract dynamical system itself. A \emph{covariant representation} of
a dynamical system $(\alg,\sfG,\alpha)$ in a $C^*$-algebra $\balg$ with
multiplier algebra ${\rm M}(\balg)$ is a
pair $(\Pim,U)$ consisting of a~homomorphism $\Pim\colon \alg\to {\rm M}(\balg)$ and a~unitary representation $U\colon \sfG\to{\rm M}(\balg)$, $\gamma\mapsto U_\gamma$ which satisfy the compatibility condition
\[
\Pim\big(\alpha_\gamma(a)\big) = U_\gamma \Pim(a) U_\gamma^{-1} ,
\]
for all $\gamma\in \sfG$ and $a\in\alg$. A natural choice is to take
$\balg=\CK(\hil)$ to be the $C^*$-algebra of compact operators on a
separable Hilbert space $\hil$, which gives a
representation $\Pim\colon \alg\to{\rm B}(\hil)$ of the algebra $\alg$ by
bounded operators ${\rm B}(\hil)$ on $\hil$ and a unitary representation
$U\colon \sfG\to{\rm B}(\hil)$ of the group $\sfG$ on~$\hil$; in this case
we call $(\Pim,U)$ a covariant representation of $(\alg,\sfG,\alpha)$ on~$\hil$.

When a group $\sfG$ acts on a space $X$, one is naturally interested
in considering the
quotient space $X/\sfG$ of $\sfG$-orbits on~$X$. When $\sfG$ acts
freely and properly on
$X$, this is described algebraically by the algebra of functions
$C_0(X/\sfG)$. More generally, the subalgebra of $\sfG$-invariant elements
$\alg^\sfG\subseteq\alg$ of a $\sfG$-algebra $\alg$ can be used to represent the quotient, even for
$\sfG$-actions with fixed points. A more general and systematic way of
dealing with the effective algebraic ``quotient'' is through the
\emph{crossed product} algebra $\alg\rtimes_\alpha \sfG$ for a dynamical
system $(\alg,\sfG,\alpha)$. For a transformation group $(X,\sfG)$,
this description is particularly powerful in the cases where the
quotient $X/\sfG$ is not a Hausdorff space, while for a free and proper
$\sfG$-action it gives an algebra with the same spectrum $X/\sfG$
as the algebra $C_0(X/\sfG)=C_0(X)^\sfG$ of
$\sfG$-invariant functions on $X$.

In order to define the crossed product
algebra, we first define
\[
\|f\|_{\rm univ} := \sup_{({\mit\Pi},U)} \big\|({\mit\Pi}\rtimes_\alpha U)(f)\big\|
\]
for compactly supported functions $f\in C_{\rm c}(\sfG,\alg) $, where the supremum is taken
over (possibly degenerate) covariant representations
$({\mit\Pi},U)$ of $(\alg, \sfG, \alpha)$ with
\[
(\Pim\rtimes_\alpha U)(f) := \int_\sfG
\Pim\big(f(\gamma)\big) U_\gamma\, \dd\mu_\sfG(\gamma) ,
\]
and $\mu_\sfG$ denotes the left invariant Haar measure on $\sfG$. This defines a norm, called the universal norm, on the
space $C_{\rm c}(\sfG,\alg)$.
Then the crossed product algebra $\alg\rtimes_\alpha \sfG$ is the
completion (in the universal norm) of the algebra $C_{\rm c}(\sfG,\alg)$
equipped with the convolution product
\begin{gather}\label{eq:convproduct}
(f\star f')(\gamma) := \int_\sfG\, f(\gamma') \alpha_{\gamma'}\big(f'\big(\gamma'^{-1} \gamma\big)\big)\, \dd \mu_\sfG(\gamma')
,
\end{gather}
for all $f,f'\colon \sfG\to\alg$. In general this is a noncommutative
multiplication, even for commutative dynamical systems.
Since $\alg$ is a $C^*$-algebra, there is a $*$-structure on the
convolution algebra defined by
\[
f^\dag(\gamma):= \Delta_\sfG(\gamma)^{-1} \alpha_\gamma\big(f\big(\gamma^{-1}\big)^*\big) ,
\]
where $\Delta_\sfG\colon \sfG\to \real^+$ is the modular function of
$\sfG$ defined through
\[
\Delta_\sfG(\gamma')
\int_\sfG\,f(\gamma)\,\dd\mu_\sfG(\gamma)=\int_\sfG f(\gamma \gamma')\,\dd\mu_\sfG(\gamma)
\]
for $f\in C_{\rm c}(\sfG,\alg)$ and $\gamma'\in \sfG$. By the uniqueness of the left
invariant Haar measure $\mu_\sfG$ up to a~positive
constant, $\Delta_\sfG(\gamma')$ is independent of~$f$
and $\Delta_\sfG$ is easily proven to be a continuous group homomorphism from~$\sfG$ to the
multiplicative group~$\real^+$; it is trivial for abelian groups and for compact
groups.

When $\alg=C_{\rm c}(X)$
is the algebra of a commutative dynamical system, the convolution
algebra $C_{\rm c}(\sfG\times X)$ consists of functions $f\colon\sfG\times X\to\complex$ and the
convolution product reads as
\[
(f\star f')(\gamma,x) = \int_\sfG f(\gamma',x)
f'\big(\gamma'^{-1} \gamma,\gamma'^{-1}\cdot x\big)\, \dd \mu_\sfG(\gamma') ,
\]
while the $*$-algebra structure is given by
\[
f^\dag(\gamma,x) = \Delta_\sfG(\gamma)^{-1} \overline{f\big(\gamma^{-1},\gamma^{-1}\cdot x\big)} .
\]
The crossed product is a generalization of the usual group algebra
$C^*(\sfG)$, the completion (in the universal norm) of $C_{\rm
 c}(\sfG)$ which is recovered in the case $\alg=\complex$
(the $C^*$-algebra of a point)
wherein $\alpha_\gamma={\rm id}_\alg\colon \alg\to \alg$ for all $\gamma\in \sfG$ and~\eqref{eq:convproduct} recovers the usual convolution product of
functions on the group $\sfG$. As explained in
Section~\ref{sec:semidirect} below, the
group $C^*$-algebra description illustrates the relation between crossed
products and semi-direct products of groups (see Theorem~\ref{cpsp2}). We also
note that if a group $\sfG$ acts trivially on an algebra~$\alg$, then
$\alg\rtimes\sfG\simeq\alg\otimes C^*(\sfG)$.

The crossed product can be thought of as a universal
object for covariant representations of the dynamical system
$(\alg,\sfG,\alpha)$, in the following sense: Define the \emph{universal
 covariant representation} $(\mbf\Pim,\mbf U)$ of $(\alg,\sfG,\alpha)$
in $\alg\rtimes_\alpha \sfG$ by
\[
\big(\mbf\Pim(a)f\big)(\gamma)=a f(\gamma) \qquad \mbox{and} \qquad
\big(\mbf U_{\gamma'}f\big)(\gamma) =
\alpha_{\gamma'}\big(f\big(\gamma'{}^{-1} \gamma\big)\big) ,
\]
for $a\in\alg$, $f\in C_{\rm c}(\sfG,\alg)$ and $\gamma,\gamma'\in \sfG$. Then the
universal property defining the crossed product implies that any covariant representation $(\Pim,U)$ of
$(\alg,\sfG,\alpha)$ in a $C^*$-algebra $\balg$ factors through the
universal covariant representation: There exists a unique homomorphism
\mbox{$\mbf\varphi\colon {\rm M}(\alg\rtimes_\alpha \sfG)\to{\rm M}(\balg)$} such that
\[
\Pim=\mbf\varphi\circ\mbf\Pim \qquad \mbox{and} \qquad U_\gamma =
\mbf\varphi(\mbf U_\gamma)
\]
for all $\gamma\in \sfG$.

If $(\Pim,U)$ is a covariant representation of the dynamical system
$(\alg,\sfG,\alpha)$ on a Hilbert space~$\hil$, then
\[
{\mit\Phi}_{(\Pim,U)}(f):=(\Pim\rtimes_\alpha U)(f)
\]
defines a representation ${\mit\Phi}_{(\Pim,U)}\colon C_{\rm c}(\sfG,\alg)\to{\rm B}(\hil)$
of the crossed product $\alg\rtimes_\alpha \sfG$ as bounded operators on
$\hil$. This is called the \emph{integrated form} of the covariant
representation $(\Pim,U)$. In particular, it maps the convolution
product \eqref{eq:convproduct} onto the operator product in the
algebra ${\rm B}(\hil)$,
\[
{\mit\Phi}_{(\Pim,U)}(f\star g) = {\mit\Phi}_{(\Pim,U)}(f)
{\mit\Phi}_{(\Pim,U)}(g) ,
\]
and it is covariant in the sense that
\[
{\mit\Phi}_{(\Pim,U)}\big(i_\sfG(\gamma)(f)\big) = U_\gamma
{\mit\Phi}_{(\Pim,U)}(f) ,
\]
where $\big(i_\sfG(\gamma)(f)\big)(\gamma'):=\alpha_\gamma\big(f\big(\gamma^{-1}
\gamma'\big)\big)$ for each $\gamma,\gamma'\in \sfG$ and $f\in C_{\rm
 c}(\sfG,\alg)$.

\begin{Example}[noncommutative two-tori]\label{ex:NCtori}
The noncommutative torus is a fundamental example
of a noncommutative space in both physics and mathematics. Its
original incarnation~\cite{Rieffel1981} is a nice example of a crossed
product construction, which will play a fundamental role later on in this
paper. Let $\big(C(\torus),\Z,\tau^\theta\big)$ be the commutative $C^*$-dynamical system
where $\tau^\theta$ is induced through pullback by rotations of the circle $\torus$ through
a fixed angle $\theta\in\real/\Z$:
\[
\tau^\theta_n(\mathtt{f})(z) = \mathtt{f}\big(\e^{2\pi\ii n \theta} z\big) ,
\]
for $n\in\Z$, $\mathtt{f}\in C(\torus)$ and $z\in\torus$. The resulting crossed
product
\[
\A_\theta:=C(\torus)\rtimes_{\tau^\theta}\Z
\]
is a called a \emph{rotation
algebra}, and for irrational values of $\theta$ it can be identified as a noncommutative two-torus
$\torus_\theta^2$ in the following way.

By definition, the algebra
$\A_\theta$ is the universal norm completion of the convolution
algebra $C_{\rm c}(\Z\times\torus)$, whose elements $f=\{f_n\}_{n\in\Z}$ can be regarded as
sequences (with only finitely many nonvanishing
terms) of functions $f_n\colon \torus\to\complex$. The convolution product is
given by
\[
(f\star_\theta g)_n(z) := \sum_{n'\in\Z} f_{n'}(z) g_{n-n'}\big(\e^{2\pi\ii n' \theta} z\big) ,
\]
and the $*$-algebra structure is
\[
f^\dag_n(z):=\overline{f_{-n}\big(\e^{2\pi\ii n \theta} z\big)} .
\]
Via the Fourier transformation
\[
f(z,w) := \sum_{n\in\Z} f_{n}(z) w^n
\]
for $w\in\torus$, we may regard the convolution algebra $C_{\rm
 c}(\Z\times\torus)$ as a subspace of the space of functions $C\big(\torus^2\big)$ equipped
with the star-product
\begin{gather}\label{eq:starprodNCtorus}
(f\star_\theta g)(z,w) = \sum_{n\in\Z} (f\star_\theta g)_n(z) w^n .
\end{gather}
After a further Fourier transformation
\[
f_n(z) = \sum_{m\in\Z} f_{m,n} z^m
\]
and some simple redefinitions of the Fourier series involved, the
star-product \eqref{eq:starprodNCtorus} may be written in the form
\[
(f\star_\theta g)(z,w) = \sum_{(m,n)\in\Z^2}
\bigg(\sum_{(m',n')\in\Z^2} f_{m',n'} g_{m-m',n-n'}
\e^{2\pi\ii (m-m') n' \theta}\bigg) z^m w^n .
\]
This recovers the usual commutative pointwise multiplication of functions in
$C\big(\torus^2\big)$ for $\theta=0$. For $\theta\neq0$ it realizes the irrational rotation algebra $\A_\theta$ as a
deformation of the algebra of functions $C\big(\torus^2\big)$ on a
two-dimensional torus $\torus^2$; it is equivalent to the usual strict
deformation quantization of $\torus^2$ whose star-product is a twisted
convolution product on $C\big(\torus^2\big)$.

In the language of covariant representations of the dynamical system
$(C(\torus),\Z,\tau^\theta)$, the crossed product $\A_\theta$ is the
universal $C^*$-algebra generated by two unitaries $U$ and $V$
satisfying the relation~\cite[Proposition~2.56]{Williams2007}
\begin{gather}\label{eq:C*UV}
U V=\e^{-2\pi\ii \theta} V U .
\end{gather}
A concrete representation
of $\A_\theta$ on the Hilbert space $\hil = {\rm L}^2(\torus)$ is
given by defining
\[
U(\mathtt{f})(z)=z \mathtt{f}(z) \qquad \mbox{and} \qquad V(\mathtt{f})(z) =
\mathtt{f}\big(\e^{2\pi\ii\theta} z\big) .
\]
\end{Example}

\begin{Example}[{{\bf Noncommutative $\mbf{d}$-tori}}]
\label{ex:NCtorid}
The natural higher-dimensional generalization of
Example~\ref{ex:NCtori} involves a skew-symmetric real $d{\times}d$
matrix $\Theta=(\theta_{ij})$, see~\cite{Rieffel1990}. Then the noncommutative $d$-torus
$\A_\Theta=\torus_\Theta^d$ is the universal $C^*$-algebra generated
by $d$ unitaries $U_1,\dots,U_d$ satisfying the relations
\[
U_i U_j = \e^{-2\pi\ii\theta_{ij}} U_j U_i
\]
for $i,j=1,\dots,d$. By~\cite[Lemma~1.5]{Phillips2006}, every
noncommutative torus $\torus_\Theta^d$ can be obtained as an iterated
crossed product by $\Z$ in the following way. Let
$\Theta_{|d-1}=(\theta_{ij})_{1\leq i,j\leq d-1}$, and let
$U_1,\dots,U_{d-1}$ be the standard generators of
$\A_{\Theta_{|d-1}}=\torus^{d-1}_{\Theta_{|d-1}}$. Define a
group homomorphism
\smash{$\tau^{\vec\theta}\colon \Z\to\Aut(\A_{\Theta_{|d-1}})$} by
\[
\tau_n^{\vec\theta}(U_i) = \e^{2\pi\ii n \theta_{id}} U_i
 ,
\]
for $n\in\Z$, where $\vec\theta:=(\theta_{id})\in\R^{d-1}$.
Then there is an isomorphism of $C^*$-algebras
\begin{gather}\label{eq:AThetacrossed}
\A_\Theta\simeq\A_{\Theta_{|d-1}}
\rtimes_{\tau^{\vec\theta}} \Z .
\end{gather}

In the particular case where $\Theta_{|d-1}=\mbf 0_{d-1}$, we denote the
corresponding noncommutative $d$-torus by
$\A_{\vec\theta}=\torus_{\vec\theta}^d\,$, and~\eqref{eq:AThetacrossed}
shows that it can be obtained by a crossed product of the commutative
algebra of functions on a $d{-}1$-torus by an action of~$\Z$:
\[
\A_{\vec\theta} \simeq C\big(\torus^{d-1}\big)\rtimes_{\tau^{\vec\theta}}\Z .
\]
\end{Example}

\subsection{Semi-direct products and group algebras}\label{sec:semidirect}

Most of our considerations later on will focus on spaces that can be
obtained from semi-direct products of groups. We will now explain the
relation between crossed products and semi-direct products which will
be useful for these examples.

There are two ways to think about the semi-direct product construction:
\begin{enumerate}\itemsep=0pt
\item[(1)] Let $\sfG$ be a
group with two subgroups $\sfN $ and $\sfH$ such that $\sfN $ is normal. If $\sfN \cap
\sfH=\{e\}\subset \sfG$ and every
element of $\sfG$ can be written as a product of an element of
$\sfN $ with an element of $\sfH$, then we say that $\sfG$ is a semi-direct
product of its subgroups $\sfN $ and $\sfH$ and we write $\sfG=\sfN \,\sfH$.

\item[(2)]
Let $\sfN $ and $\sfH$ be two groups together with a left action
$\varphi\colon \sfH\to\Aut(\sfN )$ of $\sfH$ on $\sfN $ by automorphisms, which we denote by
$\varphi_h(n)={}^hn$ for $h\in \sfH$ and $n\in \sfN $; in particular
${}^h(n n')={}^hn\, {}^hn'$. We write ${}^\sfH \sfN $ to
indicate that $\sfH $ acts on $\sfN $ from the left. The semi-direct product
of $\sfN $ and $\sfH $ is the group $\sfN \rtimes_\varphi \sfH $ defined to be the set
$\sfN \times \sfH $ with the product
\[
(n,h)\,(n',h') = \big(n\,{}^hn',h h'\big) .
\]
The inverse is then $(n,h)^{-1}=\big({}^{h^{-1\!}}{n^{-1}},h^{-1}\big)$.
\end{enumerate}
These two definitions are equivalent: Given subgroups $\sfN ,\sfH \subset \sfG$ as
in point (1), it follows that $\sfG\simeq \sfN \rtimes_{\rm Ad} \sfH $
where ${\rm Ad}$ is the adjoint action: ${\rm Ad}_h(n)=h n h^{-1}$. On the other
hand, every element of the group $\sfG=\sfN \rtimes_\varphi \sfH $ defined in~(2) can be written as $(n,h)=(n,e_\sfH) (e_\sfN,h)$ and the subgroups $\sfN \times\{e_\sfH\}$
and $\{e_\sfN\} \times \sfH $ intersect only in the identity of $\sfG$.
If the action of $\sfH $ on $\sfN $ is trivial,
i.e., $\varphi_h={\rm id}_\sfN $ for all $h\in \sfH $, then the semi-direct
product reduces to the direct product $\sfN \rtimes_\varphi
\sfH =\sfN \times \sfH $.

Later on we will need to consider the interplay between semi-direct
products and quotient groups, which is provided by the simple
\begin{Lemma}
\label{lem:semidirectquotient}
Let $\sfG=\sfN\rtimes_\varphi\sfH$ be a semi-direct
product, and let $\V\subset \sfN$ be a subgroup which is normal in
$\sfG$. Then the quotient group $\sfG/\V$ is the
semi-direct product $(\sfN/\V)\rtimes_{\varphi^\V}\sfH$, where
${\varphi^\V}$ is the action $\varphi$ of $\sfH$ induced on the
quotient group $\sfN/\V$.
\end{Lemma}

If $\sfN $ is a group, we write $C^*(\sfN )$ for the corresponding group
$C^*$-algebra, i.e., for the crossed product $\complex\rtimes \sfN $. If
$\sfN $ is finite, then $C^*(\sfN )=\complex[\sfN ]$ is the linear space
freely generated over $\complex$ by the group elements, made into an algebra by
linearly extending the product from $\sfN $ to $\complex[\sfN ]$; equivalently
it is the algebra of continuous functions
on $\sfN $ with the convolution product.
 Given a
left $\sfH $-action $\varphi\colon \sfH \to\Aut(\sfN )$, there is an induced action
$\varphi^*\colon \sfH \to \Aut(C^*(\sfN ))$ via pullback. For $\sfH $ and $\sfN $ finite the
vector spaces $ C^*(\sfN )\times \sfH $ and $C^*(\sfN \times \sfH )$ are canonically
isomorphic, and it is straightforward
to show that the corresponding crossed product and semi-direct product are
related by
\begin{Proposition}\label{cpsp}
If $\sfN$ and $\sfH$ are finite groups, then
\[
C^*(\sfN )\rtimes_{\varphi^*}\sfH \simeq C^*(\sfN \rtimes_\varphi \sfH ) .
\]
\end{Proposition}
\begin{proof} Note that $C^*(\sfN )\rtimes_{\varphi^*}\sfH =\complex[\sfN ]\rtimes_{\varphi^*}\sfH $ is the
 vector space of functions $f\colon \sfH \to \complex[\sfN ]$ with convolution product
\[
\big(f\star_{\complex[\sfN ]\rtimes_{\varphi^*}\sfH } f'\big)(h)=\sum_{h'\in
 \sfH }\,f(h')\star_{\complex[\sfN ]}\varphi^*_{h'}\big(f'\big(h'^{-1} h\big)\big) ,
\]
and using the convolution product in $\complex[\sfN]$ this can be
written as
\[
\big(f\star_{\complex[\sfN ]\rtimes_{\varphi^*}\sfH }
f'\big)(n,h)=\sum_{h'\in \sfH} \sum_{n'\in
 \sfN } f(n',h') f'\big({}^{h'^{-1\!}}\big(n'^{-1} n\big),h'^{-1} h\big) ,
\]
which is easily seen
 to coincide with the convolution product in $ C^*(\sfN \rtimes_\varphi \sfH )=\complex[\sfN \rtimes_\varphi \sfH ]$.
\end{proof}

A more general result holds if $\sfN $ and $\sfH $ are locally compact groups with $\varphi\colon \sfH \to \Aut (\sfN )$
a~continuous action of $\sfH $ on $\sfN $ via automorphisms
(i.e., $(h,n)\mapsto \varphi_h(n)$ is a continuous map from $\sfH \times \sfN $ to $\sfN $).
In this case the semi-direct product $\sfN \rtimes_\varphi \sfH $ is a locally compact
group (in the product topology on $\sfN \times \sfH $) with $\sfN $ a closed normal subgroup and $\sfH $ a closed
subgroup (see~\cite[Proposition~3.11]{Williams2007} with
$A=\complex$). The analogue of item~(1) above also holds in the
context of locally compact groups if $\sfG$ is $\sigma$-compact, and
$\sfN$ and $\sfH$ are closed subgroups of~$\sfG$.

The action $\beta$
defining the $C^*$-dynamical system $(C^*(\sfN ),\sfH ,\beta)$ and hence the crossed
product $C^*(\sfN )\rtimes_\beta \sfH $ is the composition
of the pullback $\varphi^*$ of the action $\varphi\colon \sfH \to \Aut(\sfN )$ with
the action $\sigma_\sfH\colon \sfH \to \real^+\subset \Aut(C^*(\sfN ))$ that enters the
definition of the Haar measure on $\sfN \rtimes_\varphi \sfH $ in terms of the Haar
measures on $\sfN $ and $\sfH $: If $\mu_\sfN $ is a (left invariant)
Haar measure on $\sfN $, then the integral
$I_h(F)=\int_\sfN F\big({}^hn\big)\, \dd\mu_\sfN (n)$ for $F\in C^*(\sfN )$
is left invariant, i.e.,
$I_h(\lambda_{n'}F)= I_h(F)$ where $(\lambda_{n'}F)(n):=F\big(n'^{-1} n\big)$
for all $n'\in \sfN $
(use invariance of the Haar measure under $n\to {}^{h^{-1\!}}n'$).
Uniqueness of the Haar measure up to a positive constant then
implies there exists a function $\sigma_\sfH\colon \sfH\to\real^+$ such that
\begin{gather}\label{sigmaI}
\sigma_\sfH(h) \int_\sfN F\big({}^hn\big) \, \dd\mu_\sfN (n)=\int_\sfN F(n)
\, \dd\mu_\sfN (n) .
\end{gather}
It is straighforward to see that
$\sigma_\sfH$ is a group homomorphism and that it is
continuous~\cite[Section~2]{Williams2007}. The Haar measure
$\mu_{\sfN\rtimes_\varphi\sfH}$ on
$\sfN \rtimes_\varphi \sfH $ is then given by
\[
\int_{\sfN\rtimes_\varphi\sfH} f(n,h) \,
\dd\mu_{\sfN\rtimes_\varphi\sfH}(n,h) := \int_\sfH \int_\sfN\,
f(n,h) \sigma_\sfH(h)^{-1} \, \dd\mu_\sfN (n) \, \dd\mu_\sfH (h) .
\]
This is trivially invariant under the left $\sfN $-action, and
it is also invariant under the left $\sfH $-action $(n,h)\mapsto
(1,h') (n,h)=\big({}^{h'\!}n, h' h\big)$, using \eqref{sigmaI} with $F\big({}^{h'\!}n\big) :=f\big({}^{h'\!}n, h' h\big)$
and recalling that~$h$ is fixed in~\eqref{sigmaI}.

\begin{Theorem}\label{cpsp2}
Let $\sfN $ and $\sfH $ be locally compact groups
 and $\varphi \colon \sfH \to \Aut(\sfN )$ a continuous action of $\sfH $ on $\sfN $. Define
$\beta\colon \sfH \to \Aut(C^*(\sfN ))$ by $(\beta_{h'}\ell)
(n)=\sigma_\sfH(h')^{-1} \ell\big({}^{h'^{-1\!}}n\big)$ for all $h'\in \sfH $ and
$\ell\in C_{\rm c}(\sfN )$. Then
\[
C^*(\sfN )\rtimes_{\beta}\sfH \simeq C^*(\sfN \rtimes_\varphi
\sfH ) .
\]
\end{Theorem}

For a full
proof of Theorem~\ref{cpsp2} that takes into account the topological
and $C^*$-algebraic aspects, see~\cite[Proposition~3.11]{Williams2007}. Here
we shall just show that under the canonical injection $C_{\rm c}(\sfN \rtimes_\varphi \sfH )\hookrightarrow
C_{\rm c}(\sfN )\rtimes_\beta \sfH $, given by $f(n,h)\mapsto f(h)$ where
$f(h)(n)=f(n,h)$, the convolution product in $C_{\rm c}(\sfN \rtimes_\varphi
\sfH )$ is mapped to the convolution product in $C_{\rm c}(\sfN )\rtimes_\beta \sfH $.
Let $f,f'\in C_{\rm c}(\sfN \rtimes_\varphi \sfH )$, then
\begin{gather}\label{eq:imagefstarf'}
\big(f\star_{C_{\rm c}(\sfN \rtimes_\varphi
 \sfH )}f'\big)(n,h)=\int_\sfH \int_\sfN f(n',h')
f'\big((n',h')^{-1} (n,h)\big) \sigma_\sfH(h')^{-1} \, \dd\mu_\sfN (n') \, \dd\mu_\sfH (h').
\end{gather}
On the other hand, for the images of $f$, $f'$ in $C_{\rm c}(\sfN
)\rtimes_\beta \sfH $ we have
\[
\big(f\star_{C_{\rm c}(\sfN )\rtimes_\beta
 \sfH }f'\big)(h)=\int_\sfH f(h')\star_{C^*(\sfN
 )}\beta_{h'}\big(f'\big(h'^{-1} h\big)\big) \, \dd\mu_\sfH (h') .
\]
Using the convolution product in $C^*(\sfN)$ this can be written as
\begin{align*}
\big(f\star_{C_{\rm c}(\sfN )\rtimes_\beta
 \sfH }f'\big)(h)(n)=\int_\sfH \int_\sfN
 f(h')(n') \beta_{h'}\big(f'\big(h'^{-1} h\big)\big)\big(n'^{-1} n\big)
\, \dd\mu_\sfN (n') \, \dd\mu_\sfH (h') ,
\end{align*}
which from the definition of the action $\beta$ is easily seen to
equal the image $(f\star_{C_{\rm c}(\sfN \rtimes_\varphi \sfH
 )}f')(h)(n)$ in $C_{\rm c}(\sfN )\rtimes_\beta
 \sfH $ of the product $(f\star_{C_{\rm c}(\sfN \rtimes_\varphi
 \sfH )}f')(n,h)$ in $C_{\rm c}(\sfN \rtimes_\varphi
\sfH )$ from~\eqref{eq:imagefstarf'}.

More generally we have~\cite[Proposition~3.11]{Williams2007}
\begin{Theorem}\label{cpsp3}
Let $(\alg,\sfN \rtimes_\varphi \sfH ,\alpha)$ be a $C^*$-dynamical
system for the semi-direct product group $\sfN \rtimes_\varphi \sfH $. Then $(\alg\rtimes_{\alpha|_\sfN }\sfN ,\sfH ,\beta)$ is a $C^*$-dynamical
system, where
\[
\beta\colon \ \sfH \longrightarrow \Aut(\alg\rtimes_{\alpha|_\sfN }\sfN ) ,
\qquad h\longmapsto \beta_h
\]
is defined by
$\big(\beta_h(f)\big)(n)=\sigma_\sfH(h)^{-1} \alpha_h\big(f\big(^{h^{-1\!}}n\big)\big)$
for all $f\in C_{\rm c}(\sfN ,\alg)\subset \alg\rtimes_{\alpha|_\sfN }\sfN $, with $\sigma_\sfH\colon
\sfH \to \real^+$ defined by~\eqref{sigmaI} and $^{h^{-1\!}}n=\varphi_{h^{-1}}(n)$.
Moreover, the canonical injection $C_{\rm c}(\sfN \rtimes_\varphi \sfH ,\alg)\hookrightarrow
C_{\rm c}\big(\sfH ,C_{\rm c}(\sfN ,\alg)\big)$
extends to a~$C^*$-algebra isomorphism
\[
\alg\rtimes_{\alpha}\big(\sfN \rtimes_\varphi
\sfH \big) \simeq \big(\alg\rtimes_{\alpha|_\sfN } \sfN
\big)\rtimes_\beta \sfH .
\]
\end{Theorem}

Theorem \ref{cpsp2} is then recovered by setting $\alg=\complex$.

In the spirit of Theorem~\ref{cpsp2},
 which shows that crossed products are a generalization of
 semi-direct products, let us mention the semi-direct
 product construction behind Theorem \ref{cpsp3}. Consider three groups $\sfM $, $\sfN $ and $\sfH $ with group actions ${}^\sfH \sfN $ and
${}^{\sfN \rtimes \sfH }\sfM $; then there are also group actions
${}^\sfN \sfM $ and ${}^\sfH \sfM $. The associativity
of the triple semi-direct product construction is then easily established
through
\begin{Proposition}\label{MNH}
Let $\sfM $, $\sfN $ and $\sfH $ be groups with group actions ${}^\sfH \sfM $, ${}^\sfH \sfN $ and
${}^\sfN \sfM $ satisfying the compatibility conditions
\[
{}^h\big({}^nm\big) = {}^{({}^hn)}\big({}^hm\big)
\]
for all $m\in \sfM $, $n\in \sfN $ and $h\in \sfH $. Then there exists a group
action ${}^{\sfN \rtimes \sfH }\sfM $ defined by ${}^{(n,h)}m={}^n\big({}^hm\big)$, and a
group action ${}^\sfH (\sfM \rtimes \sfN )$ defined by ${}^h(m,n)=\big({}^hm,{}^hn\big)$,
which together satisfy the associativity property
\[
\sfM \rtimes(\sfN \rtimes \sfH ) = (\sfM \rtimes \sfN )\rtimes \sfH .
\]
\end{Proposition}

\subsection{Pontryagin duality and Fourier transform}

If $\sfN $ is a locally compact {\sl abelian} group we denote by $\widehat{\sfN }$
its Pontryagin dual, i.e., the set of characters $\chi\colon \sfN \to \sfU(1)$,
which is also a locally compact abelian group (with the compact-open
topology and with the pointwise multiplication). For example, if
$\sfN =\real^d$ then $\widehat{\sfN }=\real^d$ and the characters are given by
$\chi_p(x)=\e^{2\pi\ii \langle p, x\rangle}$ for $x\in \sfN $ and
$p\in\widehat{\sfN }$. The Pontryagin duality theorem states that there is a
canonical isomorphism \smash{$\widehat{\widehat{\sfN }}\simeq \sfN $}, where
$n\in\sfN$ is associated to the character $\chi\mapsto\chi(n)$ on
\smash{$\widehat\sfN$}.

The Fourier transform shows that the group $C^*$-algebra $C^*(\sfN)$ is
isomorphic to $C_0\big(\widehat \sfN\big)$:
Given a Haar measure $\mu_{\sfN }$ on
$\sfN $, the Fourier transform $\mcF(f)$ of $f\in
C_{\rm c}(\sfN )$ is defined by
\[
\mcF(f)(\chi):=\int_{{\sfN }} {f}(n)
\chi(n) \, \dd\mu_{{\sfN }}(n)
\]
for $\chi\in \widehat \sfN $. It sends
 the convolution product of functions in $C^*(\sfN )$
to the pointwise product of functions in $C(\widehat{\sfN })$:
\[
\mcF({f}\star{f'}) = \mcF({f})\mcF({f'}) ,
\]
and extends to an isomorphism~\cite[Proposition~3.1]{Williams2007}
\begin{equation}\label{FTI}
{\cal{F}}\colon \ C^*(\sfN )
\xrightarrow{ \ \simeq \ } C_0\big(\widehat{\sfN }\big) ,
\end{equation}
where $C_0\big(\widehat \sfN \big)$ is the algebra of functions on $\widehat \sfN $ vanishing at
infinity. For $\sfN $ separable, Hausdorff and locally compact,
$C_0\big(\widehat{\sfN }\big) $ is indeed a $C^*$-algebra.

Given a continuous left group action $\varphi\colon \sfH \to\Aut(\sfN )$, which we
 also denote as before by $\varphi_h(n)={}^hn$, consider the induced action
 $\beta\colon \sfH \to\Aut\big(C^*(\sfN )\big)$ as defined in Theorem~\ref{cpsp2}.
There is also an induced left action $\widehat{\varphi} \colon \sfH \to\Aut\big(\widehat{\sfN }\big)$ defined by pullback:
$\big(\widehat{\varphi}_h\chi\big)(n):=\chi\big({}^{h^{-1}}n\big)$, together with its pullback
$\widehat{\varphi}^{\,*}\colon \sfH \to\Aut\big(C_0\big(\widehat{\sfN }\big)\big)$ defined by
$\big(\widehat{\varphi}^{\,*}_h\widehat f\,\big)(\chi)=\widehat
f \big(\widehat{\varphi}_{h^{-1}}\chi\big)$ for all $h\in\sfH$, $\widehat f\in
C_0\big(\widehat \sfN \big)$ and $\chi\in\widehat \sfN$.
The Fourier transform isomorphism \eqref{FTI} then extends to the isomorphism
\begin{Proposition} \label{prop:Pontcp}
If $\sfN$ is a locally compact abelian group and
$\varphi\colon \sfH\to\Aut(\sfN)$ is a continuous action of a locally
compact group $\sfH$ on $\sfN$, then
\[
\quad C^*(\sfN )\rtimes_{\beta}\sfH \simeq C_0\big(\widehat \sfN
\big)\rtimes_{{\widehat\varphi}^{\,*}}\sfH .
\]
\end{Proposition}
\begin{proof}
We show that the triples $(C^*(\sfN ),\sfH ,\beta)$ and $\big(C_0\big(\widehat
\sfN \big),\sfH ,\widehat{\varphi}{}^{\,*}\big)$ are equivalent dynamical systems,
see~\cite[Example~3.16]{Williams2007}.
For this, we prove that the Fourier transform
\eqref{FTI} is $\sfH $-equivariant with respect to the $\sfH$-actions $\beta$ and $
\widehat\varphi^{\,*}$. For $h\in\sfH$, $f\in C_{\rm c}(\sfN)$ and
$\chi\in\widehat\sfN$ we compute
\begin{align*}
\mcF\big(\beta_h(f)\big)(\chi)&=\int_\sfN \beta_h(f)(n) \chi(n) \,
 \dd\mu_\sfN (n)
 =\sigma_\sfH(h)^{-1} \int_\sfN f\big({}^{h^{-1\!}}n\big) \chi(n)
 \, \dd\mu_\sfN (n) \nn\\
&= \sigma_\sfH(h)^{-1} \int_\sfN
 f\big({}^{h^{-1\!}}n\big) \chi\big({\phantom{I^I\!\!\!\!\!}}^{h\!}\big({}^{h^{-1}\!}n\big)\big)
 \, \dd\mu_\sfN (n)
 = \int_\sfN f (n) \chi\big({}^h n\big) \, \dd\mu_\sfN (n) \nn\\
&=\mcF(f)\big(\widehat\varphi_{h^{-1}}\chi\big) =\big(\widehat\varphi_h^{\,*}\mcF(f)\big)(\chi) ,
\end{align*}
where in the fourth equality we used \eqref{sigmaI} with
$F\big({}^{h^{-1}\!}n\big)
=f\big({}^{h^{-1\!}}n\big) \chi\big({\phantom{I^I\!\!\!\!\!}}^{h\!}\big({}^{h^{-1}\!}n\big)\big)$.
\end{proof}

Replacing $\sfN$ with $\widehat \sfN $ in
Proposition~\ref{prop:Pontcp} we also obtain the isomorphism
\begin{gather}
\label{hathatduality}
C^*\big({\widehat \sfN }\big)\rtimes_{\widehat\beta}\sfH \simeq
C_0\big({\widehat
{\widehat \sfN }}\big) \rtimes_{{\widehat {\widehat\varphi}}{}^{\,*}}\sfH
 \simeq
C_0(\sfN )\rtimes_{{\varphi}^*}\sfH ,
\end{gather}
where $\widehat\beta\colon \sfH\to\Aut\big(C^*(\widehat\sfN)\big)$ is defined by
$\widehat\beta_h(\widehat
f\,)(\chi)=\widehat\sigma_\sfH(h)^{-1} \widehat
f \big(\widehat\varphi_{h^{-1}\!}\chi\big)$ for $h\in\sfH$, $\widehat f\in
C_{\rm c}(\widehat \sfN)$ and $\chi\in\widehat \sfN$, with $\widehat\sigma_\sfH\colon
\sfH \to \real^+$ defined as in~\eqref{sigmaI} but using the dual group
$\widehat \sfN $ instead of $\sfN$, and in the final isomorphism we
used Pontryagin duality \smash{$\widehat{\widehat \sfN } \simeq \sfN$}.

Another important property of crossed products is \emph{Takai
 duality}~\cite[Section~7.1]{Williams2007}. If $\sfG$ is a~locally compact abelian group and
$(\alg,\sfG,\alpha)$ is a $C^*$-dynamical system, then
$(\alg\rtimes_\alpha\sfG,\widehat\sfG,\widehat\alpha)$ is a~$C^*$-dynamical system, where
\[
\widehat\alpha\colon \ \widehat\sfG\longrightarrow\Aut(\alg\rtimes_\alpha\sfG)
 , \qquad \chi\longmapsto\widehat\alpha_\chi
\]
is defined by $\widehat\alpha_\chi(f)(\gamma):=\overline{\chi(\gamma)}
\, f(\gamma)$ for all $f\in C_{\rm c}(\sfG,\alg)$.
\begin{Theorem}[Takai duality] \label{thm:Takaiduality}
Let $(\alg,\sfG,\alpha)$ be a $C^*$-dynamical system where $\sfG$ is a
locally compact abelian group. Then there is an isomorphism of
$C^*$-algebras
\[
(\alg\rtimes_\alpha\sfG)\rtimes_{\widehat\alpha}\widehat\sfG \simeq
 \alg\otimes \CK\big(\Leb^2(\sfG)\big) .
\]
\end{Theorem}

\subsection{Morita equivalence and Green's theorem}\label{sec:Morita}

Crossed products of algebras provide a host of examples of dualities
which come in the form of various levels of strong and weak
equivalences of algebras, see, e.g.,~\cite{Brodzki2006}. The most
primitive form of such dualities is provided by (strong) Morita
equivalence~\cite{Rieffel1974}. A bimodule for a pair of algebras $\alg$ and $\balg$ is a
vector space $\Mcal$ which is simultaneously
 a left $\alg$-module and a right $\balg$-module, where the left action
 of $\alg$ commutes with the right action of $\balg$: $(a\cdot\xi)\cdot
 b=a\cdot (\xi\cdot b)$ for all $a\in \alg$, $b\in
 \balg$ and $\xi\in \Mcal$. If $\alg$ and $\balg$ are $C^*$-algebras,
 we say that a bimodule $\Mcal$ is an $\alg$--$\balg$ \emph{Morita
 equivalence bimodule} (or \emph{imprimitivity bimodule}) if it
is equipped with an $\alg$-valued inner product
${}_\alg\langle\,\cdot\,,\,\cdot\,\rangle$ and a $\balg$-valued inner
product $\langle\,\cdot\,,\,\cdot\,\rangle_\balg$ satisfying the
associativity condition
\[
{}_\alg\langle \psi,\phi\rangle \cdot \xi =
\psi\cdot\langle\phi,\xi\rangle_\balg ,
\]
for all $\psi,\phi,\xi\in\Mcal$, under which $\Mcal$ is complete in
the norm closures, and such that the ideal ${}_\alg\langle
\Mcal,\Mcal\rangle$ is dense in $\alg$ and $\langle
\Mcal,\Mcal\rangle_\balg$ is dense in $\balg$. The bimodule $\Mcal$ establishes a
\emph{Morita equivalence} between the algebras $\alg$ and $\balg$, and
in this case we write $\alg\ \mbf\sim_{\text{\tiny M}}\ \balg$.

Morita equivalent $C^*$-algebras have equivalent categories of
nondegenerate $*$-representations: If $\Pim_\balg\colon \balg\to{\rm B}(\hil_\balg)$ is a
representation of $\balg$ on a Hilbert space $\hil_\balg$, then we can
construct another Hilbert space
\[
\hil_\alg:=\Mcal\otimes_\balg\hil_\balg
\]
which is the quotient of the tensor product $\Mcal\otimes\hil_\balg$ by the relation $(\xi\cdot
b)\otimes\psi-\xi\otimes\Pim_\balg(b)\psi=0$ identifying the
$\balg$-actions for
$\xi\in\Mcal$, $b\in\balg$ and $\psi\in\hil_\balg$. The inner product
on $\hil_\alg$ is given by
\[
\big\langle\xi\otimes_\balg\psi\big|
\xi'\otimes_\balg\psi'\big\rangle_{\hil_\alg} :=
\big\langle\psi\big|\Pim_\balg\big(\langle\xi,\xi'\rangle_\balg
\big)\psi'\big\rangle_{\hil_\balg} ,
\]
and a representation $\Pim_\alg\colon \alg\to\rmB(\hil_\alg)$ of the algebra
$\alg$ is defined by
\[
\Pim_\alg(a)(\xi\otimes_\balg\psi) = (a\cdot\xi) \otimes_\balg\psi
\]
for $a\in\alg$ and $\xi\otimes_\balg\psi\in\hil_\alg$; this
representation is unitary equivalent to the representation
$\Pim_\balg$. Conversely, starting with a representation of $\alg$, we
can use a conjugate $\balg$--$\alg$ equivalence bimodule
$\overline{\Mcal}$ to construct a unitary equivalent representation of
$\balg$; then there are surjective bimodule homomorphisms
$\Mcal\otimes_\balg\overline{\Mcal}\to\alg$ and
$\overline{\Mcal}\otimes_\alg\Mcal\to\balg$ which satisfy a certain
transitivity law. As a particular consequence of this equivalence, Morita
equivalent algebras have homeomorphic spectra and isomorphic K-theory groups.

\begin{Example}[noncommutative two-tori]\label{ex:NCT2Morita}
A famous example of Morita equivalence in both mathematics and string
theory is provided by the
noncommutative tori ${\sf A}_\theta=\torus_\theta^2$ from
Example~\ref{ex:NCtori}. Firstly, notice from \eqref{eq:C*UV} that
changing the coset representative $\theta\in\real/\Z$ yields an identical
algebra: $\A_{\theta+m}=\A_{\theta}$ for all $m\in\Z$. Secondly, there
is an obvious
$C^*$-algebra isomorphism $\A_{-\theta}\simeq\A_\theta$ obtained by
interchanging the two generators $U$ and $V$. The converse
is also true~\cite{Rieffel1981,Rieffel1982}:
$\A_{\theta'}\simeq\A_{\theta}$ if and only if $\theta'=\theta \ {\rm
 mod}~1$. More generally, two irrational rotation $C^*$-algebras
$\A_\theta$ and $\A_{\theta'}$ are
Morita equivalent if and only if $\theta$ and $\theta'$ lie in the
same orbit under the action of $\sfGL(2,\Z)$ by fractional linear
transformations
\[
\theta'={\tt M}[\theta]:=\frac{a \theta+b}{c \theta+d} \qquad
\mbox{for} \quad {\tt M}=\bigg(\begin{matrix} a & b \\ c &
 d \end{matrix}\bigg) \in \sfGL(2,\Z) .
\]
The explicit Morita equivalence bimodules can be found
in~\cite{Rieffel1981}. On the other hand, the rational rotation
algebras $\A_\theta$ are all Morita equivalent to the commutative
algebra $C\big(\torus^2\big)$ of functions on the
two-torus~\cite{Rieffel1982}.
\end{Example}

\begin{Example}[noncommutative $d$-tori]
\label{ex:NCTdMorita}
The Morita equivalences of Example~\ref{ex:NCT2Morita} generalize to
the higher-dimensional noncommutative tori $\A_\Theta=\torus^d_\Theta$
from Example~\ref{ex:NCtorid} in the following
way~\cite{Rieffel1998}. First of all, the algebra $\A_\Theta$ is
unchanged if the matrix $\Theta$ is written in another basis of~$\Z^d$: if $B\in\sfGL(d,\Z)$ with transpose $B^{\rm t}$, then there is
a $C^*$-algebra isomorphism
$\A_{B^{\rm t}\,\Theta\,B}\simeq\A_\Theta$. More generally, consider the set of real
skew-symmetric $d{\times}d$ matrices $\Theta$ whose orbits
$\mbf M[\Theta]$ are defined for all $\mbf M\in{\sf SO}(d,d;\Z)$, where
\[
\mbf M[\Theta] = (A\,\Theta+B)\,(C\,\Theta+D)^{-1} \qquad \mbox{for} \quad
\mbf M=\left(\begin{matrix} A & B \\ C & D \end{matrix}\right) \in
{\sf SO}(d,d;\Z) ,
\]
and $A$, $B$, $C$ and $D$ are $d{\times}d$ block matrices satisfying
\[
A^{\rm t} C+C^{\rm t} A = 0 = B^{\rm t} D+D^{\rm t} B \qquad
\mbox{and} \qquad A^{\rm t} D+C^{\rm t} B=\unit_d .
\]
The set of all such matrices is dense in the space of all
skew-symmetric real $d{\times}d$ matrices, and there is a Morita
equivalence
\[
\A_{\mbf M[\Theta]} \ \mbf\sim_{\text{\tiny M}} \ \A_\Theta .
\]
However, for $d>2$ the converse is not generally true: In fact, there
are algebras $\A_\Theta$ and $\A_{\Theta'}$ that are isomorphic (and
so Morita equivalent) but for which the matrices $\Theta$ and
$\Theta'$ do not belong to the same ${\sf SO}(d,d;\Z)$
orbit~\cite{Rieffel1998}.
\end{Example}

We will also need an equivariant version of Morita equivalence in
order to show that Morita equivalent algebras induce Morita
equivalent crossed products according
to~\cite[Section~5.4]{Echterhoff2017}
\begin{Theorem}\label{equivMorita}
Let $(\alg,\sfG, \alpha)$ and $(\balg,\sfG, \beta)$ be
$C^*$-dynamical systems such that
$\alg$ and $\balg$ are Morita equivalent. Then the crossed product $C^*$-algebras
$\alg\rtimes_\alpha \sfG$ and $\balg\rtimes_\beta \sfG$ are Morita equivalent if
there exists a $\sfG$-equivariant $\alg$--$\balg$ Morita equivalence bimodule $\Mcal$,
i.e., if there is a strongly continuous action $U\colon \sfG\to \Aut(\Mcal)$ of $\sfG$ on an
$\alg$--$\balg$ Morita equivalence bimodule $\Mcal$ such that
\[
U_\gamma(a\cdot \xi)=\alpha_\gamma(a)\cdot
U_\gamma(\xi) \qquad \mbox{and} \qquad U_\gamma(\xi\cdot b)=U_\gamma(\xi)\cdot
\beta_\gamma(b) ,
\]
and
\[
{}_{\alg}\langle U_\gamma(\xi),U_\gamma(\xi')\rangle=\alpha_\gamma\big({}_\alg\langle
\xi,\xi'\rangle\big) \qquad \mbox{and} \qquad \langle
U_\gamma(\xi),U_\gamma(\xi')\rangle_\balg =\beta_\gamma\big(\langle
\xi,\xi'\rangle_\balg\big) ,
\]
for
all $\gamma\in \sfG$, $\xi,\xi'\in \Mcal$, $a\in \alg$ and $b\in \balg$.
\end{Theorem}

In this paper, our main application of Morita equivalence will involve
Green's symmetric imprimitivity theorem. Let $X$ be a locally compact space, and let
$\sfH $ and $\sfK $ be locally compact groups with commuting free and
proper actions on the right and on the left on $X$, respectively. We can lift these
actions to left actions on $C_0(X)$ by defining
$({}^h\mathtt{f})(x)=\mathtt{f}\big(h^{-1}\cdot x\big)$ and $\big({}^k\mathtt{f}\big)(x)=\mathtt{f}(x\cdot k)$
for all $\mathtt{f}\in C_0(X)$, $x\in X$, $h\in \sfH$ and $k\in \sfK $. Commutativity
of the actions of~$\sfH $ and~$\sfK $ implies that there are well-defined
induced actions of~$\sfH $ and~$\sfK $
respectively on the quotient spaces $\sfK {\setminus} X$ and $X/\sfH $, and hence
respectively on the algebras $C_0(\sfK {\setminus} X)$ and
$C_0(X/\sfH )$ which we denote~$\rt$ and~$\lt$. Green's symmetric imprimitivity theorem then reads as~\cite[Corollary~4.10]{Williams2007}
\begin{Theorem} \label{thm:Green}
There is a Morita
equivalence of $C^*$-algebras
\begin{gather}\label{eq:Greenimprim}
C_0(\sfK {\setminus} X)\rtimes_\rt \sfH \ \mbf\sim_{\text{\rm \tiny M}} \ C_0(X/\sfH )\rtimes_\lt \sfK
\end{gather}
implemented by the Morita equivalence $($or imprimitivity$)$
bimodule $\Mcal$ which is the completion of $C_{\rm c}(X)$ with the
actions
\begin{gather*}
(a\cdot\xi)(x) =\int_\sfK\,a(k,x\cdot\sfH) \xi\big(k^{-1}\cdot x\big)
\Delta_\sfK(k)^{1/2} \, \dd\mu_\sfK(k) , \\
(\xi\cdot b)(x) =\int_\sfH \xi\big(x\cdot h^{-1}\big) b\big(h,\sfK\cdot x\cdot h^{-1}\big) \Delta_\sfH(h)^{-1/2} \, \dd\mu_\sfH(h) ,
\end{gather*}
for all $x\in X$, $a\in C_{\rm c}(\sfK\times X/\sfH)$, $b\in C_{\rm
 c}(\sfH\times \sfK {\setminus} X)$ and $\xi\in C_{\rm
 c}(X)$, and the inner products
\begin{gather*}
{}_\alg\langle\xi,\xi'\rangle(k,x\cdot\sfH) = \Delta_\sfK(k)^{-1/2}
\int_\sfH \xi(x\cdot h) \overline{\xi'\big(k^{-1}\cdot x\cdot h\big)} \,
\dd\mu_\sfH(h) , \nonumber\\
\langle\xi,\xi'\rangle_\balg(h,\sfK\cdot x) =
\Delta_\sfH(h)^{-1/2} \int_\sfK \overline{\xi\big(k^{-1}\cdot x\big)}
\xi'\big(k^{-1}\cdot x\cdot h\big) \, \dd\mu_\sfK(k)  ,
\end{gather*}
for all $x\in X$, $h\in\sfH$, $k\in\sfK$ and $\xi,\xi'\in C_{\rm c}(X)$.
\end{Theorem}

Theorem~\ref{thm:Green} has several useful applications and
corollaries, see, e.g.,~\cite{Rieffel1976}. A particularly relevant
special case that we shall use below is when $\sfK $ is the trivial group,
in which case~\eqref{eq:Greenimprim} reduces to the Morita equivalence
\begin{gather*}
C_0(X)\rtimes_\rt \sfH \ \mbf\sim_{\text{\tiny M}} \ C_0(X/\sfH ) ,
\end{gather*}
illustrating the use of crossed products in describing
quotients. This equivalence can in fact be strengthened to a stable
isomorphism~\cite{Rieffel1976}
\begin{gather}\label{eq:XsfHiso}
C_0(X)\rtimes_\rt \sfH \simeq C_0(X/\sfH ) \otimes \CK\big(\Leb^2(\sfH)\big) ,
\end{gather}
where $\CK$ denotes the algebra of
compact operators.
\begin{Example}[tori]\label{ex:Moritatori}
A particularly relevant example for us is the case $X=\real^d$ with
$\sfH=\zed^d$ acting by translations $(n,x)\mapsto x+n$ for
$n\in\zed^d$ and $x\in\real^d$, which realizes the
$d$-dimensional torus $\torus^d=\real^d/\zed^d$ as a crossed product:
\[
C\big(\torus^d\big) \ \mbf\sim_{\text{\tiny M}} \ C_0\big(\real^d\big)\rtimes_\rt \zed^d .
\]
Let us illustrate the construction explicitly. The convolution algebra
$C_{\rm c}\big(\Z^d\times\real^d\big)\!\subset\! C_0\big(\R^d\big)\rtimes_\rt\Z^d$ can be
identified with the space of sequences $f=\{f_n\}_{n\in\Z^d}$ of
functions $f_n\colon \real^d\to\complex$ with the convolution product
\[
(f\star g)_n(x) = \sum_{m\in\Z^d} f_m(x) g_{n-m}(x-m) .
\]

Consider the algebra
\[
\alg:=\big\{f\in C_0\big(\real^d,\CK\big(\ell^2\big(\zed^d\big)\big)\big) \, \big| \, f(x+m)
= U_m f(x) U_{m}^{-1}\big\} ,
\]
where $U_m$ is the unitary shift operator on $\ell^2\big(\zed^d\big)$ defined by
$(U_ma)_n=a_{n-m}$ for each $m\in\Z^d$ and $a=\{a_n\}_{n\in\Z^d}$.
Define a map ${\mit\Phi}\colon C_{\rm c}\big(\Z^d\times\real^d\big)\to
C_{\rm c}\big(\real^d,\CK\big(\ell^2\big(\Z^d\big)\big)\big)$ by
\[
\big({\mit\Phi}(f)(x)\big)_{mn} = f_{m+n}(x+n) .
\]
It is easy to see that ${\mit\Phi}(f)(x+m) =
U_m {\mit\Phi}(f)(x) U_m^{-1}$ for all $f\in C_{\rm
 c}\big(\zed^d\times\real^d\big)$, and if $a=(a_{mn})\in\alg$ then defining
$f_n:=a_{0n}$ gives ${\mit\Phi}(f)=a$. It is also easy to check that
${\mit\Phi}(f\star g) = {\mit\Phi}(f) {\mit\Phi}(g)$, and consequently
$\mit\Phi$ gives an algebra isomorphism
\[
{\mit\Phi}\colon \ C_0\big(\real^d\big)\rtimes_\rt\zed^d \xrightarrow{ \ \simeq \ } \alg .
\]

The explicit Morita equivalence bimodule is now obtained from the
completion of
\[
\Mcal = \big\{\xi\in C_{\rm c}\big(\real^d,\ell^2\big(\Z^d\big)\big) \, \big| \, \xi(x+m) =
U_m \xi(x) \big\} .
\]
The left action of the algebra $\alg={\mit\Phi}\big(C_{\rm c}\big(\Z^d\times\real^d\big)\big)$ is by left
matrix multiplication on $\Mcal$:
\begin{align*}
\alg\times\Mcal &\longrightarrow\Mcal , \\
(a,\xi) & \longmapsto a\cdot\xi , \qquad (a\cdot\xi)_n=\sum_{m\in\Z^d}
a_{nm} \xi_m ,
\end{align*}
while the right action of the algebra $C\big(\torus^d\big)$ is by right pointwise
multiplication on~$\Mcal$:
\begin{align*}
\Mcal\times C\big(\torus^d\big) & \longrightarrow\Mcal , \nonumber \\
(\xi,b) & \longmapsto \xi\cdot b , \qquad (\xi\cdot b)_n=\xi_n b .
\end{align*}
The left and right inner
products are respectively given by
\begin{gather*}
{}_\alg\langle\xi,\eta\rangle = \xi\otimes\eta^* , \\
\langle\xi,\eta\rangle_{C(\torus^d)} = \sum_{n\in\zed^d}
\overline{\xi_n} \eta_n ,
\end{gather*}
for $\xi,\eta\in\Mcal$.
Together with the isomorphism $\mit\Phi$, this establishes a
Morita equivalence between the algebras $C_0\big(\real^d\big)\rtimes_\rt\zed^d$
and $C\big(\torus^d\big)$.
\end{Example}

Another important class of examples is provided by taking
$X=\sfG$ to be a locally
compact group with closed subgroups $\sfK$ and $\sfH$ acting respectively by
left and right multiplication on~$\sfG$.
In particular, in the special case $\sfK =\sfG$, so that
$C_0(\sfK {\setminus} \sfG)=\complex$, Theorem~\ref{thm:Green} gives a Morita equivalence between the commutative dynamical system
$(C_0(\sfG/\sfH ), \sfG, \lt)$, with $\sfK =\sfG$ acting by left multiplication on the
homogeneous space $\sfG/\sfH $ (so that $(\lt_\gamma f)(x)=f\big(\gamma^{-1} x\big)$ for all $f\in
C_0(\sfG/\sfH )$, $x\in \sfG/\sfH $ and $\gamma\in \sfG$),
and the group $C^*$-algebra $C^*(\sfH )=\complex\rtimes \sfH $:
\[
C_0(\sfG/\sfH )\rtimes_\lt \sfG \ \mbf\sim_{\text{\tiny M}} \ C^*(\sfH ) .
\]
This equivalence also follows from the $C^*$-algebra isomorphism~\cite[Theorem~4.29]{Williams2007}
\[ 
C_0(\sfG/\sfH )\rtimes_\lt \sfG \simeq C^*(\sfH )\otimes
\CK\big(\Leb^2(\sfG/\sfH )\big) .
\]

\section{Topological T-duality and twisted tori}\label{sec:TopTtorus}

In this section we shall apply the results of Section~\ref{sec:NAdef}, and in particular
Green's theorem, to present a scheme
that will be employed in our study of topological T-duality. We shall
then illustrate how our scheme works to reproduce some standard
(commutative) examples of T-dual spaces.

\subsection{Twisted tori and their T-duals}\label{sec:twistedtori}

We are interested in formulating a notion of T-duality for
``torus bundles without $H$-flux'', which for our purposes can be
characterised by the following general class of spaces.

\begin{Definition}[twisted tori]
Let $\sfG$ be a locally compact group which admits a
cocompact discrete subgroup $\Lambda_\sfG$, i.e., a lattice in~$\sfG$, which we let act on $\sfG$ by left
multiplication. The quotient space
\[
\torus_{\Lambda_\sfG}:=\Lambda_\sfG\backslash \sfG
\]
is a \emph{twisted torus}.
\end{Definition}

By~\cite[Lemma~6.2]{Milnor1976}, only unimodular groups can contain lattices, i.e., groups $\sfG$ whose
modular function $\Delta_\sfG$ is identically equal to~$1$.

\begin{Example}[tori]\label{ex:tori}
Every lattice in the abelian Lie group $\sfG=\real^d$ is
isomorphic to $\Lambda_\sfG=\zed^d$, acting by translations. Then
$\torus_{\zed^d}=\zed^d\backslash\real^d=:\torus^d$ is the $d$-dimensional torus.
\end{Example}

\begin{Example}[nilmanifolds]\label{ex:nilmanifolds}
Generalizing Example~\ref{ex:tori}, let $\sfG$ be a connected and simply-connected nilpotent Lie group. Then a theorem of
Malcev~\cite[Theorem~2.12]{Raghunathan1972} establishes the existence
of a lattice $\Lambda_\sfG$ in $\sfG$ if and only if
$\sfG$ can be defined over the rationals,
i.e., there exists a basis for its Lie algebra which has rational
structure constants, and in this case $\torus_{\Lambda_\sfG}=\Lambda_\sfG\backslash \sfG$ is a nilmanifold.
\end{Example}

\begin{Example}[orbifolds]
Let $\sfG$ be a compact Lie group. Then the lattices in $\sfG$ are
precisely the finite subgroups $\Gamma$ of $\sfG$, and
$\torus_\Gamma=\Gamma\backslash\sfG$ is a smooth orbifold. For instance,
for $\sfG=\SU(2)$ the twisted tori are precisely the three-dimensional
ADE orbifolds $\torus_\Gamma=\Gamma\backslash\Sphere^3$
of the three-sphere for a finite subgroup $\Gamma\subset\SU(2)$;
for $\Gamma=\zed_n$ a cyclic subgroup of order $n\geq2$, this recovers
the familiar lens spaces $\torus_{\zed_n}=\zed_n\backslash\Sphere^3=:\mathbb{L}(n,1)$.
\end{Example}

Following~\cite{Mathai2004}, we come now to a central concept of this paper.

\begin{Definition}[topological T-duality]\label{def:Tduality}
Let $\torus_{\Lambda_\sfG}$ be a twisted torus which admits a
non-trivial right action
of the abelian Lie group $\real^n$ for some $n\geq1$. The crossed product
\[
C(\torus_{\Lambda_\sfG})\rtimes_\rt \real^n
\]
is a \emph{$C^*$-algebraic T-dual} of the twisted torus.
\end{Definition}
If the spectrum of the crossed product algebra $
C(\torus_{\Lambda_\sfG})\rtimes_\rt \real^n
$ is a Hausdorff topological space~$X$ (for instance
if it is Morita
equivalent to a commutative $C^*$-algebra $C(X)$), then we say that~$X$ is T-dual to the
twisted torus $\torus_{\Lambda_\sfG}$ and call $X$ a `classical T-dual'; otherwise we say that the T-dual of
$\torus_{\Lambda_\sfG}$ is a noncommutative space.

For Definition \ref{def:Tduality} to be a `good' notion of T-duality, we should first explain
\begin{itemize}\itemsep=0pt
\item[(a)] in
what precise sense $\torus_{\Lambda_\sfG}$ and
$C(\torus_{\Lambda_\sfG})\rtimes_\rt\real^n$ are `equivalent', and
\item[(b)] how
T-duality applied twice returns the original twisted torus.
\end{itemize}
The answers to both of these points turns out to be provided by
working in a suitable category tailored to our treatment of
topological T-duality.

\subsection[T-duality in the category $\CCK\hspace{-1mm}\CCK$]{T-duality in the category $\CCK\hspace{-1mm}\CCK$}
\label{sec:TdualityKK}

The terminology `topological T-duality' refers to a coarse equivalence at
the level of topology; for $C^*$-algebras the topology is measured by
K-theory. A more powerful refinement is provided by Kasparov's bivariant
K-theory which constructs groups ${\rm KK}(\alg,\balg)$ for any pair
of separable $C^*$-algebras $\alg$ and $\balg$; when $\alg=\complex$,
the group ${\rm KK}(\complex,\balg)\simeq {\rm
 K}(\balg)$ is the K-theory group of $\balg$. The cycles in
Kasparov's groups ${\rm KK}(\alg,\balg)$, called \emph{Kasparaov
 $\alg$--$\balg$ bimodules}, are triples $(\hil,\phi,T)$ where $\hil$
is a right Hilbert $\balg$-module, $\phi$ is a $*$-representation of $\alg$
on $\hil$, and $T\in\End_\balg(\hil)$ is a $\balg$-linear operator on
$\hil$, subject to certain compactness conditions; we do not
provide further details of the definition here and instead refer to~\cite{Brodzki2006} for a
concise review of KK-theory in the context that we shall use it
in this paper. Kasparov bimodules may be thought of as generalizations
of morphisms between $C^*$-algebras, in the sense that any algebra homomorphism $\phi\colon \alg\to\balg$
determines a class $[\phi]\in{\rm KK}(\alg,\balg)$, represented by the
$\alg$--$\balg$-bimodule $(\balg,\phi,0)$.

A key feature of Kasparov's KK-theory is the composition product
\[
\otimes_\balg\colon \ {\rm KK}(\alg,\balg)\times{\rm
 KK}(\balg,\Ccal)\longrightarrow{\rm KK}(\alg,\Ccal) ,
\]
which is bilinear and associative. This product is compatible with the
composition of morphisms $\phi\colon \alg\to\balg$ and $\psi\colon \balg\to\Ccal$
of $C^*$-algebras: $[\phi]\otimes_\balg[\psi] = [\psi\circ\phi]$. It
also makes ${\rm KK}(\alg,\alg)$ into a~ring with unit
$1_\alg=[\Id_\alg]$. We say that an element $\alpha\in{\rm
 KK}(\alg,\balg)$ is \emph{invertible} if there exists an element
$\beta\in{\rm KK}(\balg,\alg)$ such that
$\alpha\otimes_\balg\beta=1_\alg$ and
$\beta\otimes_\alg\alpha=1_\balg$.

An important special instance of Kasparov
bimodules comes from Morita equivalence: Any Morita equivalence
$\alg$--$\balg$ bimodule $\Mcal$ is also a Kasparov bimodule
$(\Mcal,\phi,0)$, with $\phi\colon \alg\to\End(\Mcal)$ the left action of
$\alg$, which defines an invertible class $[\Mcal]\in{\rm
 KK}(\alg,\balg)$ with inverse $[\,\overline{\Mcal}\,]\in{\rm
 KK}(\balg,\alg)$ given by the conjugate $\balg$--$\alg$ bimodule
$\overline{\Mcal}$. Generally, if there exists an invertible element
$\alpha\in{\rm KK}(\alg,\balg)$, then the algebras $\alg$ and $\balg$
are said to be \emph{{KK}-equivalent}, and we write $\alg\
\mbf\sim_{\text{\tiny{KK}}}\ \balg$. Thus Morita equivalence implies
KK-equivalence, but the converse is not generally true. KK-equivalent algebras have isomorphic
K-theory groups, but not necessarily homeomorphic spectra.

This refinement naturally
suggests an approach to T-duality where the category of separable
$C^*$-algebras with $*$-homomorphisms is replaced with an additive
category ${\CCK\hspace{-1mm}\CCK}$, whose objects are again separable $C^*$-algebras
but whose morphisms between any two objects $\alg$ and $\balg$ are
given by the classes in ${\rm KK}(\alg,\balg)$ (see,
e.g.,~\cite{Brodzki2007}). The composition product defines the
composition law, and isomorphic algebras
in this category are precisely the KK-equivalent
algebras; in particular, Morita equivalent algebras are isomorphic
as objects in ${\CCK\hspace{-1mm}\CCK}$. Our formulation and computations of topological
T-duality will always take place in this category, and in this
setting we can easily provide answers to points~(a) and~(b) below Definition~\ref{def:Tduality}
through
\begin{Theorem}\label{thm:equivalences}
If $\torus_{\Lambda_\sfG}$ is a twisted torus with a
non-trivial right action
of $\real^n$, then there are
isomorphisms in the category ${\CCK\hspace{-1mm}\CCK}$ given by the equivalences
\begin{itemize}\itemsep=0pt
\item[{\rm (a)}]  $C(\torus_{\Lambda_\sfG}) \
 \mbf\sim_{\text{\tiny{KK}}} \
 C(\torus_{\Lambda_\sfG})\rtimes_\rt\real^n  $ $($up to a shift of degree
 $n \ {\rm mod}~2)$, and
\item[{\rm (b)}]  $\big(C(\torus_{\Lambda_\sfG})\rtimes_\rt
 \real^n\big)\rtimes_{\widehat\rt} \real^n \ \mbf\sim_{\text{\tiny
 M}} \ C(\torus_{\Lambda_\sfG})$.
\end{itemize}
\end{Theorem}
\begin{proof}
The KK-equivalence (a) follows from the Connes--Thom isomorphism,
formulated in the language of KK-theory~\cite{Fack1981}. The Morita
equivalence~(b) follows from Takai duality
(Theorem~\ref{thm:Takaiduality}).
\end{proof}

Another virtue of the categorical setting is that it enables a
general algebraic reformulation of the correspondence space
construction, which for topological spaces `geometrizes' the action of
topological T-duality. In~\cite[Proposition~5.3]{Brodzki2007} it is proven that, if
$\alg$ and $\balg$ are separable $C^*$-algebras, then any class in
${\rm KK}(\alg,\balg)$ can be represented by a `noncommutative
correspondence'. For this, we first recall,
 following~\cite{Brodzki2006,Brodzki2007}, that KK-theory provides a
definition of Gysin or ``wrong way'' homomorphisms on K-theory for
$C^*$-algebras. If $\phi:\alg\to\balg$ is a
morphism of separable $C^*$-algebras, a \emph{K-orientation} is a
functorial assignment of a corresponding element $\phi!\in{\rm
 KK}(\balg,\alg)$. If a K-orientation exists, we say that $\phi$ is
\emph{K-oriented} and call $\phi!$ the associated \emph{Gysin
 element}. The Gysin
homomorphism on K-theory is now defined by $\phi_!:=(-)\otimes_\balg\phi!\colon {\rm
 K}(\balg)\to{\rm K}(\alg)$. We then slightly adapt
the definition from~\cite{Brodzki2007} to the present context of Theorem~\ref{thm:equivalences}.
\begin{Definition}[noncommutative correspondences]
\label{def:NCcorr}
Let $\torus_{\Lambda_\sfG}$ be a twisted torus which admits a
non-trivial right action
of $\real^n$, and let
\begin{gather}\label{eq:NCcorrdiag}\begin{split}&
\xymatrix{
 & \ \Ccal \ & \\
C(\torus_{\Lambda_\sfG}) \ar[ur]^{[\phi]} & &
C(\torus_{\Lambda_\sfG})\rtimes_\rt\real^n
\ar[ul]_{[\psi]}
}\end{split}
\end{gather}
be a diagram in $\CCK\hspace{-1mm}\CCK$ whose arrows are induced by
homomorphisms $\phi\colon C(\torus_{\Lambda_\sfG})\to\Ccal$ and
$\psi\colon C(\torus_{\Lambda_\sfG})\rtimes_\rt\real^n\to\Ccal$ of separable
$C^*$-algebras. Assume that $\psi$ is K-oriented, and let $\psi!\in{\rm KK}\big(\Ccal,C(\torus_{\Lambda_\sfG})
\rtimes_\rt\real^n\big)$ be its corresponding Gysin element. The separable $C^*$-algebra $\Ccal$
is a~\emph{noncommutative correspondence} if the associated element
\[
[\phi]\otimes_\Ccal \psi! \in {\rm KK}\big(C(\torus_{\Lambda_\sfG}) ,
C(\torus_{\Lambda_\sfG}) \rtimes_\rt\real^n\big)
\]
is a KK-equivalence between the twisted torus and its
$C^*$-algebraic T-dual.
\end{Definition}Analogously to~\cite{Bouwknegt2008}, we obtain a noncommutative
correspondence by restricting the $\real^n$-action to the lattice
$\zed^n\subset\real^n$.
\begin{Proposition}\label{prop:NCcorr}
The crossed product
\[
\Ccal=C(\torus_{\Lambda_\sfG})\rtimes_{\rt|_{\zed^n}} \zed^n
\]
is a {noncommutative correspondence} in the sense of
Definition~\ref{def:NCcorr}.
\end{Proposition}
\begin{proof}
We need to construct a diagram \eqref{eq:NCcorrdiag} in $\CCK\hspace{-1mm}\CCK$ for the
crossed product. For this,
note that for any dynamical system of the form $(\alg,\Lambda,\alpha)$
where $\Lambda$ is a discrete group, there is a
natural injection $j$ of the algebra $\alg$ into the crossed product
$\alg\rtimes_\alpha\Lambda$: given $a\in\alg$, define the sequence
$j(a)\in C_{\rm c}(\Lambda,\alg)$ by $j(a)_{\gamma}=a \delta_{\gamma,e}$ for $\gamma\in\Lambda$. It is
easy to check, using the explicit formula for the convolution product,
that the map $a\mapsto j(a)$ is an algebra
monomorphism:
$j(a)\star j(b)=j(a\,b)$ for $a,b\in\alg$. In particular,
there is a $C^*$-algebra injection
\begin{gather}\label{eq:jmono}
j\colon \ C(\torus_{\Lambda_\sfG})\longrightarrow
C(\torus_{\Lambda_\sfG})\rtimes_{\rt|_{\zed^n}} \zed^n .
\end{gather}

Next we apply~\cite[Corollary~2.8]{Green1980} with $\real^n$ acting on
$\torus^n=\real^n/\Z^n$ by (right) translation and the diagonal action of
$\real^n$ on $\torus^n\times\torus_{\Lambda_\sfG}$ to obtain an
isomorphism
\[
C\big(\torus^n\times\torus_{\Lambda_\sfG}\big)\rtimes_\rt\real^n \simeq
\big(C(\torus_{\Lambda_\sfG})\rtimes_{\rt|_{\Z^n}}\Z^n\big) \otimes
\CK\big(\Leb^2\big(\torus^n\big)\big) .
\]
The projection
$\torus^n\times\torus_{\Lambda_\sfG}\to\torus_{\Lambda_\sfG}$ induces an injection
$C_{\rm c}(\real^n\times\torus_{\Lambda_\sfG})\hookrightarrow
C_{\rm c}(\real^n\times\torus^n\times\torus_{\Lambda_{\sfG}})$ which preserves the
convolution product, and we obtain a $C^*$-algebra monomorphism
\[
\psi'\colon \ C(\torus_{\Lambda_\sfG})\rtimes_\rt\real^n\longrightarrow
\big(C(\torus_{\Lambda_\sfG})\rtimes_{\rt|_{\Z^n}}\Z^n\big) \otimes
\CK\big(\Leb^2(\torus^n)\big) ,
\]
which is easily checked to be K-oriented since it is induced by a
projection.

This gives algebra morphisms $\phi':=\iota\circ
j\colon C(\torus_{\Lambda_\sfG})\to\Ccal\otimes\CK$ and
$\psi'\colon C(\torus_{\Lambda_\sfG})\rtimes_\rt\real^n\to\Ccal\otimes\CK$,
where $\Ccal=C(\torus_{\Lambda_\sfG})\rtimes_{\rt|_{\zed^n}} \zed^n$, $\CK$ denotes the
$C^*$-algebra of compact operators on a
separable Hilbert space,
and $\iota\colon \Ccal\to\Ccal\otimes\CK$ is the usual stabilization map. Taking the composition products of $[\phi']$ and $[\psi']$ with the
Morita equivalence $\Ccal\otimes\CK \ \mbf\sim_{\text{\tiny M}} \
\Ccal$ then yields the required maps in~\eqref{eq:NCcorrdiag}.
\end{proof}

When the spectrum of the
$C^*$-algebraic T-dual is a Hausdorff space, we identify the
correspondence space with the spectrum of
$\Ccal=C(\torus_{\Lambda_\sfG})\rtimes_{\rt|_{\zed^n}}\zed^n$;
otherwise $\Ccal$ is a noncommutative space.

\subsection{Computational tools}

Let us now explain how to compute these $C^*$-algebraic T-duals in
some special instances that will appear throughout the remainder of
this paper. For certain actions of $\real$, we
may compute the $C^*$-algebraic T-dual via
\begin{Proposition}\label{prop:Raeburn}
Let $\torus_{\Lambda_\sfG}$ be a twisted torus equipped with an action of
$\real$ for which every point has isotropy subgroup $\zed$. Let
$T=\torus_{\Lambda_\sfG}/\real$, and denote the corresponding principal
circle bundle by $p\colon \torus_{\Lambda_\sfG}\to T$. Then the $C^*$-algebraic
T-dual $C(\torus_{\Lambda_\sfG})\rtimes_\rt \real\simeq \C\T(\torus\times T,\delta)$ is a continuous-trace
algebra with spectrum $\torus\times T$ and Dixmier--Douady class
$\delta = \zeta\smile c_1(p)\in \rmH^3(\torus\times T,\Z)$, where
$c_1(p) \in \rmH^2(T,\Z)$ is the Chern class of the circle bundle and $\zeta$ is the
standard generator of $\rmH^1(\torus,\Z)\simeq\Z$.
\end{Proposition}
\begin{proof}
This is just a straightforward adaptation of the statement
of~\cite [Proposition~4.5]{Raeburn1988}.
\end{proof}

In these instances, the T-dual of $\torus_{\Lambda_\sfG}$ is the Hausdorff space $X=\torus\times
(\torus_{\Lambda_\sfG}/\real)$ with a three-form `$H$-flux' whose cohomology class is
represented by $[H] = \zeta\smile c_1(p)$.

More generally, suppose that the $\real^n$-action on $\torus_{\Lambda_\sfG}$ is
induced by a free and proper right action of $\real^n$ on the covering group $\sfG$
which commutes with the left action of the lattice $\Lambda_\sfG$ on~$\sfG$. We can then apply Green's theorem
(Theorem~\ref{thm:Green}) to get the Morita
equivalence
\begin{gather}\label{scheme}
C(\torus_{\Lambda_\sfG})\rtimes_\rt \real^n \ \mbf\sim_{\text{\tiny M}} \
C_0\big(\sfG/\real^n\big)\rtimes_\lt \Lambda_\sfG .
\end{gather}
In this special case, we obtain an easy proof of
Proposition~\ref{prop:NCcorr}: The inclusion
$\Z^n\hookrightarrow\R^n$ of groups induces a monomorphism
\[
C_0\big(\sfG/\real^n\big)\rtimes_\lt \Lambda_\sfG \longrightarrow
C_0\big(\sfG/\Z^n\big)\rtimes_\lt \Lambda_\sfG \ \mbf\sim_{\text{\tiny M}}
 \ C(\torus_{\Lambda_\sfG})\rtimes_{\rt|_{\Z^n}} \Z^n ,
\]
where in the last step we replaced $\R^n$ by its subgroup $\Z^n$ in
\eqref{scheme}. This gives monomor\-phisms~\eqref{eq:jmono} and
$C(\torus_{\Lambda_\sfG})\rtimes_\rt \real^n\to
C(\torus_{\Lambda_\sfG})\rtimes_{\rt|_{\Z^n}} \Z^n$ in the category
$\CCK\hspace{-1mm}\CCK$.

\subsection{Topological T-duality for the torus}\label{sec:toruscrossedprod}

Let us now describe how our considerations reproduce the standard
T-duality for tori.
The simplest example of the T-duality scheme \eqref{scheme} is the case where $\sfG
=\real$, $\Lambda_\sfG =\zed\subset\real$, and $\sfG/\real=\{0\}$ with the obviously trivial $\Lambda_\sfG$-action. Then
$\torus_\zed= \zed\backslash\real=\torus$ is a circle, and~\eqref{scheme} with $n=1$ reads
\[
C(\torus)\rtimes_\rt \real \ \mbf\sim_{\text{\tiny M}} \ \complex\rtimes \zed
 = C^*(\zed) \simeq C\big(\widetilde\torus\big) ,
\]
where in the last passage we used the Fourier transform isomorphism
${\cal F}\colon C^*(\zed)\to C\big(\widetilde\torus\big)$; explicitly,
if $a=\{a_{n}\}_{n\in \zed}\in C^*(\zed)$, then
${\cal F}(a)(\chi)=\sum_{n\in \zed} a_{n} \e^{2\pi\ii n\,\chi}$ so that
 ${\cal F}(a)$ is a function on the
dual circle $\widetilde\torus=\real^*/\zed^*$.

The generalization to T-duality along a single direction $i$ of a
$d$-dimensional torus is
straightforward. Let $\Lambda\simeq\zed^d$ be the lattice in $\real^d$ given by
$\Lambda=\big\{\sum_{i=1}^d a_i \vec e_i \, | \, a_1,\dots,a_d \in
\zed\big\}$, where $\vec e_1,\dots,\vec e_d$ is the standard basis of~$\real^d$; this is the direct
sum $\Lambda=\bigoplus_{i=1}^d \zed\vec e_i$.
Let $\real_i$ be the subgroup of $\real^d$ linearly spanned by
$\vec e_i$ and let $\zed_i\subset\real_i$ be the corresponding
lattice; we write $\torus_i=\real_i/\zed_i$ and decompose the
$d$-torus $\torus^d=\real^d/\Lambda$ as
$\torus^d=\torus^{d-1}_{\hat\imath}\times\torus_i$, where
$\torus^{d-1}_{\hat\imath}$ is the $(d{-}1)$-dimensional torus defined by
omitting the $i$-th factors of $\real^d$ and $\Lambda$. Then $\real_i$
acts trivially on~$\torus^{d-1}_{\hat\imath}$ and we have
\begin{align}
C\big(\torus^d\big)\rtimes_\rt\real_i &=
C\big(\torus^{d-1}_{\hat\imath}\times \torus_{i}\big) \rtimes_\rt\real_i
\simeq \big(C\big(\torus^{d-1}_{\hat\imath}\big)\otimes
C(\torus_{i}) \big)\rtimes_{\Id\otimes\rt}\real_i\nn\\
&\simeq
C\big(\torus^{d-1}_{\hat\imath}\big)\otimes
\big(C(\torus_{i}) \rtimes_\rt\real_i \big)
 \mbf\sim_{\text{\tiny M}} \
C\big(\torus^{d-1}_{\hat\imath}\big)\otimes
C\big(\widetilde\torus_{i}\big) \nn\\
&\simeq
C\big(\torus^{d-1}_{\hat\imath}\times \widetilde\torus_{i}\big)
 = C\big(\torus^d_{\tilde\imath}\big) ,
\label{eq:Buscher}\end{align}
where $\torus^d_{\tilde\imath}:=\torus^{d-1}_{\hat\imath}\times
\widetilde\torus_{i}$.
This is the expected action of the $i$-th factorized T-duality, and in
this way we have thus
reproduced the standard rules for T-duality of tori. In fact, in this case we can use
Proposition~\ref{prop:Raeburn} to strengthen the statement of topological
T-duality: Every point of $\torus^d$ has isotropy group $\zed$ under
the action of $\real_i$, and the corresponding circle bundle
$p\colon \torus^d\to\torus^{d-1}_{\hat\imath}$ is trivial, so the
$C^*$-algebraic T-dual of $\torus^{d}$ is a continuous-trace algebra
with spectrum $\torus^d_{\tilde\imath}$ and trivial Dixmier--Douady
class. Hence the Morita equivalence in~\eqref{eq:Buscher} can be
replaced by a~stable isomorphism.

By iterating these T-duality transformations
one can perform T-dualities along multiple directions of a
$d$-dimensional torus. In particular, iterating the procedure $d$
times and using Theorem~\ref{cpsp3} we end up with the full T-duality
\begin{gather}\label{eq:fullTduality}
C\big(\torus^d\big)\rtimes_\rt\real^d \ \mbf\sim_{\text{\tiny M}} \
C\big(\widetilde\torus{}^d\big) ,
\end{gather}
where $\widetilde\torus{}^d=\big(\real^d\big)^*/\Lambda^*$ is the dual
torus with $\Lambda^*$ the dual lattice in the dual vector space $\big(\real^d\big)^*$.
This can also be obtained directly by setting $\sfG
=\real^d$, $\Lambda_\sfG=\Lambda\subset
\real^d$ and $\sfG/\real^d=\{0\}$ in~\eqref{scheme} with $n=d$, and by using the Fourier transforms in
all directions $\vec e_1,\dots,\vec e_d$.

Finally, let us consider the correspondence space construction. For
the $i$-th factorized T-duality, this is obtained by restricting the action of
$\real_i$ to the lattice $\zed_i\subset\real_i$. The action of the
group $\zed_i$ on the algebra of functions $C\big(\torus^d\big)$ is trivial and
so we get isomorphisms
\[
C\big(\torus^d\big)\rtimes_\rt\zed_i \simeq
C\big(\torus^d\big)\otimes C^*\big(\zed_i\big) \simeq
 C\big(\torus^d\big)\otimes C\big(\widetilde\torus_i\big) \simeq
 C\big(\torus^d\times\widetilde\torus_i\big) .
\]
This results in the noncommutative correspondence induced by the diagram
\[
\xymatrix{
 & C\big(\torus^d\times\widetilde\torus_i\big) & \\
C\big(\torus^d\big) \ar[ur]^{{\rm pr}^*} & &
C\big(\torus^d_{\tilde\imath}\big)
\ar[ul]_{\pi_i^*} \\
 & C\big(\torus^{d-1}_{\hat\imath}\big) \ar[ur]_{j} \ar[ul]^{j} &
}
\]
where ${\rm pr}\colon \torus^d\times\widetilde\torus_i\to\torus^d$ is the
projection to the first factor and
$\pi_i\colon \torus^d\times\widetilde\torus_i\to\torus^d_{\tilde\imath}$
omits the $i$-th factor of~$\torus^d$. The algebra inclusions $j$ of
$C\big(\torus^{d-1}_{\hat\imath}\big)$ are induced by the trivial
circle bundle projections $\torus^d\to \torus^{d-1}_{\hat\imath}$ and
$\torus^d_{\tilde\imath}\to \torus^{d-1}_{\hat\imath}$.

By either iterating this construction using Theorem~\ref{cpsp3} or by
direct calculation, the correspondence space for a full T-duality is
obtained by restricting the action of
$\real^d$ from~\eqref{eq:fullTduality} to the lattice
$\Lambda\subset\real^d$, and we analogously find
\begin{gather*}
C\big(\torus^d\big)\rtimes_\rt\Lambda \simeq
C\big(\torus^d\times\widetilde\torus{}^d\big) .
\end{gather*}
Thus the crossed product with the lattice of periods $\Lambda$
defining the $d$-torus $\torus^d$ recovers the doubled torus
$\torus^d\times\widetilde\torus^d$ which is the correspondence space
for the smooth Fourier--Mukai transform, wherein a full or factorized T-duality has a geometric interpretation as
an element of its automorphism group $\sfGL(2d,\zed)$.

\subsection{Topological T-duality for orbifolds}\label{sec:TopTorb}

Let $\sfG$ be a compact connected semisimple Lie group of rank~$r$,
and let $\Gamma\subset\sfG$ be a finite subgroup. The maximal torus
$\T = \sfU(1)^r \simeq \real^r/\zed^r$ of $\sfG$ carries a natural action of
$\real_i$ by translation along the $i$-th direction for $i=1,\dots,r$,
and we can apply a fibrewise T-duality to
the principal torus bundle $\sfG\to\sfG/\T$. Under this $\real$-action
every point of $\sfG$ has
isotropy subgroup $\zed$, and the action descends to the smooth orbifold
$\torus_\Gamma=\Gamma\backslash\sfG$. Then the quotient map
$p_i\colon \torus_\Gamma\to\torus_\Gamma/\real_i$ is the circle fibration
$\torus_\Gamma\to\torus_\Gamma/\sfU(1)_i$, where
$\T=\sfU(1)^{r-1}_{\hat\imath}\times \sfU(1)_i$, and by
Proposition~\ref{prop:Raeburn} the $C^*$-algebraic T-dual
\begin{gather}\label{eq:Tdualorbifold}
C(\Gamma\backslash\sfG)\rtimes_\rt \real_i \simeq
\C\T\big(\Gamma\backslash\sfG/\sfU(1)_i\times \widetilde\torus_i,\delta_i \big)
\end{gather}
of the orbifold $\torus_\Gamma$ is a continuous-trace algebra with
spectrum $\torus_\Gamma/\sfU(1)_i\times\widetilde\torus_i$ and Dixmier--Douady
class $\delta_i=c_1(p_i)\smile\zeta_i$.

In the rank one case, this T-duality is well-known (see, e.g.,~\cite{Bouwknegt2003}): Then
$\sfG=\sfSU(2)$ which we regard as the three-sphere $\Sphere^3$, and
for $\Gamma=\zed_n\subset\T=\sfU(1)$ the twisted torus is the lens
space $\mathbb{L}(n,1)$. The quotient map $p\colon \mathbb{L}(n,1)\to
\Sphere^2$ is a circle bundle whose Chern class $c_1(p)$ is equal to
$n$ times the standard generator of ${\rm H}^2\big(\Sphere^2,\zed\big)\simeq\zed$, and applying~\eqref{eq:Tdualorbifold} we
find that the
$C^*$-algebraic T-dual of $\mathbb{L}(n,1)$ is a continuous-trace
algebra whose spectrum is the trivial circle bundle
$\mathbb{L}(0,1)=\Sphere^2\times\widetilde\torus$ and whose Dixmier--Douady class $\delta$ is
$n$ times the standard generator of ${\rm
 H}^3\big(\Sphere^2\times\widetilde{\torus},\zed\big)\simeq \zed$.

Generally, the correspondence space construction is obtained by noting
that, since the iso\-tro\-py subgroup
for any point of the $\real_i$-action is $\zed_i\subset\real_i$, the group $\zed_i$
acts trivially on the algebra $C(\torus_\Gamma)$ and there are isomorphisms
\[
C(\torus_\Gamma)\rtimes_\rt\zed_i \simeq C(\torus_\Gamma)\otimes
C^*(\zed_i) \simeq C(\torus_\Gamma)\otimes C\big(\widetilde\torus_i\big)
\simeq C\big(\torus_\Gamma\times\widetilde\torus_i\big) .
\]
Let ${\rm pr}\colon \torus_\Gamma\times\widetilde\torus_i\to\torus_\Gamma$
be the projection to the first factor. Since ${\rm H}^2(\sfG,\zed)=0$,
K\"unneth's theorem implies
\[
\big(p_i\times\Id_{\widetilde\torus_i}\big)^*\big(c_1(p_i)\smile\zeta_i\big)=0 \in {\rm
 H}^3\big(\torus_\Gamma\times\widetilde\torus_i,\zed\big) ,
\]
and hence the algebra
$\C\T\big(\torus_\Gamma\times\widetilde\torus_i,(p_i\times\Id_{\widetilde\torus_i})^*\delta_i\big)$
is isomorphic to
$C\big(\torus_\Gamma\times\widetilde\torus_i\big)\otimes\CK$. Then there is the noncommutative correspondence
\[
\xymatrix{
 & C\big(\torus_\Gamma\times\widetilde\torus_i\big) & \\
C(\torus_\Gamma) \ar[ur]^{ \! \! \! \! [{\rm pr}^*]} & &
\C\T\big(\torus_\Gamma/\sfU(1)_i\times\widetilde\torus_i,\delta_i\big)
\ar[ul]_{ \ \ \
 [(p_i\times\Id_{\widetilde\torus_i})^*]} \\
 & C\big(\torus_\Gamma/\sfU(1)_i\big) \ar[ur]_{[j]}
 \ar[ul]^{[p^*]} &
}
\]
as a diagram in the category $\CCK\hspace{-1mm}\CCK$.

\section{Topological T-duality for almost abelian solvmanifolds}\label{sec:TdualityMostow}

A large class of twisted tori of interest as string compactifications
come in the form of fibrations over tori. These are the solvmanifolds
which are based on solvable groups $\sfG$ and generalize the
nilmanifolds discussed in Example~\ref{ex:nilmanifolds}. The fibrations
underlying these twisted tori are called {Mostow
 bundles}~\cite{Mostow1954}, and we are particularly interested in
the cases where the Mostow bundle is a torus bundle. A good source for
the material used in this section is~\cite{Bock2009} (see also~\cite{Console2016, OpreaTralle}).

\subsection{Mostow bundles}

Let $\sfG$ be
a connected and simply-connected solvable Lie group. Recall that its
nilradical $\sfN$ is the maximal connected nilpotent normal subgroup. It
has dimension $\dim\sfN\geq\frac12\dim\sfG$.

We first consider the case $\dim\sfN=\dim\sfG$. Then
$\sfN=\sfG$ and the group $\sfG$ is nilpotent. In this case, under the
conditions discussed in
Example~\ref{ex:nilmanifolds}, there exists a lattice
$\Lambda_\sfG\subset\sfG$ and the twisted torus $\torus_{\Lambda_\sfG}$ is a
nilmanifold. If $\sfG$ is abelian then $\torus_{\Lambda_\sfG}$ is a
torus. If $\sfG$ is non-abelian then there is a group extension
\[
1\longrightarrow [\sfG,\sfG]\longrightarrow\sfG\xrightarrow{ \ \pi \
}\sfG_{\rm ab}\longrightarrow 1
\]
of its commutator subgroup $[\sfG,\sfG]$,
and both $\Lambda_\sfG\cap[\sfG,\sfG]$ and $\pi(\Lambda_\sfG)$ are
lattices in the nilpotent Lie group $[\sfG,\sfG]$ and the
abelianization $\sfG_{\rm ab}:=[\sfG,\sfG]\backslash\sfG$ of $\sfG$, respectively~\cite{Corwin1990}. This exhibits
the twisted torus $\torus_{\Lambda_\sfG}=\Lambda_\sfG\backslash \sfG $ as a fibration over the torus
$\pi(\Lambda_\sfG)\backslash\sfG_{\rm ab}$ with
nilmanifold fibres,{\samepage
\[
\big(\Lambda_\sfG\cap[\sfG,\sfG]\big)\backslash[\sfG,\sfG] \longrightarrow \torus_{\Lambda_\sfG} \longrightarrow
\pi(\Lambda_\sfG)\backslash \sfG_{\rm ab} .
\]
If $[\sfG,\sfG]$ is an abelian Lie group then the twisted
torus is a torus bundle over a torus.}

Suppose now that the group $\sfG$ is not nilpotent. Then $\sfN\backslash\sfG$
is a non-trivial abelian Lie group. If $\sfG$ admits a lattice $\Lambda_\sfG$, then
$\Lambda_{\sfN}:=\Lambda_\sfG\cap\sfN$ is a lattice in $\sfN$ and
$\Lambda_\sfG \sfN=\sfN \Lambda_\sfG$ is a closed subgroup of $\sfG$, so
$\Lambda_\sfG \sfN\backslash\sfG$ is a torus. The twisted torus
$\torus_{\Lambda_\sfG}=\Lambda_\sfG\backslash\sfG$ is then a~fibration over this torus with fibre
the nilmanifold $\Lambda_{\sfN}\backslash\sfN=\Lambda_\sfG\backslash \Lambda_\sfG\,\sfN$. This
bundle is called the \emph{Mostow bundle}~\cite{Mostow1954}. We summarise these statements as
\begin{Theorem}[Mostow bundles]
\label{thm:Mostowbundle}
Let ${\Lambda_\sfG}$ be a lattice in a connected and simply-connected
solvable Lie group $\sfG$ and $\torus_{\Lambda_\sfG} =
\Lambda_\sfG\backslash\sfG$ the associated solvmanifold. Let $\sfN$ be the nilradical
 of $\sfG$.
Then $\Lambda_\sfG\,\sfN$ is a closed subgroup of $\sfG$,
$\Lambda_{\sfN}:=\Lambda_\sfG\cap\sfN$ is a lattice in
$\sfN$, and $\Lambda_\sfG \sfN\backslash\sfG$ is a torus.
It follows that the twisted torus
$\torus_{\Lambda_\sfG}$ is a fibration over this torus with
nilmanifold fibre:
\[
\Lambda_{\sfN}\backslash\sfN=\Lambda_\sfG\backslash\Lambda_\sfG \sfN
\longrightarrow
\torus_{\Lambda_\sfG}\longrightarrow \Lambda_\sfG \sfN\backslash\sfG .
\]
\end{Theorem}

\begin{Remark}\label{rem:Mostow}
The structure group of the Mostow bundle is
$\Lambda_{\sfG_0}\backslash\Lambda_\sfG \sfN$, where $\Lambda_{\sfG_0}$ is
the largest subgroup of $\Lambda_{\sfG}$ which is normal in
$\Lambda_\sfG \sfN$ (cf.~\cite{Bock2009}).
In particular, if $\Lambda_{\sfG}=\Lambda_{\sfG_0}$ then the Mostow bundle is
a principal $\Lambda_{\sfG}\backslash\Lambda_\sfG \sfN$-bundle. In
this case there is a well-defined left $\Lambda_\sfG \sfN$-action on $\torus_{\Lambda_\sfG} =
\Lambda_\sfG\backslash\sfG$ and each point has isotropy subgroup ${\Lambda_\sfG}$,
so that the induced $\Lambda_{\sfG}\backslash\Lambda_\sfG \sfN$-action
is principal.
\end{Remark}

If the solvable Lie group $\sfG$ admits an abelian normal subgroup
$\V$, then $\Lambda_\sfG \V=\V \Lambda_\sfG$ is a~subgroup of $\sfG$; if $\Lambda_\sfG$
is normal in $\Lambda_\sfG \V$, then the Mostow bundle construction can be refined via an intermediate step
involving a principal torus bundle over a second solvmanifold.
Adapting~\cite[Theorem~3.6]{Bock2009} we have

\begin{Proposition}\label{prop:doublebundle}
Let $\sfG$ be a connected and simply-connected solvable Lie group and
$\Lambda_\sfG$ a lattice in $\sfG$. Let $\V$ be a closed normal
abelian Lie subgroup of $\sfG$ such that $\Lambda_\sfG$ is normal in $\V\,\Lambda_\sfG$.
If $\V^\zed:=\Lambda_\sfG \cap \V$ is a lattice in $\V$, then
$\Lambda_\sfG\backslash\Lambda_\sfG \V$ is a torus and
the solvmanifold $\torus_{\Lambda_\sfG}=\Lambda_\sfG\backslash\sfG$ is the total space of
the principal torus bundle
\begin{equation*}
\Lambda_\sfG\backslash\Lambda_\sfG
\V\longrightarrow \torus_{\Lambda_\sfG}
\longrightarrow \Lambda_\sfG \V\setminus \sfG ,
\end{equation*}
with base the solvmanifold
$\torus_{\Lambda_{\sfG^{\text{\tiny$\V$}}}}=\Lambda_{\sfG^{\text{\tiny$\V$}}} \backslash
\sfG^{\text{\tiny$\V$}}:=\big(\V^\Z\backslash\Lambda_{\sfG}\big)
 \big\backslash \big(\V\backslash\sfG\big) =
\Lambda_\sfG \V\setminus \sfG$.
There is moreover a~double fibration
\begin{equation}\label{doublefib}\begin{split}&
\xymatrix{\torus^n \ar[r]&
\torus_{\Lambda_\sfG} \ar[d] \\
\Lambda_{\sfN^{\text{\tiny$\V$}}}
\backslash\sfN^{\text{\tiny$\V$}} \ar[r] &
\torus_{\Lambda_{\sfG^{\text{\tiny$\V$}}}} \ar[d] \\
& \torus^m
}\end{split}
\end{equation}
where $n=\dim \V$,
$m=\dim\big(\sfN^{\text{\tiny$\V$}}\backslash\sfG^{\text{\tiny$\V$}}\big)$,
$\sfN^{\text{\tiny$\V$}}$ is the nilradical of
$\sfG^{\text{\tiny$\V$}}$ and
$\Lambda_{\sfN^{\text{\tiny$\V$}}}={\sfN^{\text{\tiny$\V$}}}\cap
\Lambda_{\sfG^{\text{\tiny$\V$}}}$ the associated lattice.
\end{Proposition}
\begin{proof}
 Let ${\rm p}\colon \sfG\to \V\backslash\sfG$ be the canonical
 projection. Since $\V$ is normal in $\sfG$, and $\Lambda_\sfG$ and
 $\Lambda_\sfG \cap \V$ are lattices in~$\sfG$ and~$\V$, respectively,
 by~\cite[Lemma~5.1.4(a)]{Corwin1990} it follows that
$ {\rm p} (\Lambda_\sfG)$ is a~lattice in~$\V\backslash\sfG$. Hence
${\rm p}^{-1}({\rm p} (\Lambda_\sfG))=\V \Lambda_\sfG=\Lambda_\sfG \V$
is closed in $\sfG$, and $\pi\colon \sfG\to
\Lambda_\sfG \V\backslash\sfG$ is a bundle. By~\cite[Section~7.4]{Steenrod} (adapted to the smooth case), since
$\Lambda_\sfG$ is a~closed normal subgroup of~$\Lambda_\sfG \V$, it
follows that~$\Lambda_\sfG\backslash\sfG\to
\Lambda_\sfG \V\backslash\sfG$ is a
principal $\Lambda_\sfG\backslash\Lambda_\sfG \V$-bundle (or in
other words, Remark~\ref{rem:Mostow} holds as well under the present hypotheses).
The fiber is a torus because~$\V$ is abelian,
$\Lambda_\sfG\backslash\Lambda_\sfG \V=\V^\zed\backslash\V=\torus^n$,
with $n=\dim\V$. Moreover, because $\V$ is normal in $\sfG$, there
is a~canonical action of the group~$\V^\Z\backslash\Lambda_{\sfG}$ on
the connected and simply-connected solvable Lie group
$\V\backslash\sfG$ (given by $\big(\V^\zed \lambda\big) (\V
g)=\V (\lambda g)$ for $\lambda\in\Lambda_\sfG$ and $g\in\sfG$), so that
$\Lambda_{\sfG^{\text{\tiny$\V$}}} \backslash
\sfG^{\text{\tiny$\V$}}:=\big(\V^\Z\backslash\Lambda_{\sfG}\big)
 \big\backslash \big(\V\backslash\sfG\big)$ is a solvmanifold. It
is then easily proven that $\big(\V^\Z\backslash\Lambda_{\sfG}\big)
 \big\backslash \big(\V\backslash\sfG\big)=\Lambda_\sfG\V\backslash
\sfG$.
The double fibration~\eqref{doublefib} follows immediately from
Theorem~\ref{thm:Mostowbundle} applied to
$\torus_{\Lambda_{\sfG^{\text{\tiny$\V$}}}}= \Lambda_{\sfG^{\text{\tiny$\V$}}} \backslash
\sfG^{\text{\tiny$\V$}}$.
\end{proof}

\begin{Remark}An equivalent form for the double fibration \eqref{doublefib} is given by
\begin{equation*}
\xymatrix{
\torus^n \ar[r] &
\torus_{\Lambda_\sfG} \ar[d] \\
\Lambda_{\sfN^{\text{\tiny$\V$}}
}\backslash\sfN^{\text{\tiny$\V$}} \ar[r] & \torus_{\Lambda_\sfG} /\V \ar[d] \\
& \torus^m
}
\end{equation*}
which is obtained by observing that $\Lambda_\sfG\V\backslash
\sfG=\Lambda_\sfG\backslash \sfG\big{/}\V$ since $\V$
is normal in $\sfG$. Notice also that
$\sfG/\V=\V\backslash\sfG$ and $\Lambda_\sfG/\V^\zed=\V^\zed\backslash
\Lambda_\sfG$ (because $\V^\zed=\Lambda_\sfG\cap \V$ is normal in
$\Lambda_\sfG$), and moreover
$
\Lambda_{\sfG^{\text{\tiny$\V$}}} \backslash
\sfG^{\text{\tiny$\V$}}=\big(\V^\Z\backslash\Lambda_{\sfG}\big)
\big\backslash (\V\backslash\sfG)=
\big(\Lambda_{\sfG}/\V^\zed\big)
\big\backslash (\sfG/\V)$.
\end{Remark}

\subsection{Almost abelian solvmanifolds}\label{sec:abeliansolv}

In contrast to the case of nilpotent groups, there is no
simple criterion for the existence of a~lattice in a~general
connected and simply-connected solvable Lie group, as is required to define a corresponding twisted
torus. To formulate such a criterion, we specialise to \emph{almost
abelian solvable groups}: these are the solvable Lie groups $\sfG$ of
dimension $d$
whose nilradical $\sfN$ has codimension one and is abelian:
\[
\sfN \simeq \real^{d-1} .
\]
Then $\sfG$ has the
structure of a semi-direct product
\[
\sfG = \sfN\rtimes_\varphi\real
\]
for a continuous one-parameter left group action
$\varphi\colon \real\to\Aut(\sfN)$; concretely, $\varphi$ is given by the
adjoint action of the one-dimensional subgroup $\sfH=\real$ on $\sfN$
in the group $\sfG$ (cf.~Section~\ref{sec:semidirect}). This exhibits
$\sfG$ as a nontrivial group extension
\[
1\longrightarrow \sfN \longrightarrow \sfG \longrightarrow \R \longrightarrow 1 .
\]

We can regard the one-parameter group action $\varphi$ as a matrix
$\varphi_x\in\sfGL(d-1,\real)$ for each $x\in\real$ acting on the
vector space $\sfN$, which we identify with $\real^{d-1}$ via a choice
of basis $\vec e_1,\ldots,\vec e_{d-1}$. Since
$\varphi_{0}=\unit_{d-1}$ and $\varphi_x$ is always non-singular, it
follows from the continuity of $\varphi$ and the determinant that
$\det\varphi_x>0$ for all $x\in\real$. Then $\sfG$ admits a lattice
$\Lambda_\sfG$ if and
only if there exists $x_0\in\real^\times$ such that $\varphi_{x_0}$ is
conjugate to an integer matrix $\ttM\in\sfSL(d-1,\zed)$:
\begin{gather}\label{eq:monodromydef}
\Sigma^{-1} \varphi_{x_0} \Sigma=\ttM
\end{gather}
for some
$\Sigma\in\sfGL(d-1,\real)$.
In this case the twisted torus
$\torus_{\Lambda_\sfG}=\Lambda_{\sfG}\backslash \sfG$ is called an {\emph{almost abelian solvmanifold}}.
The condition~\eqref{eq:monodromydef} strongly restricts the homomorphisms
$\varphi\colon \real\to\Aut(\sfN)$; in particular, it requires that the
characteristic polynomial of $\varphi_{x_0}$ has integer
coefficients. In this case the lattice (in the standard basis $\vec
e_1,\ldots,\vec e_{d-1}$) is given by
\[
\Lambda_\sfG = \Sigma\cdot\zed^{d-1}\rtimes_{\varphi|_{x_0\,\zed}}
x_0 \zed ,
\]
which correspondingly sits as a nontrivial group extension
\[
1\longrightarrow \Sigma \cdot \zed^{d-1} \longrightarrow \Lambda_\sfG
\longrightarrow x_0 \zed \longrightarrow 1 .
\]
Then $\Lambda_{\sfN}\setminus\sfN\simeq\torus^{d-1}$, and the corresponding Mostow bundle realises the twisted torus
$\torus_{\Lambda_\sfG}$ as a torus bundle over a circle
$\Lambda_\sfG \sfN\backslash\sfG \simeq \torus$, whose monodromy
is specified by the matrix $\ttM$ in the mapping class group $\sfSL(d-1,\zed)$ of
orientation-preserving
automorphisms up to homotopy of the torus fibres~$\torus^{d-1}$.

For an almost abelian solvmanifold we can make the twisted torus construction more concrete
by choosing the global
coordinates $(z,x)\in\real^{d-1}\times\real$ on the group
manifold (associated with the basis $\vec e_1,\ldots,\vec e_{d-1}$).
The group multiplication of the semi-direct product
$\sfG=\real^{d-1}\rtimes_\varphi\real$ is
\begin{gather}\label{eq:grouplaw}
(z,x) (z',x') = (z + \varphi_x\cdot z',x+x') ,
\end{gather}
where we used $\varphi_x \varphi_{x'} = \varphi_{x+x'}$, and the
inverse of a group element is
\[
(z,x)^{-1} = (-\varphi_{-x}\cdot z,-x) ,
\]
where we used $\varphi_x^{-1}=\varphi_{-x}$.
The twisted torus is defined as the quotient $\torus_{\Lambda_\sfG}=\Lambda_\sfG
{\setminus} \sfG$ which is generated by the equivalence relation $(z,x)\sim
(\Sigma\cdot\gamma,x_0\,\alpha)\,(z,x)$ for all $(\gamma,\alpha)\in\zed^{d-1}\times\zed$. The global structure
of $\torus_{\Lambda_{\sfG}}$ is generated by the simultaneous local coordinate
identifications under the action of the elements
$(\Sigma\cdot\gamma,x_0 \alpha) $ of $\Lambda_\sfG$ given by
\begin{gather}
(z,x) \longmapsto (
 z+\Sigma\cdot\gamma,x) , \nonumber \\
(z,x) \longmapsto (
\Sigma {\tt M}^\alpha\,\Sigma^{-1}\cdot z,x+x_0\,\alpha) .\label{eq:Mostowtorus}
\end{gather}
These identifications explicitly exhibit the twisted torus as a torus bundle over a circle, with
local fiber coordinates $z\in \torus^{d-1}$ and base coordinate $x\in \torus$,
 whose monodromy is specified by the matrix ${\tt M}\in
 \sfSL(d-1,\zed)$, and whose periods are given respectively by
 $\Sigma\in\sfGL(d-1,\real)$ and $x_0\in\real^\times$.

\subsection{$C^*$-algebra bundles}\label{sec:C*bundles}

Mostow bundles and their $C^*$-algebraic T-duals can be grouped together
under the general heading of `$C^*$-algebra bundles', which encompasses
the notions of $C_0(X)$-algebras and $C^*$-bundles, as we now
explain; see Section~8.1 and Appendix~C of~\cite{Williams2007} for
further details.
Let~$X$ be a locally compact Hausdorff space. There are two equivalent
notions for the $C^*$-algebra analogue of a fibre bundle over $X$.

A \emph{$C_0(X)$-algebra} is a $C^*$-algebra $\alg$
equipped with a nondegenerate injection $\iota$ of $C_0(X)$
into the centre of its multiplier algebra, called the \emph{structure
 map}. For $\mathtt{f}\in C_0(X)$ and
$a\in\alg$, we abbreviate $\iota(\mathtt{f})a$ by $\mathtt{f}\cdot a$. This endows
$\alg$ with a $C_0(X)$-bimodule
structure.

Given a family $\balg=(\balg_x)_{x\in X}$ of $C^*$-algebras, a
\emph{section} of $\balg$ is
a map $s\colon X\to\balg$ such that $s(x)\in\balg_x$ for all $x\in X$; we
denote the space of sections of $\balg$ which vanish at infinity by~$\Gamma_0(\balg)$. The family $\balg$ is then called a~\emph{$C^*$-bundle} over $X$ with fibres $\balg_x$ if the following
conditions are satisfied:
\begin{itemize}\itemsep=0pt
\item
$\Gamma_0(\balg)$ is a $C^*$-algebra under pointwise operations and the
 supremum norm;
\item $\balg_x=\{s(x)\, |\, s\in\Gamma_0(\balg)\}$ for each
 $x\in X$;
\item $\Gamma_0(\balg)$ is closed under multiplication by
 $C_0(X)$; and
\item For each $s\in\Gamma_0(\balg)$, the function $x\mapsto \|s(x)\|$
is upper semi-continuous, i.e., the set
$\{x\in X \,|\, \|s(x)\|< \epsilon\}$ is open in $X$ for all $\epsilon>0$.
\end{itemize}
In this
paper we will only be concerned with $C^*$-bundles that have non-zero fibres.

If $\balg$ is a $C^*$-bundle over $X$, then its section algebra
$\Gamma_0(\balg)$ is a $C_0(X)$-algebra: its structure map $\iota$ from
$C_0(X)$ is defined by $\iota(\mathtt{f})s=\mathtt{f} s$. Conversely, if $\alg$ is a
$C_0(X)$-algebra, then the fibre~$\alg_x$ of $\alg$ over $x\in X$ is $\alg_x:=\alg/\CI_x$, where
$\CI_x=\{\mathtt{f}\cdot a\, |\, \mathtt{f}\in C_0(X) , \,
\mathtt{f}(x)=0 , \, a\in\alg\}$ is identified as the ideal in
$\alg$ of sections vanishing at $x$. If
$a\in\alg$, we write $a(x)=a+\CI_x$ for its image in~$\alg_x$. The
function $x\mapsto\|a(x)\|$ is upper semi-continuous
and vanishes at infinity with
\[
\|a\| = \sup_{x\in X} \big\|a(x)\big\|
\]
for all $a\in\alg$. The
elements $a\in \alg$ can in this way be viewed as sections of a
$C^*$-bundle $(\alg_x)_{x\in X}$. We will sometimes use the
notation $\coprod_{x\in X} \alg_x$ for $\alg$ when we wish to emphasise
its structure as a $C^*$-bundle over $X$ with fibre $C^*$-algebras
$\alg_x$. These definitions do not require
local triviality of the bundle nor the fibres of the bundle to be
isomorphic to one another. $C^*$-algebra bundles over $X$ are objects of a
category whose morphisms are \emph{fibrewise} $*$-homomorphisms,
i.e., $C_0(X)$-linear morphisms $\psi\colon \alg\to\balg$: $\psi(\mathtt{f}\cdot
a)=\mathtt{f}\cdot\psi(a)$ for all $\mathtt{f}\in C_0(X)$ and $a\in\alg$; then $\psi$
induces $*$-homomorphisms $\psi_x\colon \alg_x\to\balg_x$ such that
$\psi_x\big(a(x)\big)=\psi(a)(x)$ for all $a\in\alg$.

\begin{Example}[trivial $C^*$-algebra bundles]
If $\Dcal$ is any $C^*$-algebra, then $\alg=C_0(X,\Dcal)\simeq
C_0(X)\otimes\Dcal$ is naturally a $C_0(X)$-algebra with structure map
\[
\big(\iota(\mathtt{f})a\big)(x) := \mathtt{f}(x)\,a(x)
\]
for $\mathtt{f}\in C_0(X)$, $a\in\alg$ and $x\in X$. In this case each fibre
$\alg_x$ is canonically identified with~$\Dcal$ and elements of~$\alg$
are obviously identified with sections.
\end{Example}

\begin{Example}[continuous maps]\label{ex:ctsmaps}
Let~$X$ and~$Y$ be locally compact spaces and $\sigma\colon Y\to X$ a~continuous surjective map. Then $C_0(Y)$ is a $C_0(X)$-algebra with structure map
$\iota(\mathtt{f})\mathtt{g}:=(\mathtt{f}\circ\sigma)\,\mathtt{g}$, for $\mathtt{f}\in C_0(X)$ and $\mathtt{g}\in C_0(Y)$,
and fibers $C_0(Y)_x\simeq C_0\big(\sigma^{-1}(x)\big)$.
\end{Example}

In this paper we are particularly interested in crossed
products of $C^*$-algebra bundles. Let~$\alg$ be a $C_0(X)$-algebra,
and denote by $\Aut_X(\alg)$ the group of {fibrewise
automorphisms} of~$\alg$. A~\emph{fibrewise action} of a locally compact group $\sfG$ on
$\alg$ is then a group homomorphism $\alpha\colon \sfG\to\Aut_X(\alg)$. This
implies that $\alpha$ induces an action~$\alpha^x$ on each fiber
$\alg_x$ for $x\in X$, and in
this case we say that the dynamical system $(\alg,\sfG,\alpha)$ is
\emph{$C_0(X)$-linear}.

\begin{Theorem}\label{thm:crossedfibres}
Let $X$ be a locally compact Hausdorff space, and let
$(\alg,\sfG,\alpha)$ be a $C_0(X)$-linear $C^*$-dynamical system. Then
the crossed product $\alg\rtimes_\alpha\sfG$ is again a
$C_0(X)$-algebra with fibres
\[
(\alg\rtimes_\alpha\sfG)_x\simeq\alg_x\rtimes_{\alpha^x}\sfG ,
\]
where
\[
\alpha_\gamma^x\big(a(x)\big) = \alpha_\gamma(a)(x)
\]
for each $x\in X$, $\gamma\in\sfG$ and $a\in\alg$.
\end{Theorem}

\begin{proof}The structure map
of $\alg\rtimes_\alpha\sfG$ is given by precomposing the structure map
of $\alg$ with the natural injection of the center of the multiplier
algebra of $\alg$ into the center of the multiplier algebra of
$\alg\rtimes_\alpha\sfG$; it satisfies
$(\mathtt{f}\cdot f)(\gamma)=\mathtt{f}\cdot\big(f(\gamma)\big)$ for all
$\mathtt{f}\in C_0(X)$, $f\in C_{\rm c}(\sfG,\alg)$ and
$\gamma\in\sfG$. See~\cite[Theorem~8.4]{Williams2007} for further
details.
\end{proof}

\begin{Example}[transformation groups]
\label{ex:transfgroups}
Let $(X,\sfG)$ be a second countable transformation group whose quotient $X/\sfG$ is a
Hausdorff space. By
Example~\ref{ex:ctsmaps},
$C_0(X)$ is a $C_0(X/\sfG)$-algebra whose fiber over $\sfG\cdot x$ is
isomorphic to
$C_0(\sfG/\sfG_x)$, where $\sfG_x=\{\gamma\in\sfG\,|\,\gamma\cdot x=x\}$ is the stabilizer subgroup at $x\in
X$. Then the crossed product $C_0(X)\rtimes_\alpha\sfG$ is the
section algebra of a $C^*$-algebra bundle over $X/\sfG$ whose fiber
over $\sfG\cdot x$ is isomorphic to
$C^*(\sfG_x)\otimes\CK\big(\Leb^2(\sfG/\sfG_x)\big)$~\cite{Williams1989}. In
the special case where
$\sfG=\real$ and $\sfG_x=\zed$ for all $x\in X$, this is contained in
the statement of Proposition~\ref{prop:Raeburn}.
\end{Example}

\begin{Example}[principal torus bundles]\label{ex:NCprincipalT}
Let $E\to X$ be a principal $\torus^r$-bundle. By
Example~\ref{ex:transfgroups}, $C_0(E)$ is a $C_0(X)$-algebra
with fibers $C_0(E)_x\simeq
C(\torus^r)$, and by \eqref{eq:XsfHiso} there is a
stable isomorphism $C_0(E)\rtimes_\rt\torus^r\simeq
C_0(X)\otimes\CK\big(\Leb^2(\torus^r)\big)$. More generally, a
\emph{noncommutative principal $\torus^r$-bundle} on $X$ is a
$C_0(X)$-linear $C^*$-dynamical system $(\alg,\torus^r,\alpha)$ with
an isomorphism
\[
\alg\rtimes_\alpha\torus^r\simeq C_0(X,\CK)
\]
of $C^*$-algebra bundles over $X$. For further details and a
classification of noncommutative principal torus bundles,
see~\cite{Echterhoff2009,Hannabuss2010}.
\end{Example}

\begin{Example}[noncommutative correspondences]
The noncommutative correspondence $\Ccal =
C(\torus_{\Lambda_\sfG})\rtimes_\rt\Z^n$ from
Proposition~\ref{prop:NCcorr} is a noncommutative principal
$\torus^n$-bundle on $X=\torus_{\Lambda_\sfG}$ in the sense of
Example~\ref{ex:NCprincipalT}: The $C^*$-algebra $\Ccal$ is naturally equipped with the
dual action of $\torus^n=\widehat{\Z^n}$, and the Takai duality
theorem implies that there is an isomorphism $\Ccal\rtimes_{\widehat{\rt}}\torus^n\simeq
C(\torus_{\Lambda_\sfG})\otimes\CK\big(\ell^2(\Z^n)\big)$ of
$C^*$-algebra bundles over the twisted torus $\torus_{\Lambda_\sfG}$.
\end{Example}

\begin{Remark}\label{rem:MoritaRKK}
There is a natural notion of Morita equivalence of $C^*$-algebra bundles over~$X$, similar to the notion of
equivariant Morita equivalence from Theorem~\ref{equivMorita}, which
uses the $C_0(X)$-bimodule structures: a~\emph{$C_0(X)$-linear Morita equivalence} between two
$C_0(X)$-algebras is a Morita equivalence which is compatible with the
$C^*$-bundle structures over~$X$. More generally, there is a category
$\mathcal{R}\CCK\hspace{-1mm}\CCK_X$ of $C^*$-algebra bundles over $X$ whose morphisms are
elements of Kasparov's groups $\mathcal{R}{\rm KK}(X;\alg,\balg)$, see, e.g., \cite{Echterhoff2009}: the cycles are the usual cycles
$(\hil,\phi,T)$ for Kasparov's bivariant K-theory ${\rm
 KK}(\alg,\balg)$ (cf.\ Section~\ref{sec:TdualityKK}) with the
additional requirement that $\phi\colon \alg\to\End_\balg(\hil)$ is
$C_0(X)$-linear. There is an obvious faithful functor
$\mathcal{R}\CCK\hspace{-1mm}\CCK_X \to \CCK\hspace{-1mm}\CCK$ which forgets the
$C_0(X)$-algebra structures. Isomorphic $C^*$-bundles in the category
$\mathcal{R}\CCK\hspace{-1mm}\CCK_X$ are precisely the $\mathcal{R}{\rm KK}$-equivalent
$C^*$-bundles. If $\alg$ and $\balg$ are isomorphic in
$\mathcal{R}\CCK\hspace{-1mm}\CCK_X$, i.e., there exists an invertible class
$\alpha\in \mathcal{R}{\rm KK}(X;\alg,\balg)$, then they are also
isomorphic as
$C^*$-algebras in the category $\CCK\hspace{-1mm}\CCK$.
\end{Remark}

\subsection[$\real^n$-actions on Mostow bundles]{$\boldsymbol{\real^n}$-actions on Mostow bundles}\label{sec:RMostow}

We can now apply the results of Section~\ref{sec:TopTtorus} to the
class of twisted tori given in
Section~\ref{sec:abeliansolv}.
The Mostow fibration of any almost abelian solvmanifold identifies
$\torus_{\Lambda_\sfG}$ as a torus bundle over a~circle, hence the
algebra of functions $C(\torus_{\Lambda_\sfG})$ is a~$C(\torus)$-algebra. In other words, $C(\torus_{\Lambda_\sfG})$ is an
object of the category $\mathcal{R}\CCK\hspace{-1mm}\CCK_\torus$, and we are
interested in the T-duality isomorphisms of~$C(\torus_{\Lambda_\sfG})$
in this category. In particular, given a fibrewise right action of the abelian
Lie group~$\R^n$ on~$\torus_{\Lambda_\sfG}$, it follows from
Theorem~\ref{thm:crossedfibres} that the $C^*$-algebraic T-dual
$C(\torus_{\Lambda_\sfG})\rtimes_\rt\R^n$ is also a $C(\torus)$-algebra, and by~\cite[Theorem~3.5]{Echterhoff2009} the $C^*$-bundles
$C(\torus_{\Lambda_\sfG})$ and
$C(\torus_{\Lambda_\sfG})\rtimes_\rt\R^n$ are isomorphic as $C^*$-algebras in the
category $\mathcal{R}\CCK\hspace{-1mm}\CCK_\torus$.

In order to have sensible definitions of
T-duality, we need to identify the
homologically non-trivial one-cycles of the twisted torus
$\torus_{\Lambda_\sfG}$, which are determined
in~\cite[Proposition~4.7]{Bock2009}. Write
\[
\ttM = (m_{ij})
\]
for the
integer matrix elements $m_{ij}\in\Z$ of the monodromy matrix. Since
$\sfG=\real^{d-1}\rtimes_\varphi\real$ is simply-connected, the
fundamental group of the twisted torus is
$\pi_1(\torus_{\Lambda_\sfG})\simeq\Lambda_\sfG$ whose abelianisation
$\Lambda_\sfG/[\Lambda_\sfG,\Lambda_\sfG]$ gives
the first homology group via the presentation
\begin{gather}\label{eq:H1torus}
{\rm H}_1(\torus_{\Lambda_\sfG},\zed) = \zed \oplus \bigg\langle
\hat e_1,\dots,\hat e_{d-1} \, \Big| \, \sum_{j=1}^{d-1} m_{ji} \hat
e_j = \hat e_i
\ \mbox{for} \
i=1,\dots,d-1 \bigg\rangle,
\end{gather}
where the first factor of $\zed$ corresponds to the base circle of the
torus fibration and the generators $\hat e_1,\dots,\hat e_{d-1}$
correspond to the torus fibres; they are given by
\begin{gather}\label{eq:hatei}
\hat e_i:=\sum_{k=1}^{d-1} \Sigma_{ki} \vec e_k,
\end{gather}
where
$\vec e_1,\ldots, \vec e_{d-1}$ is the standard basis of $\zed^{d-1}$ giving the
group law~\eqref{eq:grouplaw}.

We can then apply the structure theorem for finitely-generated
$\zed$-modules by appealing to some classical matrix algebra.
From the presentation \eqref{eq:H1torus}, ${\rm
 H}_1(\torus_{\Lambda_\sfG},\zed) = \zed \oplus {\rm{coker}}(\varphi_{x_0}-{\rm{id}}_{\Z^{d-1}})$
where $(\varphi_{x_0}-{\rm{id}}_{\Z^{d-1}})\colon \zed^{d-1}\to\zed^{d-1}$ in the
basis $\hat e_1,\ldots,\hat e_{d-1}$ is given by the integer relation matrix
\begin{gather}\label{eq:relationmatrix}
{\tt A}:=\ttM-\unit_{d-1}~.
\end{gather}
Let $r$ be the rank of $\tt A$. This matrix can be brought into its Smith normal form~$\tt D$
by finding invertible integer matrices ${\tt L}, {\tt
 R}\in\sfGL(d-1,\zed)$ such that
\[
{\tt D} = {\tt L} {\tt A} {\tt R}
\]
is diagonal with entries $m_i\in\zed$ for
$i=1,\dots,d-1$. The
integers~$m_i$ are the \emph{elementary divisors} of
${\tt A}$. They have the properties that $m_i$ divides
$m_{i+1}$, for $0<i<d-1$, and in particular $m_i=0$ for $i>r$; they can
be computed explicitly (up to sign) as
\[
m_i = \frac{d_i({\tt A})}{d_{i-1}({\tt A})},
\]
where the \emph{$i$-th determinant divisor} $d_i({\tt A})$ is the greatest common
divisor of all $i{\times} i$ minors of the relation matrix ${\tt A}$,
with $d_0({\tt A}):=1$. The matrices ${\tt L}, {\tt
 R}\in\sfGL(d-1,\zed)$ are found by reducing the matrix ${\tt A}$ to
its Smith normal form ${\tt D}$ through a sequence of elementary row
and column operations over~$\zed$, see, e.g.,~\cite{Havas1997}.

Given the Smith normal form, we set
\begin{gather}\label{eq:tildeei}
\tilde e_i:=\sum_{j=1}^{d-1}\,\big({\tt L}^{-1}\big)_{ji}\,\hat e_j
\end{gather}
and observe that
the image of $\tt A$, which is generated over $\zed$ by the vectors
\[
\sum_{j=1}^{d-1} {\tt A}_{ji} \hat e_j=\sum_{k,l=1}^{d-1} \big({\tt
 R}^{-1}\big)_{ki} {\tt
 D}_{lk} \tilde e_k ,
\]
is equivalently generated by the
vectors $m_k \tilde e_k$ with $k=1,\ldots, r$.
Hence
\begin{align*}
{\rm{coker}}(\varphi_{x_0}-{\rm{id}}_{\Z^{d-1}})&=\big\langle \hat e_1, \ldots , \hat
e_{d-1}\big\rangle\Big/\bigg\langle \sum\limits_{j=1}^{d-1} {\tt A}_{j1} \hat e_j,\ldots,
 \sum\limits_{j=1}^{d-1} {\tt A}_{j\:\!d-1} \hat e_j\bigg\rangle\\
&=\big\langle \tilde e_1,
\ldots , \tilde e_{d-1}\big\rangle\big/\big\langle m_1 \tilde e_1,\ldots, m_r \tilde
e_r\big\rangle
\end{align*}
and
\begin{gather}\label{eq:H1gendecomp}
{\rm H}_1(\torus_{\Lambda_\sfG},\zed) \simeq \zed \oplus
\zed^{d-1-r} \oplus \bigoplus_{i=1}^r \zed_{m_i} .
\end{gather}

We are exclusively
interested in the natural $\R^n$-actions on $\torus_{\Lambda_\sfG}$ which descend
from actions of abelian subgroups $\R^n\subset\sfG$, acting on $\sfG$
by right multiplication. They can be
organised into three classes associated with the different types of
summands in the $\Z$-module presentation of the homology group~\eqref{eq:H1gendecomp}, and we only retain those which are fiberwise
actions on the Mostow bundle.
The first summand~$\Z$ in~\eqref{eq:H1gendecomp} corresponds to the subgroup
\[
\R_x = \big\{(0,\xi)\in\sfG\big\}
\]
acting on $\sfG$ by right multiplication:
\[
(z,x) (0,\xi) = (z,x+\xi)
\]
for all $(z,x)\in\sfG$ and $\xi\in\R$. Clearly this does not descend
to a fiberwise
action on the twisted torus $\torus_{\Lambda_\sfG}$, and the crossed
product $C(\torus_{\Lambda_\sfG})\rtimes_\rt\R_x$ is no longer a
$C(\torus)$-algebra. Thus an $\R_x$-action takes us out of the
category $\mathcal{R}\CCK\hspace{-1mm}\CCK_\torus$, and we will henceforth
discard $\R^n$-actions where $\R^n$ contains the subgroup~$\R_x$.

For the remaining types of summands in \eqref{eq:H1gendecomp}, we can
give explicit descriptions of the $C^*$-algebraic T-duals of an almost
abelian solvmanifold. We consider both classes in turn. As the only solvmanifolds in one and two dimensions are tori, which are
already treated by our analysis from
Section~\ref{sec:toruscrossedprod}, we assume $d\geq3$ for the
remainder of this paper.

\subsection[$\real_y$-actions: Circle bundles with $H$-flux]{$\boldsymbol{\real_y}$-actions: Circle bundles with $\boldsymbol{H}$-flux}\label{sec:Rygen}

Let us consider the second summand $\Z^{d-1-r}$ in
\eqref{eq:H1gendecomp}, which corresponds to the lattice $\Lambda_\sfG\cap
\ker\big(\varphi_{x_0}-{\rm id}_{\R^{d-1}}\big)$ in $\ker\big(\varphi_{x_0}-\Id_{\R^{d-1}}\big)$.
In terms of the generators
\[
e'_i:=\sum_{j=1}^{d-1} {\tt R}_{ji} \hat e_j
\]
of
$\Lambda_\sfG\cap \sfN$, the sublattice $\Lambda_\sfG\cap
\ker\big(\varphi_{x_0}-\Id_{\R^{d-1}}\big)$ is generated by the
vectors $e'_{r+1},e'_{r+2},\ldots,e'_{d-1}$ in the kernel of the
relation matrix~\eqref{eq:relationmatrix}. We
 begin with some
elementary observations. Firstly, the subgroups $\ker\big(\varphi_{x_0}-\Id_{\R^{d-1}}\big)$ and
$\Lambda_\sfG$ commute in $\sfG$:
$(-v,0) (0,x_0) (v,0)=(0,x_0)\in \Lambda_\sfG$ for all $v\in \ker\big(\varphi_{x_0}-\Id_{\R^{d-1}}\big)$.
Secondly, $\ker\big(\varphi_{x_0}-\Id_{\R^{d-1}}\big)$ is a closed abelian
normal subgroup of~$\sfG$:
$(z,x) (v,0) (z,x)^{-1}=(\varphi_x(v),0)\in
\ker\big(\varphi_{x_0}-\Id_{\R^{d-1}}\big)$ since
$\varphi_{x_0}(\varphi_x(v))=\varphi_{x}(\varphi_{x_0}(v))=\varphi_x(v)$.

In the following we consider a subgroup $\V\simeq\real^n\subset
\ker\big(\varphi_{x_0}-\Id_{\R^{d-1}}\big)$ which is normal in~$\sfG$
such that $\V^\zed:=\V\cap\Lambda_\sfG$ is a lattice in~$\V$.
This is the case, for example, if~$\V$ is the span of a~subset of the generators
 $e'_{r+1},e'_{r+2},\ldots, e'_{d-1}$. As an immediate consequence of
 Lemma~\ref{lem:semidirectquotient} we then have

\begin{Proposition}\label{lem:quotientgroup}
The quotient group $\sfG^{\text{\tiny$\V$}}:=\sfG/\V$ is a $(d{-}n)$-dimensional almost
abelian solvable Lie group
$$
\sfG^{\text{\tiny$\V$}}\simeq\R^{d-n-1}
\rtimes_{\varphi^{\text{\tiny$\V$}}}\R ,
$$
where $\varphi^{\text{\tiny{$\V$}}}\colon \R\to\sfGL(d-n-1,\R)$ is
defined by $\varphi^{\text{\tiny{$\V$}}}_x[z]
=(\unit_{d-1}-{\rm pr}_{\text{\tiny{$\V$}}})\,\varphi_x(z)$ for
all $x\in\R$ and $z\in \real^{d-1}$, with ${\rm pr}_{\text{\tiny{$\V$}}}$ the projection
of $\R^{d-1}$ to $\V$ and $[z]=(\unit_{d-1}-{\rm
 pr}_{\text{\tiny{$\V$}}})(z)\in \real^{d-1}/\V$.
\end{Proposition}

We are particularly interested in the induced action of $\V$ on the twisted torus $\torus_{\Lambda_\sfG}$.

\begin{Proposition}\label{lem:quotienttorus}
Let $\V$, as above, be normal in $\sfG$ and let $\Lambda_\sfG$ be
normal in $\Lambda_\sfG\,\V$. Then the quotient map
$p_{\text{\tiny{$\V$}}}\colon \torus_{\Lambda_\sfG}\to\torus_{\Lambda_\sfG}/\V$
is a principal torus bundle of rank~$n=\dim \V$ over an almost abelian
solvmanifold $\torus_{\Lambda_{\sfG^{\text{\tiny$\V$}}}}$ of
dimension $d{-}n$. Its Chern class
$c_1(p_{\text{\tiny{$\V$}}})\in\rmH^2(\torus_{\Lambda_{\sfG^{\text{\tiny$\V$}}}},\Z)$
can be computed by Chern--Weil theory from the curvature of the
connection $\kappa_{\text{\tiny{$\V$}}}\in\Omega^1(\torus_{\Lambda_\sfG},\V)$ given by
\[
\kappa_{\text{\tiny{$\V$}}} = - ({\rm pr}_\V \varphi_{-x})\cdot \dd z
\]
in the notation of Proposition~{\rm \ref{lem:quotientgroup}}.
\end{Proposition}
\begin{proof}
The first statement follows from Proposition~\ref{prop:doublebundle}.
For the Chern class, we note that the left-invariant Maurer--Cartan
one-forms on the Lie group $\sfG$ are given by $\dd x$ and
$P\in\Omega^1\big(\sfG,\R^{d-1}\big)$ where
\[
P = \varphi_{-x}\cdot \dd z ,
\]
and this descends to the twisted torus $\torus_{\Lambda_\sfG}$.
The desired principal $\torus^n$-connection on $\torus_{\Lambda_\sfG}$
is then given by
\[
\kappa_\V=-{\rm pr}_\V P = -({\rm pr}_\V \varphi_{-x})\cdot \dd z
\]
and the result follows.
\end{proof}

By virtue of the fibration
$p_\V\colon \torus_{\Lambda_\sfG}\to\torus_{\Lambda_{\sfG^{\text{\tiny$\V$}}}}$, the
algebra of functions $C(\torus_{\Lambda_\sfG})$ is also a
$C(\torus_{\Lambda_{\sfG^{\text{\tiny$\V$}}}})$-algebra. We are particularly interested in the case $n=1$, whereby we can
explicitly apply our framework of topological T-duality. Combining
Propositions~\ref{lem:quotienttorus} and~\ref{prop:Raeburn}, we immediately arrive at
\begin{Theorem}\label{prop:RyTduality}
Let $y_0\in \ker\big(\varphi_{x_0}-\Id_{\R^{d-1}}\big)$ and let
\[
\R_{y_0}:=\R (y_0,0)
\]
be the corresponding one-dimensional subgroup. Suppose that $\R_{y_0}$ is normal
in $\sfG$ and $\Lambda_\sfG$ is normal in $\Lambda_\sfG \R_{y_0}$
$($this is the case, for example, if $(y_0,0)$ is in the center of $\sfG)$.
Let $\sfG^{y_0}=\sfG/\R_{y_0}$ be the almost abelian
solvable Lie group constructed by Proposition~{\rm \ref{lem:quotientgroup}}, and
$p_{y_0}\colon \torus_{\Lambda_{\sfG}}\to\torus_{\Lambda_{\sfG^{y_0}}}$ the
principal circle bundle
constructed by Proposition~{\rm \ref{lem:quotienttorus}}.
Then the $C^*$-algebraic T-dual
\[
C(\torus_{\Lambda_\sfG})\rtimes_\rt\R_{y_0} \simeq \C\T\big(\torus_{\Lambda_{\sfG^{y_0}}}\times\torus_{y_0},\delta_{y_0}\big)
\]
is a continuous-trace algebra with spectrum $\torus_{\Lambda_{\sfG^{y_0}}}\times\torus_{y_0}$
and Dixmier--Douady class
\[
\delta_{y_0}=c_1(p_{y_0})\smile \zeta_{y_0} ,
\]
where $\zeta_{y_0}$ is the standard
generator of $\rmH^1(\torus_{y_0},\Z)\simeq\Z$ and the Chern--Weil
representative of
$c_1(p_{y_0})\in\rmH^2(\torus_{\Lambda_{\sfG^{y_0}}},\Z)$
is the curvature of the connection
$\kappa_{y_0}\in\Omega^1(\torus_{\Lambda_{\sfG}})$ on this circle
bundle given by
\[
\kappa_{y_0} = - ({\rm pr}_{y_0} \varphi_{-x})\cdot\dd z .
\]
\end{Theorem}

Thus in the case of an action of $\R$ due to a normal subgroup of
$\sfG$ that is in $\ker\big(\varphi_{x_0}-\Id_{\R^{d-1}}\big)$ in the setting
of Theorem~\ref{prop:RyTduality}, which we collectively refer to as $\R_y$-actions, the
T-dual of an almost abelian solvmanifold $\torus_{\Lambda_\sfG}$ is
the Hausdorff space $X=\torus_{\Lambda_{\sfG^{y_0}}}\times
\torus_{y_0}$ with a three-form `$H$-flux' whose cohomology class is
represented by $[H_{y_0}] = c_1(p_{y_0})\smile \zeta_{y_0}$.
The associated
correspondence space construction proceeds analogously to
Section~\ref{sec:TopTorb}, which we can give explicitly as
\begin{Proposition}\label{prop:Mostowcorr}
The topological T-duality of Theorem~{\rm \ref{prop:RyTduality}} is
implemented by the noncommutative
correspondence
\[
\xymatrix{
 & C\big(\torus_{\Lambda_\sfG}\times\widetilde\torus_{y_0}\big) & \\
C\big(\torus_{\Lambda_\sfG}\big) \ar[ur]^{ \! \! \! \! [{\rm pr}^*]} & &
\C\T\big(\torus_{\Lambda_{\sfG^{y_0}}}\times\torus_{y_0},\delta_{y_0}\big)
\ar[ul]_{ \ \ \
 [(p_{y_0}\times\Id_{\widetilde\torus_{y_0}})^*]} \\
 & C\big(\torus_{\Lambda_{\sfG^{y_0}}}\big) \ar[ur]_{[j]}
 \ar[ul]^{[p_{y_0}^*]} &
}
\]
as a diagram in the category $\mathcal{R}\CCK\hspace{-1mm}\CCK_{\torus}$.
\end{Proposition}
\begin{proof}
Since the subgroup
\[
\Z_{y_0}:=\R_{y_0}\cap\Lambda_\sfG
\]
acts trivially on the algebra of functions $C(\torus_{\Lambda_\sfG})$, there is an isomorphism
\[
C(\torus_{\Lambda_\sfG})\rtimes_\rt\zed_{y_0} \simeq
C\big(\torus_{\Lambda_{\sfG}}\times\widetilde\torus_{y_0}\big) ,
\]
where $\widetilde\torus_{y_0}$ is the circle dual to $\torus_{y_0}\simeq\R_{y_0}/\Z_{y_0}$.
Proposition~\ref{lem:quotienttorus} shows that the Chern--Weil
representative of
$c_1(p_{y_0})\in\rmH^2\big(\torus_{\Lambda_{\sfG^{y_0}}},\Z\big)$ pulls back to the exact
two-form $\dd\kappa_{y_0}\in\Omega^2(\torus_{\Lambda_{\sfG}})$ under
the bundle projection $p_{y_0}:\torus_{\Lambda_\sfG}\to
\torus_{\Lambda_\sfG^{y_0}}$. Hence
$p_{y_0}^*c_1(p_{y_0})=0$ which implies
\[
\big(p_{y_0}\times\Id_{\widetilde\torus_{y_0}}\big)^*\big(c_1(p_{y_0})\smile\zeta_{y_0}\big)=0 \in {\rm
 H}^3\big(\torus_{\Lambda_\sfG}\times\widetilde\torus_{y_0},\zed\big) .
\]
Thus the algebra $\C\T(\torus_{\Lambda_\sfG^{y_0}}\times\torus_{y_0},\delta_{y_0})$ pulls back to an algebra
isomorphic to
$C\big(\torus_{\Lambda_\sfG}\times\widetilde\torus_{y_0}\big)\otimes\CK$ by
$p_{y_0}\times\Id_{\widetilde\torus_{y_0}}$, and the result follows.
\end{proof}

\subsection[$\real_z$-actions: noncommutative torus bundles]{$\boldsymbol{\real_z}$-actions: noncommutative torus bundles}\label{sec:Rzgen}

Let us now come to the torsion summands $\Z_{m_i}$ in the
$\Z$-module decomposition \eqref{eq:H1gendecomp}. Pick a~non-trivial
elementary divisor $m_i>1$ for some $i\in\{1,\dots,r\}$. The
corresponding homology generator $\tilde e_i$ is constructed as a~$\Z$-linear combination~\eqref{eq:tildeei} of the generators~\eqref{eq:hatei}. It defines a~fixed
element~$z_0$ in the image of $\varphi_{x_0}-\rm{id}_{\R^{d-1}}$ and a~corresponding one-dimensional subgroup of~$\sfG$ given by
\[
\R_{z_0} := \R (z_0,0) .
\]
We further assume that
\[
\Z_{ z_0} := \R_{ z_0}\cap\Lambda_\sfG
\]
is a lattice in $\R_{z_0}$.
The choice of $z_0\in\R^{d-1}$ is not unique, and
any change of basis of $\Z^{d-1}$, represented by a matrix
$B\in\sfGL(d-1,\Z)$, defines an equally good element $B\cdot
z_0\in\R^{d-1}$ as long as $B\cdot
z_0\in{\rm im}\big(\varphi_{x_0}-\Id_{\R^{d-1}}\big)$. We write $\proj_{z_0}$ for the
linear projection of $\R^{d-1}$ to $\R_{z_0}$, and denote by $\langle
z_0,z\rangle\in\R$ the component of $z\in\R^{d-1}$ in $\R_{z_0}$,
i.e.,~$\proj_{z_0}\cdot z = \langle z_0,z\rangle \, z_0$.

In the basis $\hat e_i$, the lattice of $\sfG$ is given by
\[
\Lambda_\sfG=\Z^{d-1}\rtimes_{\hat\varphi|_{x_0\,\Z}} x_0 \Z ,
\]
where $\hat\varphi_x:=\Sigma^{-1}\,\varphi_x\,\Sigma$ for all
$x\in\R$. Then the fibres of the underlying Mostow bundle are `square' tori
$\torus^{d-1}$ with unit periodicities $ z\sim z+\vec
e_i$ for $i=1,\dots,d-1$, where as
before $\vec e_i$ denotes the standard basis of
$\R^{d-1}$. The action of the
subgroup $\R_{ z_0}$ on
elements $(z,x)\in\sfG$ by right multiplication is given by
\begin{gather}\label{eq:Rz0action}
(z,x) (\zeta\,z_0,0) = \big(z+\zeta \Sigma^{-1} \varphi_x\cdot
z_0,x\big)
\end{gather}
for $\zeta\in\R$, where $\Sigma^{-1} \varphi_x\cdot z_0$ lies in the
image of the relation matrix ${\tt A}=\ttM-\unit_{d-1}$.

Our principal tool to compute the $C^*$-algebraic
T-dual for such an action of $\R$ in the image of $\varphi_{x_0}-\Id_{\R^{d-1}}$, which we collectively refer to as
$\R_z$-actions, will be Green's theorem in the form~\eqref{scheme}:
\begin{gather}\label{Mostowscheme}
C(\torus_{\Lambda_\sfG})\rtimes_\rt \real_{ z_0} \ \mbf\sim_{\text{\tiny M}} \
C_0(\sfG/\real_{ z_0})\rtimes_\lt \Lambda_\sfG .
\end{gather}
By appealing to
Theorem~\ref{thm:crossedfibres}, we may apply \eqref{Mostowscheme}
fibrewise. For fixed $x\in\R$, the fibre $\sfG_x$ of the semi-direct
product $\sfG=\R^{d-1}\rtimes_{\hat\varphi}\R$ is the subgroup~$\R^{d-1}$, and the corresponding fibre of
the solvmanifold $\torus_{\Lambda_\sfG}$ over $x\in\R/x_0\,\Z$ is the
torus $\torus^{d-1}=\R^{d-1}/\Z^{d-1}$. The Morita
equivalence~\eqref{Mostowscheme} is $C(\torus)$-linear and the fibres of the corresponding
T-dual $C(\torus)$-algebra are given by the fibrewise Morita
equivalence
\begin{gather}\label{eq:Mostowschemefiber}
\big(C(\torus_{\Lambda_\sfG})\rtimes_\rt\real_{ z_0}\big)_x \simeq
C\big(\R^{d-1}/\Z^{d-1}\big)\rtimes_{\rt^x}\real_{ z_0} \
\mbf\sim_{\text{\tiny M}} \
C_0\big(\real^{d-1}/\R_{ z_0}\big)\rtimes_{\lt^x}\Z^{d-1} .
\end{gather}
The action of the subgroup $\Z^{d-1}\subset\Lambda_\sfG$ on the
coset space $\R^{d-1}/\R_{ z_0}$ is induced by left multiplication in the
group $\sfG$.
After a basis transformation, we can decompose the discrete group $\Z^{d-2}$ into a direct sum
\[
\Z^{d-1} \simeq \Z_v^{d-2}\oplus\Z_{ z_0} ,
\]
where $\Z_v^{d-1} = (\unit_{d-1}-\proj_{z_0})\cdot\Z^{d-1}$.

Let
\[
\mathbb{F}_* := \big\{x\in\R \, \big| \, \big\langle
z_0,\Sigma^{-1} \varphi_x\cdot z_0\big\rangle=0\big\} .
\]
\begin{Lemma}\label{lem:ellfibercom}
Over any $x\in \mathbb{F}_*$, the fiber
$\big(C(\torus_{\Lambda_\sfG})\rtimes_\rt\real_{ z_0}\big)_x$
is Morita equivalent to the
commutative $C^*$-algebra \smash{$C\big(\torus^{d-1}\big)$}.
\end{Lemma}
\begin{proof}
If $x\in\mathbb{F}_*$, then $w_0:=\Sigma^{-1}\,\varphi_x\cdot z_0$ only
shifts the corresponding component of $ z$ in
$(\unit_{d-1}-\proj_{ z_0})\cdot\R^{d-1}$ in the $\R_{z}$-action
\eqref{eq:Rz0action}. In this case any
element $(z,x)\in\sfG$ may be factorized as
$$
(z,x) = \big(\proj_{ z_0}\cdot
z+(\unit_{d-1}-\proj_{z_0}-\proj_{w_0})\cdot z\,,\,x\big) \,
\big(\langle w_0,z\rangle\, z_0 \,,\,0\big) \ .
$$
Hence the coset space $\R^{d-1}/\R_{ z_0}\simeq\R^{d-2}$ can be
parameterized by $\proj_{ z_0}\cdot
z+\big(\unit_{d-1}-\proj_{z_0}-\proj_{w_0}\big)\cdot z$ with $z\in\R^{d-1}$,
which we decompose correspondingly into a direct product
$\R_{z_0}\times\R^{d-3}$ (with the second factor absent for $d=3$). By \eqref{eq:Mostowtorus} the discrete group
$\Z_v^{d-2}$ acts trivially on the line $\R_{z_0}$ and by translations
on $\R^{d-3}$,
while~$\Z_{ z}$ acts by translations on $\R_{z_0}$ and trivially
on $\R^{d-3}$. Then the crossed product on the right-hand side of~\eqref{eq:Mostowschemefiber} may be unravelled to get
\begin{align*}
C_0\big(\real^{d-1}/\R_{ z_0}\big)\rtimes_{\lt^x}\Z^{d-1}
&\simeq C_0\big(\R_{
 z_0}\times\R^{d-3}\big)\rtimes_{\lt^x}\big(\Z_v^{d-2}\times\Z_{ z_0}\big) \\
&\simeq
\big[\big(C_0(\R_{z_0})\otimes C_0\big(\R^{d-3}\big)\big)\rtimes_{\Id\otimes\lt^x}\Z^{d-2}_v\big]\rtimes_{\beta^x}\Z_{z_0} \\
& \ \mbf\sim_{\text{\tiny M}} \ \big(C_0(\R_{ z_0})\otimes C^*(\Z)\otimes C\big(\torus^{d-3}\big)\big)\rtimes_{\lt^x\otimes\Id\otimes\Id}\Z_{ z_0}
 \\
&\simeq \big(C_0(\R_{ z_0})\rtimes_{\lt^x}\Z_{ z_0}\big)\otimes
C^*(\Z)\otimes C\big(\torus^{d-3}\big)
\nonumber\\
& \ \mbf\sim_{\text{\tiny M}} \ C(\torus)\otimes C^*(\Z)\otimes C\big(\torus^{d-3}\big) \simeq C\big(\torus^{d-1}\big) .
\end{align*}
In the second line we applied
Theorem~\ref{cpsp3}. In the third line we used
Example~\ref{ex:Moritatori} together with the fact that the
homomorphism $\sigma_{\Z_{ z_0}}$ from \eqref{sigmaI} is trivial since the
groups $\Z_{ z_0}$ and $\Z_v^{d-2}$ are discrete, and so the action of $\Z_{ z_0}$ on the crossed
product is induced by left multiplication on $\R_{ z_0}$, the trivial
action on $\R^{d-3}$, and the
trivial action on $\Z_v^{d-2}$. In the fifth line we used
Example~\ref{ex:Moritatori} again.
\end{proof}

The central result of this paper is
\begin{Theorem}
\label{thm:NCtorusbundlegen}
The $C^*$-algebraic T-dual of any almost abelian
solvmanifold $\torus_{\Lambda_\sfG}$ with respect to an $\R_z$-action
is Morita equivalent to a $C^*$-algebra bundle
of noncommutative tori over the circle $\torus$:
\begin{gather*}
C(\torus_{\Lambda_\sfG})\rtimes_\rt\real_{ z_0}
 \ \mbf\sim_{\text{\tiny M}} \ \mcoprod
\torus^{d-1}_{\vec\theta_{z_0}(x)} ,
\end{gather*}
where the noncommutativity parameters $\vec\theta_{z_0}(x)\in\R^{d-2}$ are given
by
\begin{gather}\label{eq:thetaz0}
\vec\theta_{z_0}(x) = \begin{cases}
 0 & \mbox{for} \ x \in \F_* , \\
 \dfrac{(\unit_{d-1}-\proj_{ z_0}) \Sigma^{-1} \varphi_{x_0 x}\cdot
 z_0}{\big\langle{ z_0},\Sigma^{-1} \varphi_{x_0 x}\cdot z_0\big\rangle}
 & \mbox{for} \ x\in\R\setminus\F_* . \end{cases}
\end{gather}
\end{Theorem}
\begin{proof}
That the fibres over $x\in\F_*$ are just ordinary tori $\torus^{d-1}$
is established by Lemma~\ref{lem:ellfibercom}, so we may assume that $x\in\R\setminus\F_*$. Then a simple
calculation shows that any element $(z,x)\in\sfG$ can be factorized as
\[
(z,x) = (v,x) \left(\frac{\proj_{ z_0}\cdot z}{\big\langle z_0,\Sigma^{-1} \varphi_x\cdot z_0\big\rangle} , 0\right) ,
\]
where
\begin{gather}\label{eq:Mostowcoset}
v:=\big(\unit_{d-1}-\proj_{ z_0}\big)\cdot z - \frac{\langle{ z_0},
 z\rangle}{\big\langle{ z_0},\Sigma^{-1} \varphi_x\cdot z_0\big\rangle}
\big(\unit_{d-1}-\proj_{ z_0}\big) \Sigma^{-1} \varphi_x\cdot z_0 .
\end{gather}
Thus the coset space $\R_{x,v}^{d-2}:=\R^{d-1}/\R_{ z_0}$ may be
parameterized by the coordinates $v\in\R^{d-2}$ over any
$x\in\R\setminus\F_*$, and we explicitly retain the fibre index in the
notation for convenience.

We now need to unravel the crossed product
on the right-hand side of~\eqref{eq:Mostowschemefiber}.
From~\eqref{eq:Mostowtorus} and~\eqref{eq:Mostowcoset} it follows that the action of elements
$(\gamma_v,\gamma_{ z_0})\in\Z^{d-1}\simeq\Z_v^{d-2}\oplus\Z_{
 z_0}$ on the coset is given by
\begin{gather}\label{eq:Zcosetactionv}
(\gamma_v,0,0)\cdot (v,x) = (v+\gamma_v,x) , \\
(0,\gamma_{ z_0},0)\cdot (v,x) = \left(v-\frac{
 \gamma_{ z_0}}{\big\langle{ z_0},\Sigma^{-1} \varphi_x\cdot z_0\big\rangle}
\big(\unit_{d-1}-\proj_{ z_0}\big) \Sigma^{-1} \varphi_x\cdot
 z_0 , x\right) .
\label{eq:Zcosetactionz}\end{gather}
Applying
Theorem~\ref{cpsp3} as in the proof of Lemma~\ref{lem:ellfibercom} gives
\begin{gather*}
C\big(\R^{d-1}/\Z^{d-1}\big)\rtimes_{\rt^x}\real_{ z_0}
\ \mbf\sim_{\text{\tiny M}} \
\big(C_0\big(\real^{d-2}_{x,v}\big)\rtimes_{\lt^x}\Z^{d-2}_{v}\big) \rtimes_{\beta^x} \Z_{ z_0} ,
\end{gather*}
where the action of $\Z^{d-2}_v$ on the coset $\R^{d-2}_{x,v}$ is
given by~\eqref{eq:Zcosetactionv}, while the action of $\Z_{ z_0}$ on the crossed product
$C_0\big(\real^{d-2}_{x,v}\big)\rtimes_{\lt^x}\Z^{d-2}_{v}$ is induced by the
action on $\real^{d-2}_{x,v}$ given in~\eqref{eq:Zcosetactionz} and the trivial action on~$\Z^{d-2}_{v}$.

Next, Example~\ref{ex:Moritatori} yields
\begin{gather}\label{eq:NMoritav}
C_0\big(\real_{x,v}^{d-2}\big)\rtimes_{\lt^x} \Z_{v}^{d-2} \ \mbf\sim_{\text{\tiny M}} \
C\big(\torus^{d-2}\big) .
\end{gather}
Since the homomorphism
$\sigma_{\Z_{ z_0}}\colon \Z_{ z_0}\to\real^+$ is trivial, it is not difficult to see
that the Morita equivalence bimodule $\Mcal$ implementing this
equivalence is $\Z_{ z_0}$-equivariant (in the sense of Theorem~\ref{equivMorita}):
the $\Z_{ z_0}$-action
$U\colon \Z_{ z_0}\rightarrow \Aut(\Mcal)$ is given by $U_{\gamma_v}(\xi)(v)=\xi(v+\gamma_v)$, for all $\gamma_v\in
\Z_{ z_0}$, $\xi\in \Mcal$ and $v\in \real_{x,v}^{d-2} $.
Applying Theorem \ref{equivMorita} we conclude that the
Morita equivalence \eqref{eq:NMoritav} induces the Morita equivalence
\begin{gather}\label{eq:fibreNCMostow}
\big(C_0\big(\real^{d-2}_{x,v}\big)\rtimes_{\lt^x}\Z^{d-2}_{v}\big) \rtimes_{\beta^x}
\Z_{ z_0} \ \mbf\sim_{\text{\tiny M}} \
C\big(\torus^{d-2}\big)\rtimes_{\lt^x}\Z_{ z_0} ,
\end{gather}
where the action
of the group $\Z_{ z_0}$ on $C\big(\torus^{d-2}\big)$ is the pullback of the action on
$\torus^{d-2}=\real_{x,v}^{d-2}/\Z_{v}^{d-2}$ induced
from~\eqref{eq:Zcosetactionz}.

After rescaling $x$ by $x_0$
to give it unit period, this shows that the
corresponding algebra of functions on the fiber at $\e^{2\pi\ii x}$ is that
of a noncommutative $d{-}1$-torus $\torus_{\vec\theta_{
 z_0}(x)}^{d-1}$ with noncommutativity parameter
$\vec\theta_{z_0}(x)$ given by~\eqref{eq:thetaz0}, see Example~\ref{ex:NCtorid}.
Thus the $C^*$-algebraic T-dual
$C(\torus_{\Lambda_\sfG})\rtimes_\rt\real_{ z_0}$ is a $C^*$-algebra bundle
of noncommutative tori $\A_{\vec\theta_{
 z_0}(x)}=\torus^{d-1}_{\vec\theta_{ z_0}(x)}$ over the
circle $\torus=\big\{\e^{2\pi\ii x}\, | \, x\in \real/\Z\big\}$.
\end{proof}

What becomes of the non-trivial monodromy \eqref{eq:Mostowtorus} of
the $\torus^{d-1}$ fibers parameterized by~$z$ in the original Mostow
bundle? To answer this question, we write the monodromy matrix
$\ttM=(m_{ij})\in\sfSL(d-1,\Z)$ in the block form
\begin{gather}\label{eq:ttMblock}
\ttM = \left(\begin{matrix} \ttM_{|d-2} & \vec m \\
\vec m' & m_{d-1 \, d-1}
\end{matrix}\right) ,
\end{gather}
where $\ttM_{|d-2}=(m_{ij})_{1\leq i,j\leq d-2}$, while $\vec m=(m_{i \,
 d-1})_{i=1,\dots,d-2}$ and $\vec m' = (m_{d-1 \,
 i})_{i=1,\dots,d-2}$ are respectively column and row vectors in $\Z^{d-2}$. Below
we denote the usual Euclidean inner product on $\R^{d-2}$ by
$\langle\,\cdot\,,\,\cdot\,\rangle$.
\begin{Proposition}\label{prop:monodromyprop}
The noncommutativity parameter
$\vec\theta_{\vec e_{d-1}}(x)$ from~\eqref{eq:thetaz0} varies with a
change of coset representative $x\in\R/\Z$ according to its
$\sfSL(d-1,\Z)$ orbit under the action of the monodromy matrix
\eqref{eq:ttMblock} by $(d{-}2)$-dimensional linear fractional transformations
\begin{gather}\label{eq:thetaMoritagen}
\vec\theta_{\vec e_{d-1}}(x+1) = \ttM\big[\vec\theta_{\vec
 e_{d-1}}(x)\big] :=
\frac{\ttM_{|d-2}\cdot\vec\theta_{\vec e_{d-1}}(x) + \vec m}
{\big\langle \vec m',\vec\theta_{\vec e_{d-1}}(x)\big\rangle + m_{d-1 \,
 d-1}} .
\end{gather}
Under a change of basis of $\Z^{d-1}$ given
by a matrix $B\in\sfGL(d-1,\Z)$, the corresponding noncommutativity
parameter varies according to
\[
\vec\theta_{B\cdot\vec e_{d-1}}(x+1) = \big(B^{\rm
 t} \ttM B\big)\big[\vec\theta_{B\cdot\vec e_{d-1}}(x)\big] .
\]
\end{Proposition}
\begin{proof}
The key feature stems from the definition \eqref{eq:monodromydef} of the monodromy matrix
$\ttM$:
\[
\Sigma^{-1}\,\varphi_{x+x_0} = \Sigma^{-1} \varphi_{x_0} \varphi_x =
\ttM \Sigma^{-1} \varphi_x .
\]
The statements then follow from straightforward calculations in components.
\end{proof}

\begin{Remark}
\label{rem:MoritaKK}
The right-hand side of \eqref{eq:thetaMoritagen} is an example of a
linear fractional transformation in higher dimensions, known from
complex analysis, see, e.g.,~\cite{Cowen2000}. We can show that it defines a Morita equivalence of noncommutative $d{-}1$-tori from
Example~\ref{ex:NCTdMorita}. For this, we introduce the skew-symmetric
matrix $\Theta$ corresponding to the vector $\vec\theta\in\R^{d-2}$ as
in Example~\ref{ex:NCtorid}:
\[
\Theta = \left(
\begin{matrix}
\mbf 0_{d-2} & \vec\theta \\ -\vec\theta\,^{\rm t} & 0
\end{matrix}
\right) .
\]
We denote by $\mbf g_{\ttM}$ the element of ${\sf SO}(d-1,d-1;\Z)$
corresponding to the monodromy matrix $\ttM\in\sfSL(d-1,\Z)$:
\[
\mbf g_{\ttM} = \left(
\begin{matrix}
\ttM & \mbf 0_{d-1} \\ \mbf 0_{d-1} & \big(\ttM^{\rm t}\big)^{-1}
\end{matrix}
\right)
\in {\sf SO}(d-1,d-1;\Z) .
\]
Introduce matrices $\mbf T_{d-1}$ of order $2$ with determinant $-1$
by
\[
\mbf T_{d-1} = \left(
\begin{matrix}
\unit_{d-1}-{\tt E}_{d-1} & {\tt E}_{d-1} \\ {\tt E}_{d-1} &
\unit_{d-1}-{\tt E}_{d-1}
\end{matrix}
\right) \in {\sf O}(d-1,d-1;\Z) ,
\]
where ${\tt E}_{d-1}$ is the matrix unit whose only
non-zero element is $({\tt E}_{d-1})_{d-1\,d-1} = 1$.
Finally, define the element $\mbf M\in{\sf SO}(d-1,d-1;\Z)$ by
\[
\mbf M = \mbf T_{d-1} \mbf g_{\ttM} \mbf T_{d-1} .
\]
Using the adjugate formula
for matrix inverses, a straightforward if tedious calculation then shows that the
corresponding ${\sf SO}(d-1,d-1;\Z)$ orbit of $\Theta$ reproduces the
$(d{-}2)$-dimensional linear fractional transformation of
\eqref{eq:thetaMoritagen}:
\[
\mbf M[\Theta] = \left(
\begin{matrix}
\mbf 0_{d-2} & \ttM\big[\vec\theta\, \big] \\
-\ttM\big[\vec\theta\, \big]^{\rm t} & 0
\end{matrix}
\right) \qquad \mbox{with} \quad \ttM\big[\vec\theta\, \big] = \frac{\ttM_{|d-2}\cdot\vec\theta + \vec m}
{\big\langle \vec m',\vec\theta\, \big\rangle + m_{d-1 \,
 d-1}} .
\]
It follows that, while the fibre noncommutative
tori of Theorem~\ref{thm:NCtorusbundlegen} are not generally identical under a change of
representative $x\in\R/\Z$, by Example~\ref{ex:NCTdMorita}
they are always Morita
equivalent:
\[
\torus^{d-1}_{\vec\theta_{\vec e_{d-1}}(x+1)} \ \mbf\sim_{\text{\tiny
 M}} \ \torus^{d-1}_{\vec\theta_{\vec e_{d-1}(x)}} .
\]
Recalling that our formulation of topological T-duality takes place in the additive category~${\CCK\hspace{-1mm}\CCK}$ from
Section~\ref{sec:TdualityKK}, this has a natural interpretation: The
non-trivial monodromy in the automorphism group of the~\smash{$\torus^{d-1}$} fibres of the
twisted torus $\torus_{\Lambda_\sfG}$ is
manifested as a (generally non-trivial) isomorphism in the Morita automorphism
group of the fibre noncommutative tori in~${\CCK\hspace{-1mm}\CCK}$.
\end{Remark}

\section{Three-dimensional solvmanifolds and their T-duals}\label{sec:3dexamples}

The goal of this final section is to give some explicit examples in
low dimensions of the general formalism we have developed in
Section~\ref{sec:TdualityMostow}, recovering some previously known
results in the literature from a new perspective, as well as providing several new examples.
Note that tori~$\torus^d$ in any
dimension $d$ are covered by our general framework of
Section~\ref{sec:RMostow} for~$\R_y$-actions: when
\mbox{$\varphi_x=\unit_{d-1}$} for all~$x\in\R$, the torus bundles
$\torus^d\to\torus^{d-n}$ of Proposition~\ref{lem:quotienttorus} are
\mbox{trivial}. Theorem~\ref{prop:RyTduality} then shows that any
$C^*$-algebraic T-dual $C\big(\torus^d\big)\rtimes_\rt\R_{y_0}$ is isomorphic
to $\C\T\big(\torus^{d-1}\times\torus_{y_0},0\big)\simeq
C\big(\torus^{d-1}\times\torus_{y_0}\big)\otimes\CK$. Below we
apply our formalism to the well-known Mostow fibrations of
three-dimensional solvmanifolds.

\subsection[Mostow bundles and $\sfSL(2,\R)$ conjugacy classes]{Mostow bundles and $\boldsymbol{\sfSL(2,\R)}$ conjugacy classes}

In three dimensions, solvmanifolds are completely
classified, see, e.g.,~\cite{Bock2009}. In particular,
by~\cite[Proposition~5.1]{Bock2009} all solvmanifolds in
this dimension which are based on connected and simply-connected Lie groups
$\sfG$ are almost abelian. If $\sfG=\R^2\rtimes_\varphi\R$ is
abelian, i.e., $\varphi_x=\unit_2$ for all $x\in\R$, then $\torus_{\Lambda_\sfG}\simeq\torus^3$ is a torus, while
the remaining cases correspond to non-trivial one-parameter group actions
\[
\varphi\colon \ \R\longrightarrow\sfSL(2,\real) .
\]
Elements of $\sfSL(2,\R)$ are classified up to conjugacy by trace, so there are three classes, determined by which of the three types
of conjugacy classes of $\sfSL(2,\real)$ that
the image of the homomorphism~$\varphi$ lands in: parabolic
($|\Tr\varphi_x|=2$ for all $x\in\R$), elliptic ($|\Tr\varphi_x|<2$
for all $x\not=k\,\pi/2$ with $k\in \mathbb{Z}$),
or hyperbolic ($|\Tr\varphi_x|>2$ for all
$x\not=0$), see, e.g.,~\cite{Hull2005}.

The existence of cocompact discrete subgroups $\Lambda_\sfG$ of $\sfG$ is highly
restrictive on the allowed homomorphisms $\varphi$; a necessary condition is that there exists
$x_0\in\R^\times$ such that $\Tr\varphi_{x_0}\in\Z$. The monodromy matrices
\begin{gather}\label{eq:monodromy3d}
\ttM=\Sigma^{-1}\,\varphi_{x_0} \Sigma=\left(\begin{matrix} a & b \\
 c & d \end{matrix}\right) \in \sfSL(2,\Z)
\end{gather}
are integer matrices that live in the corresponding conjugacy classes
of $\sfSL(2,\Z)$. As automorphisms of the torus fibres of the Mostow bundle, in the mapping
class group $\sfSL(2,\Z)$, they act on $\torus^2$ in the following
way. Let $\big(\lambda,\lambda^{-1}\big)$ be the eigenvalues of $\ttM$, which
are the roots of the characteristic polynomial $t^2-(\Tr\ttM) t+1$. Then
there are three possibilities:
\begin{itemize}\itemsep=0pt
\item{$\underline{\rm Parabolic}$:} In this case
 $\Tr\ttM=\pm\,2$ and $\lambda=\lambda^{-1}=\pm\,1$. Then $\ttM$ has
 an integral eigenvector corresponding to a closed curve on
 $\torus^2$ which is invariant under the
 associated automorphism, and homeomorphisms of this type are Dehn twists of
 $\torus^2$.
\item{$\underline{\rm Elliptic}$:} In this case $\Tr\ttM=0,\pm 1$,
 and the eigenvalues $\big(\lambda,\lambda^{-1}\big)$ are complex of
 modulus~$1$. Then the associated automorphism has finite order
 equal to $2$, $3$, $4$ or $6$, and corresponds to a periodic
 homeomorphism of~$\torus^2$.
\item{$\underline{\rm Hyperbolic}$:} In this case
 $|\lambda|>1>\big|\lambda^{-1}\big|$. The associated automorphism has no
 invariant closed curve on $\torus^2$; it `stretches' the
 eigenspace corresponding to $\lambda$ and `contracts' the eigenspace
 corresponding to~$\lambda^{-1}$. These are called Anasov homeomorphisms of~$\torus^2$.
\end{itemize}

In this section we shall apply Theorem~\ref{thm:NCtorusbundlegen} with
$z_0=\vec e_2=(0,1)$, which for convenience we combine with
Proposition~\ref{prop:monodromyprop} and Remark~\ref{rem:MoritaKK} to reformulate it as
\begin{Theorem}\label{thm:NCtorusbundle3d}
The $C^*$-algebraic T-dual of a three-dimensional
solvmanifold $\torus_{\Lambda_\sfG}$ with respect to an $\R_z$-action
with $z_0=(0,1)$
is Morita equivalent to a $C^*$-algebra bundle
of noncommutative two-tori over the circle $\torus$:
\begin{gather*}
C(\torus_{\Lambda_\sfG})\rtimes_\rt\real_{(0,1)}
\ \mbf\sim_{\text{\tiny M}} \ \mcoprod \,
\torus^2_{\theta(x)} ,
\end{gather*}
where the noncommutativity parameters $\theta(x)\in\R$ are given by
\begin{gather*}
\theta(x) = \begin{cases}
 0 & \mbox{for} \ x \in \F_* , \\
\dfrac{\big(\Sigma^{-1} \varphi_{x_0 x}\big)_{12}}{\big(\Sigma^{-1} \varphi_{x_0 x}\big)_{22}}
 & \mbox{for} \ x\in\R\setminus\F_* , \end{cases}
\end{gather*}
where $\big(\Sigma^{-1} \varphi_{x_0 x}\big)_{ij}$ for $i,j\in\{1,2\}$ denote
the matrix elements of $\Sigma^{-1} \varphi_{x_0 x}$. The
noncommutativity parameter $\theta(x)$ varies under a
change of coset representative $x\in\R/\Z$
according to its $\sfSL(2,\Z)$ orbit under the corresponding monodromy
matrix~\eqref{eq:monodromy3d}:
\begin{gather}\label{eq:thetaMorita3d}
\theta(x+1) = \ttM\big[\theta(x)\big] =
\frac{a \theta(x)+b}{c \theta(x)+d} .
\end{gather}
By Example~{\rm \ref{ex:NCT2Morita}} the corresponding fibre noncommutative
tori are Morita equivalent:
\[
\torus^2_{\theta(x+1)} \ \mbf\sim_{\text{\tiny
 M}} \ \torus^2_{\theta(x)},
\]
and so isomorphic in the category ${\CCK\hspace{-1mm}\CCK}$.
\end{Theorem}

The $C^*$-algebraic T-dual in these instances has more structure in general than the
original algebra of functions $C(\torus_{\Lambda_\sfG})$, as a consequence of
\begin{Proposition}\label{prop:starproduct3d}
The noncommutative torus bundle $\coprod_{x\in\R/\Z}
\torus_{\theta(x)}^2$ is a $C\big(\torus^2\big)$-algebra.
\end{Proposition}
\begin{proof}
We can also describe these noncommutative torus bundles via (strict)
deformation quantization: Performing an identical calculation to that of
Example~\ref{ex:NCtori}, we may describe the fibrewise crossed product~\eqref{eq:fibreNCMostow} (with $d=3$) explicitly in terms of a
star-product, and hence the $C^*$-algebra bundle of Theorem~\ref{thm:NCtorusbundle3d}
as a deformation of the algebra of functions
$C\big(\torus\times\torus^2\big)$ on the trivial torus bundle over $\torus$ by regarding the convolution algebra $C_{\rm
 c}(\torus\times\torus\times\zed)$ as the space of functions
$C\big(\torus\times\torus^2\big)$ equipped with the star-product
\begin{gather}
(f\star_{\theta} g)(x,v_1,v_2) = \sum_{(p_1,p_2)\in\zed^2}
\e^{2\pi\ii (p_1 v_1+p_2 v_2)} \nn \\
\hphantom{(f\star_{\theta} g)(x,v_1,v_2) = \sum_{(p_1,p_2)\in\zed^2} }{} \times \sum_{(q_1,q_2)\in\zed^2}
f_{q_1,q_2}(x) g_{p_1-q_1,p_2-q_2}(x) \e^{2\pi\ii
 \theta(x) (p_1-q_1) q_2},\label{eq:parabolicstar}
\end{gather}
where we used the fibrewise Fourier transformation
\[
f(x,v_1,v_2)=\sum_{(p_1,p_2)\in\zed^2} f_{p_1,p_2}(x) \e^{2\pi\ii(p_1 v_1+
 p_2\,v_2)} ,
\]
with functions $f_{p_1,p_2}\colon \torus\to\complex$.
In this formulation of the
$C^*$-algebraic T-dual, it is immediately evident that there is an injection
$C\big(\torus^2\big)\hookrightarrow \coprod_{x\in\R/\Z}\,
\torus_{\theta(x)}^2$ as the space of functions which are independent of
the coordinate~$v_1$ (or~$v_2$); for such functions, the sums in~\eqref{eq:parabolicstar} truncate to $p_1=q_1=0$ (or $p_2=q_2=0$) and
the noncommutative star-product reduces to the commutative pointwise
product of functions in $C\big(\torus^2\big)$. This defines $C^*$-algebra
monomorphisms making $\coprod_{x\in\R/\Z}\,
\torus_{\theta(x)}^2$ into a $C\big(\torus^2\big)$-algebra.
\end{proof}

The remainder of this paper is devoted to providing illustrations of
Theorem~\ref{thm:NCtorusbundle3d}, through explicit calculations in
each of the three conjugacy classes of~$\sfSL(2,\R)$. The features
differ for each conjugacy class so we consider them individually in
turn.

\subsection{Parabolic torus bundles}\label{sec:parabolicbundles}

We start with the best
studied example in the literature, which is based on the nilpotent Heisenberg
group. In string theory it is T-dual to the three-torus $\torus^3$ with $H$-flux by
the standard Buscher rules (see, e.g.,~\cite{Hull2009}), and in
topological T-duality it gives the basic example of a noncommutative
principal torus
bundle~\cite{Echterhoff2009,Hannabuss2010, Mathai2004}. Here we shall give a new algebraic
perspective on
both these T-duals by applying our formalism of topological T-duality
directly to the Heisenberg nilmanifold. This class has several
features that make it special among the three-dimensional
solvmanifolds, which we explain in detail.

\subsubsection*{Heisenberg nilmanifolds}

The three-dimensional Heisenberg group $\Heis(3)$ is the nilpotent
Lie group whose Mostow bundle
structure is based on the semi-direct product
\[
\Heis(3) = \real^2\rtimes_\varphi\real \qquad \mbox{with} \quad
\varphi_x = \left( \begin{matrix}
1 & x \\
0 & 1
\end{matrix} \right) .
\]
Since $|\Tr\varphi_x|=2$, the matrix $\varphi_x$ parameterizes a
parabolic conjugacy class of $\sfSL(2,\real)$.
Here we denote the coordinates on the group manifold of $\sfN=\real^2$ by
$(y,z)$. The group multiplication on~$\Heis(3)$ is then given by
\[
(x,y,z) (x',y',z' ) = (x+x',y+y'+x z', z+z' ) ,
\]
and the inverse of a group element is
\[
(x,y,z)^{-1}= (-x,x z-y,-z) .
\]

In this case it is clear that for
\[
x_0=m
\]
with $m\in\Z^\times$, the matrix $\varphi_{x_0}$ is integer-valued and hence
may be taken as monodromy matrix~${\tt M}_m$ of infinite order with $\Sigma=\unit_2$:
\begin{gather}\label{eq:parmon}
\ttM_m = \varphi_m = \left( \begin{matrix}
1 & m \\
0 & 1
\end{matrix} \right) .
\end{gather}
The \emph{Heisenberg nilmanifold} $\torus_{\Heis_m(3;\Z)}$ is the compact
space obtained as the quotient of~$\Heis(3)$ with respect to the
lattice given by the discrete Heisenberg group
\[
\Heis_m(3;\Z):=\big\{(m\,\alpha,\beta,\gamma)\in\Heis(3)\, \big|\,
\alpha,\beta,\gamma\in\zed\big\} .
\]
The equivalence relation
is given by the left action of $\Heis_m(3;\Z)$, which
leads to the local coordinate identifications under the action of the
generators of $\Heis_m(3;\Z)$ given by
\begin{gather}
(x,y,z)\longmapsto (x+m,y+m\,z,z) , \nonumber \\
(x,y,z)\longmapsto (x,y+1,z) , \nonumber \\
(x,y,z)\longmapsto (x,y,z+1) .\label{eq:H3integeraction}
\end{gather}
Geometrically, this exhibits the Heisenberg nilmanifold as a
non-trivial principal $\torus^2$-bundle
$\torus_{\Heis_m(3;\Z)}\rightarrow \torus$.
Using the algorithm described in Section~\ref{sec:RMostow}, the
relation matrix
\[
{\tt A}_m = \ttM_m-\unit_2
\]
has rank $r=1$ with elementary divisors
$m_1=m$ and $m_2=0$, and the first homology group of~$\torus_{\Heis_m(3;\Z)}$ can thus be presented as the $\Z$-module
\begin{gather*}
{\rm H}_1(\torus_{\Heis_m(3;\Z)},\Z) \simeq \Z
\oplus(\Z\oplus\Z_m) ,
\end{gather*}
with respective free generators denoted $e_x$ and $e_y$, and torsion
generator $e_z$ of order~$m$, $m e_z=0$.

Dually, the topology of the nilmanifold is captured by its algebra of
functions, which can be computed as the subalgebra of invariant
functions on the
Heisenberg group $\Heis(3)$ with respect to the left action of the
lattice $\Heis_m(3;\Z)$:
\[
C\big(\torus_{\Heis_m(3;\Z)}\big) =
C_0\big(\Heis(3)\big)^{\Heis_m(3;\Z)} .
\]
Harmonic analysis on the nilmanifold can be used to determine the
Fourier decomposition of any function $\mathtt{f}\in
C(\torus_{\Heis_m(3;\Z)})$. Invariance under the generators
\eqref{eq:H3integeraction} of $\Heis_m(3;\Z)$ forces the expansion to take the form
\begin{gather}\label{eq:nilfoldexpansion}
\mathtt{f}(x,y,z) = \sum_{k\in\Z_m} \sum_{p,q\in\Z} \mathtt{f}_{p,k}(x+m q)
\e^{2\pi\ii p y +2\pi\ii (k+m p q) z} ,
\end{gather}
for functions $\mathtt{f}_{p,k}\colon \real\to\complex$ vanishing at
infinity. From this expression we can easily see the
$C(\torus)$-algebra structure of $C(\torus_{\Heis_m(3;\Z)})$: Restricting \eqref{eq:nilfoldexpansion} to $y=z=0$
determines a function on the circle $\torus=\big\{\big(\e^{2\pi\ii x/m},1,1\big)\in\torus_{\Heis_m(3;\Z)}\big\}$ and defines a $C^*$-algebra monomorphism
$C(\torus)\hookrightarrow C(\torus_{\Heis_m(3;\Z)})$. On the other hand, the evaluation
of \eqref{eq:nilfoldexpansion} at $z=0$ yields a function on the
two-torus $\torus^2=\big\{\big(\e^{2\pi\ii x/m},\e^{2\pi\ii
 y},1\big)\in\torus_{\Heis_m(3;\Z)}\big\}$ and defines a
$C^*$-algebra monomorphism $C\big(\torus^2\big)\hookrightarrow
C(\torus_{\Heis_m(3;\Z)})$ making $C(\torus_{\Heis_m(3;\Z)})$ into a
 $C\big(\torus^2\big)$-algebra in this case.

Following the prescriptions of Sections~\ref{sec:Rygen} and~\ref{sec:Rzgen}, we shall now study the
known T-duals of the Heisenberg nilmanifold
through the formalism of topological T-duality discussed
in Section~\ref{sec:TopTtorus}.

\subsubsection*{$\boldsymbol{\real_y}$-action: $\boldsymbol{\torus^3}$ with $\boldsymbol{H}$-flux}

The one-parameter subgroup
\[
\ker({\tt A}_m) = \real_{(1,0)}:=\big\{(0,\lambda,0) \in \Heis(3) \big\}
\]
is the center of $\Heis(3)$. The quotient group $\Heis(3)/\R_{(1,0)}$ can be parameterized by equivalence
classes $\big[(x,0,z)\big]$, and the
multiplication law
\[
\big[(x,0,z)\big]\,\big[(x',0,z')\big] = \big[(x+x',0,z+z')\big]
\]
is that of the abelian Lie group $\Heis(3)/\R_{(1,0)}\simeq\R^2$. It follows that the quotient
$\torus_{\Heis_m(3;\Z)}/\real_{(1,0)}$ is the two-dimensional torus
$\torus^2$ with coordinates $\big(\e^{2\pi\ii x/m},\e^{2\pi\ii z}\big)$. The corresponding principal circle bundle
$p_{(1,0)}\colon \torus_{\Heis_m(3;\Z)}\to{ \torus^2}$, with fibre coordinate $\e^{2\pi\ii y}$, is the standard realization of
the Heisenberg nilmanifold as a $\torus$-fibration over $\torus^2$ of
degree~$m$: the connection
\[
\kappa_{(1,0)} = -\dd y + x \dd z
\]
of Theorem~\ref{prop:RyTduality} has
curvature
\begin{gather*}
\dd\kappa_{(1,0)} = \dd x\wedge \dd z ,
\end{gather*}
and so the first Chern class $c_1(p_{(1,0)})$ is $m$ times the standard generator of
${\rm H}^2\big(\torus^2,\Z\big)\simeq\Z$.
By Theorem~\ref{prop:RyTduality} it follows that the $C^*$-algebraic
T-dual
\[
C\big(\torus_{\Heis_m(3;\Z)}\big) \rtimes_\rt \real_{(1,0)} \simeq \C\T\big(\torus^3,m\big)
\]
is a continuous-trace algebra with spectrum~$\torus^3$ and
Dixmier--Douady class equal to~$m$ times the standard generator of ${\rm
 H}^3\big(\torus^3,\Z\big)\simeq\Z$. Thus in this case we recover the
standard T-duality between the three-torus $\torus^3$ with $H$-flux
and the Heisenberg nilmanifold viewed as a circle bundle.
The correspondence space construction is given by
Proposition~\ref{prop:Mostowcorr} with
$\torus_{\Lambda_{\sfG^{(1,0)}}}=\torus^2$.

\subsubsection*{$\boldsymbol{\real_z}$-action: Noncommutative principal torus bundles}\label{sec:nilNCtorus}

We shall now recover the known noncommutative principal torus bundle
which is T-dual to $\torus_{\Heis_m(3;\Z)}$ from a new
perspective, by directly working with the algebraic description of the
nilmanifold. For this, we follow our prescription for $\R_z$-actions
from Section~\ref{sec:Rzgen} and consider the one-parameter subgroup
$\real_{(0,1)}$ of $\Heis(3)$ given by
\[
\real_{(0,1)} = \big\{(0,0,\zeta)\in\Heis(3) \big\} ,
\]
whose right action by multiplication generates
\[
(x,y,z) (0,0,\zeta) = (x,y+x \zeta,z+\zeta) .
\]
The quotient of this action to $\torus_{\Heis_m(3;\Z)}$ fixes
every point of the form $\big(\e^{2\pi\ii n/m},\e^{2\pi\ii
 y},\e^{2\pi\ii z}\big)\in\torus^2\subset \torus_{\Heis_m(3;\Z)}$ for
$n\in\Z$ with isotropy
subgroup
\[
\Z_{(0,1)} = \real_{(0,1)}\cap\Heis_m(3;\Z) .
\]
In this case $\F_*=\varnothing$, and Theorem~\ref{thm:NCtorusbundle3d} thus identifies the $C^*$-algebraic
T-dual through the Morita equivalence{\samepage
\[
C\big(\torus_{\Heis_m(3;\Z)}\big)\rtimes_\rt\real_{(0,1)} \ \mbf\sim_{\text{\tiny M}} \
\mcoprod \torus_{m\,x}^2
\]
of {$C(\torus)$-algebras}.}

We can
explicitly check the anticipated monodromy \eqref{eq:thetaMorita3d}:
Under a change of coset representative $x\in\real/\Z$
of the fibres of the noncommutative torus bundle
$\coprod_{x\in\R/\Z} \torus^2_{\theta_m(x)}$, the noncommutativity parameter changes
\begin{gather}\label{eq:parthetamon}
\theta_m(x+1) = \theta_m(x)+m = {\tt M}_{m}\big[\theta_m(x)\big]
\end{gather}
by the $\sfSL(2,\Z)$ orbit of $\theta_m(x)=m\,x$ under the monodromy matrix ${\tt M}_{m}$ given by~\eqref{eq:parmon}, and by Example~\ref{ex:NCT2Morita} the fiber
noncommutative tori are identical:
$\torus^2_{\theta_m(x+1)}=\torus^2_{\theta_m(x)}$. Hence in this case the
non-trivial monodromy in
the automorphism group of the $\torus^2$ fibres of the twisted torus
is implemented trivially as an equality in its T-dual $C^*$-algebra
bundle.

It is well known that these $C^*$-bundles are isomorphic to the
convolution $C^*$-algebras of the corresponding integer Heisenberg
groups, and indeed this is precisely the original description of these T-dual
noncommutative torus bundles from~\cite{Mathai2004}. For later
comparison, it is instructive to explicitly demonstrate this result using our
scheme for topological T-duality, which results in
\begin{Proposition}\label{prop:C*Heis}
The $C^*$-algebraic T-dual
$C\big(\torus_{\Heis_m(3;\Z)}\big)\rtimes_\rt\real_{(0,1)}$ is Morita equivalent to
the group $C^*$-algebra $C^*\big(\Heis_m(3;\Z)\big)$.
\end{Proposition}
\begin{proof}
The crux of the proof hinges on the fact that the parabolic torus
bundles are the only class in three dimensions for which a fibrewise
analysis is not necessary and the Morita isomorphisms can be
implemented directly in the category $\CCK\hspace{-1mm}\CCK$.
For this, we observe that any element $(x,y,z)\in\Heis(3)$ can
be uniquely factorized as
\[
(x,y,z) = (x,y-x z,0) (0,0,z) =: (x,v,0) (0,0,z) ,
\]
with the action of the subgroup $\real_{(0,1)}$ on the normal abelian subgroup
$\real_{(x,v)}^2:=\{(x,v,0)\in\Heis(3)\}$ given by
\[
{\rm Ad}_z(x,v,0):={}^{(0,0,z)}(x,v,0) = (0,0,z) (x,v,0) (0,0,-z) = (x,v-x z,0) .
\]
It follows that the Heisenberg group can alternatively be presented as the semi-direct
product $\Heis(3) = \real_{(x,v)}^2 \real_{(0,1)}$. Correspondingly,
its lattice is also a semi-direct product
$\Heis_m(3;\Z)=\Z_{(x,v)}^2\rtimes_{{\rm Ad}}\Z_{(0,1)}$, where
$\Z_{(x,v)}^2=\real_{(x,v)}^2\cap\Heis_m(3;\Z)$. Thus the Morita
equivalence from~\eqref{scheme} can be expressed as
\begin{align}
C\big(\torus_{\Heis_m(3;\Z)}\big)\rtimes_\rt\real_{(0,1)} & \ \mbf\sim_{\text{\tiny M}} \
C_0\big(\Heis(3)/\real_{(0,1)}\big) \rtimes_\lt \Heis_m(3;\Z) \nn\\
& \ \mbf\sim_{\text{\tiny M}} \
C_0\big(\real_{(x,v)}^2\big)\rtimes_\lt \big(\Z_{(x,v)}^2\rtimes_{{\rm
 Ad}}\Z_{(0,1)} \big) \nn\\
&\simeq
\big(C_0\big(\real_{(x,v)}^2\big)\rtimes_{\lt}
\Z_{(x,v)}^2\big) \rtimes_\beta\Z_{(0,1)} ,\label{firststep}
\end{align}
where in the last line we have further applied
Theorem~\ref{cpsp3}. Since the groups $\Z_{(x,v)}^2$ and $\Z_{(0,1)}$
are discrete, the homomorphism $\sigma_{\Z_{(0,1)}}$ from \eqref{sigmaI}
is trivial, and so the $\Z _{(0,1)}$-action on $C_0\big(\real^2_{(x,v)} \big)\rtimes_{\lt} \Z_{(x,v)}^2$ is the
canonical action obtained from the diagonal action of $\Z_{(0,1)}$ on
$\Z_{(x,v)}^2\times \real_{(x,v)}^2 $:
\[
\beta_\zeta(f)\big((m \alpha,\nu) , (x,v)\big) =
f\big((m \alpha,\nu+m \alpha \zeta) , (x,v+x \zeta)\big)
\]
for
all $\zeta\in \Z_{(0,1)}$, $(m \alpha,\nu)\in\Z_{(x,v)}^2$, $(x,v)\in \real_{(x,v)}^2 $ and
$f\in C_{\rm c}\big(\Z_{(x,v)}^2\times \real_{(x,v)}^2 \big)\subset C_0\big(\real_{(x,v)}^2 \big)\rtimes_{\lt}
\Z_{(x,v)}^2$.

By Example~\ref{ex:Moritatori} there is a Morita equivalence
\begin{gather}\label{eq:NMorita}
C_0\big(\real_{(x,v)}^2\big)\rtimes_{\lt} \Z_{(x,v)}^2 \ \mbf\sim_{\text{\tiny M}} \
C\big(\torus^2\big) .
\end{gather}
Since the homomorphism
$\sigma_{\Z_{(0,1)}}\colon \Z_{(0,1)}\to\real^+$ is trivial, it follows
that the Morita equivalence bimodule $\Mcal$ implementing this
equivalence is $\Z_{(0,1)}$-equivariant:
the $\Z_{(0,1)}$-action
$U\colon \Z_{(0,1)}\rightarrow \Aut(\Mcal)$ is given by $U_\zeta(\xi)(x,v)=\xi(x,v+x \zeta)$, for all $\zeta\in
\Z_{(0,1)}$, $\xi\in \Mcal$ and $(x,v)\in \real_{(x,v)}^2 $.
Applying Theorem~\ref{equivMorita} we conclude that the
Morita equivalence~\eqref{eq:NMorita} induces the Morita equivalence
$\big(C_0\big(\real_{(x,v)}^2 \big)\rtimes_{\lt}
\Z_{(x,v)}^2\big) \rtimes_\beta \Z_{(0,1)}\
\mbf\sim_{\text{\tiny M}}\ C\big(\torus^2 \big) \rtimes_{{\rm Ad}^*} \Z_{(0,1)}$, where the action
of the group~$\Z_{(0,1)}$ on $C\big(\torus^2\big)$ is the pullback of the action on
$\torus^2=\real_{(x,v)}^2/\Z_{(x,v)}^2$ induced from the left
multiplication of $\Z_{(0,1)}$ on $\Heis(3)$.

Substituting into~\eqref{firststep} we thus obtain
\begin{gather}\label{T3RT2Z}
C(\torus_{\Heis_m(3;\Z)})\rtimes_\rt\real_{(0,1)} \ \mbf\sim_{\text{\tiny M}} \
C\big(\torus^2 \big) \rtimes_{{\rm Ad}^*} \Z_{(0,1)} .
\end{gather}
Now we implement Pontryagin duality via Fourier transform and apply~\eqref{hathatduality} to the right-hand side of~\eqref{T3RT2Z} to
obtain
\[
C\big(\torus^2\big)\rtimes_{{\rm Ad}^*}\Z_{(0,1)} \simeq
C^*\big(\Z_{(x,v)}^2\big)\rtimes_{\widehat{\rm Ad}{}^*} \Z_{(0,1)} ,
\]
where we used the fact that $\widehat\sigma_{\Z_{(0,1)}}\colon \Z_{(0,1)}\to\real^+$ is
trivial since $\Z_{(0,1)}$ is discrete. Applying
Theorem~\ref{cpsp2}, we then arrive at
\[
C\big(\torus_{\Heis_m(3;\Z)}\big)\rtimes_\rt\real_{(0,1)} \
\mbf\sim_{\text{\tiny M}} \
C^*\big(\Z_{(x,v)}^2\rtimes_{\rm
 Ad}\Z_{(0,1)}\big) = C^*\big(\Heis_m(3;\Z)\big) ,
\]
as required.
\end{proof}

\begin{Remark}
The star-products \eqref{eq:parabolicstar} with
$\theta(x)=\theta_m(x)=m x$ are equivalent to the star-products
of~\cite{Hannabuss2010,Hull2019, Lowe2003}. In~\cite{Echterhoff2009,Hannabuss2010},
it is shown that the description of Proposition~\ref{prop:C*Heis}
defines a noncommutative principal $\torus^2$-bundle
(cf.\ Example~\ref{ex:NCprincipalT}). It is also shown
in~\cite{Echterhoff2009} that the monodromy~\eqref{eq:parthetamon},
while acting trivially at the purely algebraic level, has a~non-trivial action on the K-theory group of the $C^*$-algebra
bundle $\coprod_{x\in\R/\Z} \torus_{m x}^2$, viewed as a bundle of
abelian groups over~$\torus$; a physical picture of this action in
terms of monodromies of fiber D-branes is
discussed in~\cite{Hull2019}.
\end{Remark}

\subsubsection*{Noncommutative correspondences}

We come now to the noncommutative correspondences underlying this
topological T-duality. For this, we recall that both
$C\big(\torus_{\Heis_m(3;\Z)}\big)$ and $C\big(\torus^2\big)\rtimes_{{\rm Ad}^*}\Z_{(0,1)}$
are $C\big(\torus^2\big)$-algebras, with~$\torus^2$
parameterized by the local coordinates~$(x,y)$. We will show that the
noncommutative correspondence is given by the balanced tensor product
\[
C\big(\torus_{\Heis_m(3;\Z)}\big)\rtimes_\rt\Z_{(0,1)} \simeq
C\big(\torus_{\Heis_m(3;\Z)}\big)\otimes_{C(\torus^2)} \mcoprod \torus_{m x}^2
\]
over $C\big(\torus^2\big)$.
Note that here the subgroup $\Z_{(0,1)}$ acts non-trivially on the algebra
of functions $C(\torus_{\Heis_m(3;\Z)})$, in contrast to our previous examples. We consider elements of the convolution
algebra $
C_{\rm c}(\torus_{\Heis_m(3;\Z)}\times\Z_{(0,1)})$ as sequences $f=\{f_{\tilde q}\}_{\tilde q\in\Z_{(0,1)}}$ of functions
$f_{\tilde q}\colon \torus_{\Heis_m(3;\Z)}\to\complex$ with the convolution product
\[
(f\star g)_{\tilde q}(x,y,z) = \sum_{\tilde p\in\Z} f_{\tilde p}(x,y,z)
g_{\tilde q-\tilde p}(x,y-\tilde p\,m\,x,z) .
\]
We write the Fourier transformation
\[
f(x,y,z,\tilde z) := \sum_{\tilde q\in\Z} f_{\tilde q}(x,y,z)
\e^{2\pi\ii \tilde q \tilde z}
\]
for functions on $\torus_{\Heis_m(3;\Z)}\times\widetilde\torus_z$, and define the star-product
\[
(f\star g)(x,y,z,\tilde z) = \sum_{\tilde q\in\Z} (f\star
g)_{\tilde q}(x,y,z) \e^{2\pi\ii \tilde q \tilde z} .
\]
We further use the harmonic expansion of functions on the nilmanifold
given from~\eqref{eq:nilfoldexpansion} by
\[
f_{\tilde q}(x,y,z) = \sum_{k\in\Z_m} \sum_{p,q\in\Z} f_{\tilde
 q;p,k}(x+q) \e^{2\pi\ii q y+2\pi\ii (k+m p q) z} .
\]
After some elementary manipulations, we can then write the star-product on
the space of functions $C\big(\torus_{\Heis_m(3;\Z)}\times\widetilde\torus_z\big)$ as
\begin{gather}
(f\star g)(x,y,z,\tilde z) = \sum_{l\in\Z_m} \sum_{r,s,\tilde
 p\in\Z} \e^{2\pi\ii r y+2\pi\ii(l+m r s)z}
 \e^{2\pi\ii\tilde p \tilde z} \nn \\
 \qquad{} \times \sum_{k\in\Z_m}
 \sum_{p,q,\tilde q\in\Z}
 f_{\tilde q;p,k}(x+q)
 g_{\tilde p-\tilde
 q;r-p,l-k}(x+q-s) \e^{2\pi\ii
 mr(q-s)z}
 \e^{-2\pi\ii
 mx(r-p)\tilde q}.\label{eq:starprodNCcorr}
\end{gather}

We can define an injection $C\big(\torus_{\Heis_m(3;\Z)}\big)\hookrightarrow
C\big(\torus_{\Heis_m(3;\Z)}\times\widetilde\torus_z\big)$ as the space of
functions which are independent of the coordinate $\tilde z$. For
these functions the sums in~\eqref{eq:starprodNCcorr} truncate to
$\tilde p=\tilde q=0$, and one recovers the pointwise product of
functions in $C\big(\torus_{\Heis_m(3;\Z)}\big)$. On the other hand, we can include the
noncommutative algebra $\coprod_{x\in\R/\Z} \torus_{m x}^2 \hookrightarrow
C\big(\torus_{\Heis_m(3;\Z)}\times\widetilde\torus_z\big)$ as the space of
functions which are independent of the coordinate $z$, and one sees
that~\eqref{eq:starprodNCcorr} truncates for $l=k=0$ and $s=q=0$ to the star-product~\eqref{eq:parabolicstar} with $\theta(x)=\theta_m(x)=m x$. Altogether, the noncommutative correspondences are
induced by the diagram
\[
\xymatrix{
& C\big(\torus_{\Heis_m(3;\Z)}\big)\otimes_{C\big(\torus^2\big)} \mcoprod \torus_{m\,x}^2 & \\
C\big(\torus_{\Heis_m(3;\Z)}\big) \ar[ur]& &
 \ \mcoprod \torus_{m\,x}^2
\ar[ul] \\
 & C\big(\torus^2\big) \ar[ur]
 \ar[ul] &
}
\]
where all arrows are $*$-monomorphisms of $C^*$-algebras.

\subsection{Elliptic torus bundles}\label{sec:ellipticbundles}

We shall now illustrate Theorem~\ref{thm:NCtorusbundle3d} for the three-dimensional solvmanifolds
based on the Euclidean group in two dimensions, which have no
classical Hausdorff T-duals and whose noncommutative T-dual fibration was first
discussed in~\cite{Hull2019} for the case of the~$\Z_4$ elliptic
monodromy. Here we shall recover this noncommutative geometry rigorously from
our algebraic framework, and moreover extend it to the~$\Z_2$ and~$\Z_6$ elliptic
monodromies which have not been considered previously.

The Euclidean group $\ISO(2)$ in two dimensions is the
three-dimensional almost abelian solvable Lie group
\begin{gather*}
\ISO(2) = \real^2\rtimes_\varphi\real \qquad \mbox{with} \quad
\varphi_x = \left(
\begin{matrix}
\cos x & \sin x \\
-\sin x & \cos x
\end{matrix}
\right) .
\end{gather*}
Since $|\Tr\varphi_x|<2$ for $x\notin \pi \Z$, the rotations~$\varphi_x$
parameterize an elliptic conjugacy class of~$\sfSL(2,\real)$.
Here we shall denote coordinates on the group manifold of $\sfN=\real^2$ by
$z=(z_1,z_2)$. The group multiplication on $\ISO(2)$ is then
\[
(x,z_1,z_2)\,(x',z_1',z_2') = \big(x+x' ,
z_1+z_1' \cos x+z_2'\sin x ,
z_2-z_1' \sin x+z_2'\cos x\big)
\]
and the inverse of a group element is
\[
(x,z_1,z_2)^{-1} = \big({-}x , -z_1\cos x+z_2\sin x,
-z_1\sin x-z_2\cos x\big) .
\]

The existence of lattices in $\ISO(2)$ is highly
restrictive~\cite{Hull2005,Hull2019}. They require angles
$x_0\in\R^\times$ for which $2\cos x_0\in\Z$, i.e., $\cos
x_0=0,\pm\frac12,\pm1$. For our purposes, there are three
inequivalent choices of non-trivial elliptic monodromies,
which are characterized by matrices ${\tt M}\in\sfSL(2,\Z)$ of finite
order in the cyclic subgroups $\Z_2$, $\Z_4$ and $\Z_6$. They lead to three
types of \emph{Euclidean solvmanifolds}, according to the order of the monodromy in the fibre of the
associated Mostow bundle. A crucial distinction from the case of the
Heisenberg nilmanifold is that here the Mostow fibrations are not
principal $\torus^2$-bundles.
We consider each of the three cases
separately in turn.

\subsubsection*{Euclidean $\boldsymbol{\Z_2}$-solvmanifolds}

We observe that for
\[
x_0 = m \pi
\]
with $m\in\Z^\times$, the matrix $\varphi_{x_0}$ is integer-valued, and so
may be set equal to $\tt M$ with $\Sigma=\unit_2$. For~$m$ even, the
monodromy ${\tt M}=\unit_2$ is trivial and the corresponding
Mostow bundle is simply the three-torus $\torus^3$, which we have
already considered in Section~\ref{sec:toruscrossedprod}. For~$m$ odd,
which we now assume,
the nontrivial monodromy matrix is given by
\begin{gather}\label{eq:mon1}
{\tt M}^{\text{\tiny(1)}} = \varphi^{\phantom{\dag}}_{m \pi} = -\unit_2 .
\end{gather}

The Euclidean
$\Z_2$-solvmanimanifold $\torus_{\ISO_m ^{\text{\tiny(1)}} (2;\Z)}$ is the compact
space obtained as the quotient of $\ISO(2)$ with
respect to the lattice given by the discrete Euclidean group
\[
\ISO ^{\text{\tiny(1)}}_m(2;\Z) := \big\{ (m \pi \alpha,\gamma_1,\gamma_2)\in\ISO(2) \,
\big| \, \alpha\in\Z , \, \gamma=(\gamma_1,\gamma_2)\in\Z^2 \big\} .
\]
The quotient is taken by the left action of
$\ISO ^{\text{\tiny(1)}}_m(2;\Z)$, which leads
to the local coordinate identifications under the action of the
generators of $\ISO ^{\text{\tiny(1)}}_m(2;\Z)$ given by
\begin{gather*}
(x,z_1,z_2) \longmapsto (x+m\,\pi,-z_1,-z_2) , \nonumber \\
(x,z_1,z_2) \longmapsto (x,z_1+1,z_2) , \nonumber \\
(x,z_1,z_2) \longmapsto (x,z_1,z_2+1) .
\end{gather*}
The relation matrix $
{\tt A}^{\text{\tiny(1)}} = {\tt M}^{\text{\tiny(1)}}-\unit_2 = -2\unit_2$ has maximal rank $r=2$ with elementary divisors
$m_1=m_2=2$, and
the first homology group of
$\torus_{\ISO ^{\text{\tiny(1)}}_m(2;\Z)}$ can thus be presented as the $\Z$-module
\[
\rmH_1\big(\torus_{\ISO ^{\text{\tiny(1)}}_m(2;\Z)},\Z\big) \simeq
\Z\oplus(\Z_2\oplus\Z_2) .
\]
It follows that the $\Z_2$-solvmanifolds do not possess any classical
T-duals. By symmetry the two possible noncommutative torus bundles
corresponding to the two torsion generators $e_{z_1}$ and $e_{z_2}$ of
order two are the same.

Consider the one-parameter subgroup $\real_{(0,1)}$ of $\ISO(2)$
given by
\[
\real_{(0,1)} = \big\{(0,0,\zeta) \in \ISO(2)\big\} ,
\]
which acts on the Euclidean group $\ISO(2)$ by
right multiplication
\begin{gather*}
(x,z_1,z_2) (0,0,\zeta) = (x,z_1+\zeta \sin x , z_2 + \zeta \cos x) .
\end{gather*}
In this case $\F_*$ is the subset of $x\in\R$ where $\cos x=0$, so that
\begin{gather}\label{eq:F*ell}
\F_* = \tfrac\pi2 (2 \Z+1) .
\end{gather}
Applying Theorem~\ref{thm:NCtorusbundle3d}, it follows that the $C^*$-algebraic T-dual of the
Euclidean $\Z_2$-solvmanifold is Morita equivalent to the
$C(\torus)$-algebra
\begin{gather}\label{eq:algbundleell1}
C\big(\torus_{\ISO^{\text{\tiny($1$)}}_m(2;\Z)}\big)\rtimes_\rt\real_{(0,1)}
\ \mbf\sim_{\text{\tiny M}} \ \mcoprod
\torus^2_{\theta^{\text{\tiny(1)}}_m(x)} ,
\end{gather}
where
\begin{gather}\label{eq:thetam1}
\theta_m^{\text{\tiny(1)}}(x) =
 \begin{cases}
 0 & \mbox{for} \ x \in \frac1{m} \big(\Z+\frac12\big) , \\
 \tan(m \pi x)
 & \mbox{for} \ x\in\R\setminus \frac1{m} \big(\Z+\frac12\big) . \end{cases}
\end{gather}

Recalling that the integer $m$ is odd here, it is easily seen that the noncommutativity para\-me\-ter~\eqref{eq:thetam1} is invariant under changing coset representative~$x\in\R/\Z$, which is consistent with its~$\sfSL(2,\Z)$ orbit under
the corresponding monodromy matrix~\eqref{eq:mon1}:
\[
\theta_m^{\text{\tiny(1)}}(x+1) = \theta_m^{\text{\tiny(1)}}(x) =
{\tt M}^{\text{\tiny(1)}}\big[\theta_m^{\text{\tiny(1)}}(x)\big],
\]
so that fibrewise $\A_{\theta^{\text{\tiny(1)}}_m(x+1)} =
\A_{\theta^{\text{\tiny(1)}}_m(x)}$. Thus the non-trivial
monodromy in the automorphism group of the $\torus^2$ fibres of the
twisted torus $\torus_{\ISO_m^{\text{\tiny(1)}}(2;\Z)}$ becomes a
trivial identity action on its T-dual $C^*$-algebra bundle.
This is analogous to what we found in the parabolic case. Moreover,
by Proposition~\ref{prop:starproduct3d} the $C^*$-algebra
bundle~\eqref{eq:algbundleell1} can be described as a deformation of the algebra of functions
$C\big(\torus\times\torus^2\big)$ on a trivial $\torus^2$-bundle over the
circle via a star-product
$f\star_{\theta_m^{\text{\tiny(1)}}}g$ given by~\eqref{eq:parabolicstar}. In marked contrast to the parabolic case, however, the algebra of
functions $C\big(\torus_{\ISO_m^{\text{\tiny(1)}}(2;\Z)}\big)$ is not itself
a \smash{$C\big(\torus^2\big)$}-algebra, and the
noncommutative torus bundle cannot be identified with the group
$C^*$-algebra of the integer Euclidean group
$\ISO_m^{\text{\tiny(1)}}(2;\Z)$.

\subsubsection*{Euclidean $\boldsymbol{\Z_4}$-solvmanifolds}

We next observe that for
\[
x_0 = \frac{m \pi}2
\]
with $m$ an odd integer, the matrix $\varphi_{x_0}$ is again integer-valued, and so
may be set equal to $\tt M$ with $\Sigma=\unit_2$:
\begin{gather}\label{eq:mon2}
{\tt M}^{\text{\tiny(2)}} = \varphi^{\phantom{\dag}}_{\frac{m \pi}2} = \pm\left( \begin{matrix}
0 & 1 \\ -1 & 0
\end{matrix} \right)  .
\end{gather}
Without loss of generality, we shall fix the positive sign by assuming that $m\in
4 \mathbb{Z}+1$ (the choice of negative sign for $m\in4 \Z+3$ simply
corresponds to a reflection $(z_1,z_2)\mapsto(-z_1,-z_2)$ below).

The Euclidean
$\Z_4$-solvmanifold $\torus_{\ISO ^{\text{\tiny(2)}}_m(2;\Z)}$ is the compact
space obtained as the quotient of $\ISO(2)$ with
\[
\ISO ^{\text{\tiny(2)}}_m(2;\Z) := \big\{ \big(\tfrac{m \pi}2 \alpha,\gamma_1,\gamma_2\big)\in\ISO(2) \,
\big| \, \alpha\in\Z , \, \gamma=(\gamma_1,\gamma_2)\in\Z^2 \big\} .
\]
The quotient is taken by the left action of
$\ISO ^{\text{\tiny(2)}}_m(2;\Z)$, which leads
to the local coordinate identifications under the action of the
generators of $\ISO ^{\text{\tiny(2)}}_m(2;\Z)$ given by
\begin{gather*}
(x,z_1,z_2)\longmapsto \big(x+\tfrac{m \pi}2,z_2,-z_1\big) , \nonumber \\
(x,z_1,z_2)\longmapsto (x,z_1+1,z_2) , \nonumber \\
(x,z_1,z_2)\longmapsto (x,z_1,z_2+1) .
\end{gather*}
The relation matrix $
{\tt A}^{\text{\tiny(2)}} = {\tt M}^{\text{\tiny(2)}}-\unit_2 $ has maximal rank $r=2$ with elementary divisors $m_1=1$ and
$m_2=2$, and the first homology group of
$\torus_{\ISO ^{\text{\tiny(2)}}_m(2;\Z)}$ can thus be presented as the $\Z$-module
\[
\rmH_1\big(\torus_{\ISO ^{\text{\tiny(2)}}_m(2;\Z)},\Z\big) \simeq \Z\oplus\Z_2 ,
\]
where the torsion
generator of order two is $e_{z_2}=e_{z_1}$. Thus again there is no
classical T-dual, as is well-known in
this case (see, e.g.,~\cite{Hull2019}).

Proceeding as in the previous
case, the subset $\F_*\subset\R$ is again given by \eqref{eq:F*ell} and we arrive at the Morita equivalence
\[
C\big(\torus_{\ISO^{\text{\tiny($2$)}}_m(2;\Z)}\big)\rtimes_\rt\real_{(0,1)}
\ \mbf\sim_{\text{\tiny M}} \ \mcoprod
\torus^2_{\theta^{\text{\tiny(2)}}_m(x)} ,
\]
where
\begin{gather}\label{eq:thetam2}
\theta_m^{\text{\tiny(2)}}(x) = \begin{cases}
 0 &\mbox{for} \ x \in \frac1{m} (2 \Z+1) , \\
 \tan\big(\frac{m \pi}2 x\big)
 & \mbox{for} \ x\in\R\setminus \frac1{m} (2 \Z+1) . \end{cases}
\end{gather}
The corresponding star-product
$f\star_{\theta_m^{\text{\tiny(2)}}}g$ on $C\big(\torus\times\torus^2\big)$
is given by~\eqref{eq:parabolicstar} and was originally written down
in~\cite{Hull2019}, where it was also pointed out that the
noncommutative torus fibration is not isomorphic to the group
$C^*$-algebra of $\ISO_m^{\text{\tiny(2)}}(2;\Z)$.

Compared to the previous cases, the new feature here is that a change of coset
representative $x\in\R/\Z$ generally has a
non-trivial action on the fibers of the $C^*$-bundle $\coprod_{x\in\R/\Z}
\A_{\theta_m^{\text{\tiny(2)}}(x)}$ according to
\eqref{eq:thetaMorita3d}. We check this explicitly here: Recalling that $m\in 4\,\Z+1$ in
 this case, the fibre over any $x\in\frac1m \Z$ is preserved
 as then the form of~\eqref{eq:thetam2} implies
\[
\theta_m^{\text{\tiny(2)}}(x+1) =
 0=\theta_m^{\text{\tiny(2)}}(x)
\qquad \mbox{for} \quad x\in\tfrac1m \Z .
\]
However,
 over any $x\in\R\setminus\frac1m \Z$ the noncommutativity parameter~\eqref{eq:thetam2} changes non-trivally according to the
 trigonometric identity $\tan\big(x+\frac{m \pi}2\big)=-\cot x$, but in a way which is
 consistent with its $\sfSL(2,\Z)$ orbit under the corresponding
 monodromy matrix~\eqref{eq:mon2}:
\[
\theta_m^{\text{\tiny(2)}}(x+1) = -\frac1{\theta_m^{\text{\tiny(2)}}(x)} =
{\tt M}^{\text{\tiny(2)}}\big[\theta_m^{\text{\tiny(2)}}(x)\big]
\qquad \mbox{for} \quad x\in\R\setminus\tfrac1m \Z .
\]
By Theorem~\ref{thm:NCtorusbundle3d}, the fibre noncommutative tori are Morita
equivalent, and so are isomorphic in the
category~${\CCK\hspace{-1mm}\CCK}$. A physical picture of the action
of the monodromy on the K-theory group of this $C^*$-bundle in terms
of fiber D-branes is given in~\cite{Hull2019}.

\subsubsection*{Euclidean $\boldsymbol{\Z_6}$-solvmanifolds}

Finally, we observe that for
\[
x_0 = \frac{m \pi}3
\]
with $m\notin 3 \Z$, the matrix $\varphi_{x_0}$ takes four possible forms
\[
\varphi^{\phantom{\dag}}_{\frac{m\,\pi}3} =
\frac\epsilon2\left(\begin{matrix} \epsilon' & \sqrt3 \\ -\sqrt3 &
 \epsilon' \end{matrix}\right)
\]
with independent signs $\epsilon,\epsilon'=\pm 1$. It
conjugates to an
integer matrix via the element $\Sigma\in\sfSL(2,\real)$ given by
\begin{gather}\label{eq:Sigmaell}
\Sigma = \sqrt{\tfrac2{\sqrt3}} \left(\begin{matrix} 1 & \frac12
 \\ 0 & \frac{\sqrt3}2 \end{matrix}\right) .
\end{gather}
The precise form of $\varphi^{\phantom{\dag}}_{x_0}$ and hence of the
corresponding monodromy matrix $\tt M$ depends on the congruence
class of the integer~$m$ in~$\Z_3$:
\[
\Sigma^{-1} \, \varphi^{\phantom{\dag}}_{\frac{m \pi}3}
\Sigma = \begin{cases}
\pm \left(\begin{smallmatrix} \hphantom{-}1 & 1 \\ -1 & 0\end{smallmatrix}\right) , & m\in 3 \Z+1 ,
\\
\pm \left(\begin{smallmatrix} \hphantom{-}0 & \hphantom{-}1 \\ -1 & -1\end{smallmatrix}\right), & m\in 3 \Z+2 .
\end{cases}
\]
Without loss of generality, we shall fix $m\in 6\,\Z+2$ and hence take
\begin{gather}\label{eq:mon3}
{\tt M}^{\text{\tiny(3)}} = \left( \begin{matrix}
0&1\\ -1&-1
\end{matrix}\right)
\end{gather}
in the following (with the other three possibilities obtained simply
by reflection $(z_1,z_2)\mapsto(-z_1,-z_2)$ and interchange
$(z_1,z_2)\mapsto(z_2,z_1)$ below).

The Euclidean
$\Z_6$-solvmanifold $\torus_{\ISO ^{\text{\tiny(3)}}_m(2;\Z)}$ is the compact
space obtained as the quotient of $\ISO(2)$ with
respect to the lattice given by the discrete subgroup
\begin{gather*} \ISO ^{\text{\tiny(3)}}_m(2;\Z) := \left\{
\left(\tfrac{m\pi}3\alpha,
\sqrt{\tfrac2{\sqrt3}} \gamma_1+\sqrt{\tfrac1{2 \sqrt3}} \gamma_2 ,
\sqrt{\tfrac{\sqrt3}2} \gamma_2\right)\!\in\ISO(2) \,
\Big| \, \alpha\in\Z,
 \gamma=(\gamma_1,\gamma_2)\in\Z^2 \right\} .
\end{gather*}
The quotient is taken by the left action of
$\ISO ^{\text{\tiny(3)}}_m(2;\Z)$, which leads
to the local coordinate identifications under the action of the
generators of $\ISO ^{\text{\tiny(3)}}_m(2;\Z)$ given by
\begin{gather*}
(x,z_1,z_2) \longmapsto
 \left(x+\tfrac{m \pi}3 , \tfrac12 \big(\sqrt3 z_2-z_1\big) ,
 -\tfrac12 \big(\sqrt3 z_1+z_2\big)\right) , \\
(x,z_1,z_2) \longmapsto \left(x ,z_1+\sqrt{\tfrac2{\sqrt3}} , z_2\right) , \\
(x,z_1,z_2) \longmapsto \left(x , z_1+\sqrt{\tfrac1{2 \sqrt3}},
 z_2+\sqrt{\tfrac{\sqrt3}2} \right) .
\end{gather*}
The relation matrix ${\tt A}^{\text{\tiny(3)}} = {\tt M}^{\text{\tiny(3)}}-\unit_2$ has maximal rank $r=2$ with elementary divisors $m_1=1$ and
$m_2=3$, and the first homology group can thus be presented as the $\Z$-module
\[
\rmH_1\big(\torus_{\ISO ^{\text{\tiny(3)}}_m(2;\Z)},\Z\big) \simeq \Z\oplus\Z_3 ,
\]
where the torsion generator of order three is $e_{z_2}=e_{z_1}$. Here too there is no
classical T-dual.

Applying Theorem~\ref{thm:NCtorusbundle3d} in this case, with the
non-trivial period matrix $\Sigma$ from~\eqref{eq:Sigmaell} and the
subset $\F_*\subset\R$ again given by~\eqref{eq:F*ell}, gives
 the Morita equivalence
\[
C\big(\torus_{\ISO^{\text{\tiny($3$)}}_m(2;\Z)}\big)\rtimes_\rt\real_{(0,1)}
\ \mbf\sim_{\text{\tiny M}} \ \mcoprod
\torus^2_{\theta^{\text{\tiny(3)}}_m(x)} ,
\]
where
\begin{gather}\label{eq:thetam3}
\theta_m^{\text{\tiny(3)}}(x) = \begin{cases}
 0 & \mbox{for} \ x \in
 \frac1{m} \big(3 \Z+\frac32\big) , \\
 -\frac12 + \frac{\sqrt3}2\tan\big(\frac{m \pi}3 x\big)
 & \mbox{for} \ \ x\in\R\setminus \frac1{m} \big(3 \Z+\frac32\big) . \end{cases}
\end{gather}
As in the previous case, changing coset representative
$x\in\R/\Z$ generally has a non-trivial action on the fibers of the
$C^*$-bundle over $\torus$: Recalling that $m\in 6\,\Z+2$ here, the
fiber over any $x\in\frac1m\,\big(3\,\Z+\frac32\big)$ is preserved by
the form of \eqref{eq:thetam3}:
\[
\theta_m^{\text{\tiny(3)}}(x+1) =
 0=\theta_m^{\text{\tiny(3)}}(x)
\qquad \mbox{for} \quad x\in\tfrac1m\big(3\Z +\tfrac32\big) .
\]
On the other hand, over any $x\in\R\setminus\frac1m\big(3\Z+\frac32\big)$
the noncommutativity parameter~\eqref{eq:thetam3} changes according to
the trigonometric identity
\[
\tan\big(x+\tfrac{m \pi}3\big) = \frac{\tan x-\sqrt3}{1+\sqrt3 \tan x} ,
\]
but consistently with its $\sfSL(2,\Z)$ orbit under the corresponding
monodromy matrix \eqref{eq:mon3}:
\[
\theta_m^{\text{\tiny(3)}}(x+1) =
-\frac1{\theta_m^{\text{\tiny(3)}}(x)+1} = {\tt
 M}^{\text{\tiny(3)}}\big[ \theta_m^{\text{\tiny(3)}}(x)\big] \qquad \mbox{for} \quad x\in\R\setminus\tfrac1m \big(3 \Z +\tfrac32\big).
\]
By Theorem~\ref{thm:NCtorusbundle3d} the corresponding fibre noncommutative tori are Morita
equivalent,
and so isomorphic in the category ${\CCK\hspace{-1mm}\CCK}$.

\subsection{Hyperbolic torus bundles}\label{sec:hyperbolicbundles}

Our last application of
Theorem~\ref{thm:NCtorusbundle3d} is to the final class of
three-dimensional solvmanifolds, which are
hyperbolic analogues of the Euclidean solvmanifolds based on the
Poincar\'e group, and have also not been
previously studied in the present context.

The Poincar\'e group $\ISO(1,1)$ in two dimensions is the
three-dimensional almost abelian solvable Lie group which can be presented as
\[
\ISO(1,1) = \real^2\rtimes_\varphi\real \qquad \mbox{with} \quad
\varphi_x = \bigg(\begin{matrix} \cosh x & \sinh x \\ \sinh x &
 \cosh x \end{matrix}\bigg) .
\]
Since $|\Tr\varphi_x|>2$ for all $x\in\real^\times$, the hyperbolic rotation
$\varphi_x$ parameterizes a hyperbolic conjugacy class of
$\sfSL(2,\real)$. Again we denote coordinates on the group
manifold of $\sfN=\real^2$
by $z=(z_1,z_2)$, so that the group multiplication on
$\ISO(1,1)$ is given by
\[
(x,z_1,z_2) (x',z_1',z_2') = \big(x+x' , z_1+z_1'\cosh
x+z_2'\sinh x ,
z_2+z_1'\sinh x+z_2'\cosh x\big)
\]
and the inverse of a group element is
\[
(x,z_1,z_2)^{-1} = \big({-}x , -z_1\cosh x+z_2\sinh x , z_1\sinh
x-z_2\cosh x\big) .
\]

Clearly there is no $x_0\in\real^\times$ for which $\varphi_{x_0}$ is
an integer matrix in this case. However, there is an infinite family of discrete points
\begin{gather}\label{eq:hypx0}
x_0 = \log\left(\frac m2 \pm \sqrt{\left(\frac m2\right)^2-1}\right)
\end{gather}
labelled by an integer $m>2$ with $2\cosh x_0=m$, at which
the matrix $\varphi_{x_0}$ takes the form
\[
\varphi_{x_0} = \frac12 \left(\begin{matrix} m & \pm \sqrt{m^2-4}
 \\ \pm \sqrt{m^2-4} &
 m \end{matrix}\right) .
\]
This matrix has eigenvalues $\big(\lambda,\lambda^{-1}\big)$ where
\[
\lambda^{\pm1} = \frac12 \left(m\pm\sqrt{m^2-4}\right) = \e^{x_0} ,
\]
and it conjugates to the integer matrix
\begin{gather}\label{eq:monhyp}
{\tt M}_m = \Sigma^{-1} \varphi_{x_0} \Sigma =
\left(\begin{matrix} m & 1 \\ -1 & 0 \end{matrix}\right)
\end{gather}
by the element $\Sigma\in\sfSL(2,\real)$ given by
\begin{gather}\label{eq:hypSigma}
\Sigma = \begin{pmatrix} \dfrac{\lambda^{\pm1}}{2 \big(\lambda-\lambda^{-1}\big)} - \lambda^{\mp1} & \dfrac\lambda{2 \big(\lambda^2-1\big)}-1 \vspace{1mm}\\ \dfrac{\lambda^{\pm1}}{2 \big(\lambda-\lambda^{-1}\big)} + \lambda^{\mp1} &
 \dfrac\lambda{2 \big(\lambda^2-1\big)} + 1 \end{pmatrix} .
\end{gather}
For definiteness, we choose the positive square root in~\eqref{eq:hypx0} (with the choice of negative square root obtained
from the interchange $(z_1,z_2)\mapsto(z_2,z_1)$ below).

The Poincar\'e
solvmanifold $\torus_{\ISO_m(1,1;\Z)}$ is the compact space
obtained as the quotient of $\ISO(1,1)$ by the lattice
\[
\ISO_m(1,1;\Z) := \big\{(x_0\,\alpha,\Sigma\cdot\gamma)\in\ISO(1,1) \, \big| \,
\alpha\in\Z , \gamma=(\gamma_1,\gamma_2)\in\Z^2\big\} .
\]
The quotient by the left action of this lattice leads to the local
coordinate identifications under the action of the generators given by
\begin{gather*}
(x,z_1,z_2) \longmapsto
\left(x+\cosh^{-1}\big(\tfrac m2\big),\tfrac12\big(mz_1+\sqrt{m^2-4}z_2\big),\tfrac12\big(\sqrt{m^2-4}z_1+mz_2\big) \right) , \\
(x,z_1,z_2) \longmapsto \left(x , z_1+\tfrac{\lambda^2}{2 (\lambda^2-1)}-\lambda^{-1} , z_2+\tfrac{\lambda^2}{2(\lambda^2-1)}+\lambda^{-1}\right) , \nonumber
 \\
(x,z_1,z_2) \longmapsto \left(x, z_1+\tfrac\lambda{2(\lambda^2-1)}-1, z_2+\tfrac\lambda{2(\lambda^2-1)}+1\right).
\end{gather*}
The relation matrix ${\tt A}_m = {\tt M}_m-\unit_2$ has maximal rank $r=2$ with elementary divisors $m_1=1$ and
$m_2=m-2$, so the first homology group of the Poincar\'e solvmanifold may be
presented as the $\Z$-module
\[
{\rm H}_1\big(\torus_{\ISO_m(1,1;\Z)},\Z\big)\simeq\Z\oplus\Z_{m-2} ,
\]
with torsion generator $e_{z_2}=-e_{z_1}$ of order $m-2$.

Consider the one-parameter subgroup
\[
\real_{(0,1)} = \big\{(0,0,\zeta)\in\ISO(1,1)\big\}
\]
acting on the Poincar\'e group $\ISO(1,1)$ by right multiplication
\[
(x,z_1,z_2) (0,0,\zeta) = (x,z_1+\zeta\sinh x,z_2+\zeta\cosh x) .
\]
In this case the set $\F_*$ consists of a distinguished point $x_*\in\R$ on the base of the
$C^*$-algebra bundle given by
\[
x_*=\tanh^{-1}\left(\frac{2 \big(\lambda^2-1\big)-\lambda^3}{2 \big(\lambda^2-1\big)+\lambda^3}\right).
\]
Applying
Theorem~\ref{thm:NCtorusbundle3d}, with the periods
\eqref{eq:hypx0} and \eqref{eq:hypSigma}, then identifies the
Morita equivalence
\[
C\big(\torus_{\ISO_m(1,1;\Z)}\big)\rtimes_\rt\real_{(0,1)} \ \mbf\sim_{\text{\tiny M}} \
\mcoprod \torus_{\theta_m(x)}^2
\]
of $C(\torus)$-algebras, where
\[
\theta_m(x_*)=0
\]
and
\begin{gather}\label{eq:thetamxhyp}
\theta_m(x) =
\frac{2 \lambda \big(\lambda^2-1\big)-\lambda^2+\lambda
 \big(2\big(\lambda^2-1\big)+\lambda\big)\tanh\big(\cosh^{-1}\big(\frac
 m2\big) x\big)}{\lambda^3-2 \big(\lambda^2-1\big) - \big(\lambda^3+2 \big(\lambda^2-1\big)\big)\tanh\big(\cosh^{-1}(\frac
 m2) x\big)}
\end{gather}
for $x\neq x_*$.

This noncommutative torus bundle, along with its topological T-duality with
$C\big(\torus_{\ISO_m(1,1;\Z)}\big)$, has formally the same properties
as the $C(\torus)$-algebras described in
Section~\ref{sec:ellipticbundles}, so
we refrain from repeating the details here. We only mention the
fibre monodromy behaviour anticipated from~\eqref{eq:thetaMorita3d}: Using the hyperbolic identity
\[
\tanh\big(x+\cosh^{-1}\big(\tfrac m2\big)\big) = \frac{m\tanh x+\sqrt{m^2-4}}{m+\sqrt{m^2-4} \tanh x} ,
\]
we find
\[
\theta_m(x_*+1)=0=\theta_m(x_*)
\]
and
we see explicitly here that the noncommutativity parameter \eqref{eq:thetamxhyp}
changes consistently with its $\sfSL(2,\Z)$ orbit under the monodromy
matrix \eqref{eq:monhyp}:
\[
\theta_m(x+1) = -\frac1{\theta_m(x)}-m = {\tt
 M}_m\big[\theta_m(x)\big] \qquad \mbox{for} \quad x\neq x_* .
\]
Thus the fiber noncommutative tori are Morita equivalent by
Theorem~\ref{thm:NCtorusbundle3d},
and so are isomorphic $C^*$-algebras in the category
$\CCK\hspace{-1mm}\CCK$.

\subsection*{Acknowledgments}

We thank Ryszard Nest and Erik Plauschinn for helpful
discussions. We thank the anonymous referees for their detailed suggestions. This research was supported by funds from Universit\`a del
Piemonte Orientale (UPO). P.A.\ acknowledges partial support from INFN,
CSN4, and Iniziativa Specifica GSS. P.A.\ is affiliated to INdAM-GNFM. R.J.S.\ acknowledges a Visiting
Professorship through UPO Internationalization Funds. R.J.S.\ also
acknowledges the Arnold--Regge Centre for the visit, and INFN. The work of R.J.S.\ was supported in
part by the Consolidated Grant ST/P000363/1 from the UK Science and
Technology Facilities Council.

\LastPageEnding

\end{document}